    \documentclass[11pt,a4paper,nosumlimits,reqno,dvipsnames]{amsart}
 \usepackage{amsaddr}
 \usepackage[utf8]{inputenc} 

\usepackage{amsthm}
\usepackage{amsmath}
\usepackage{mathrsfs}
\usepackage{xcolor}
   \definecolor{azulf}{HTML}{0092D2}
 \usepackage[bookmarks=false]{hyperref}
    \hypersetup{
         colorlinks   = true,
         citecolor    =azulf!70!black
    }
    \hypersetup{linkcolor=azulf!70!black}
 
\usepackage{latexsym,amssymb,enumerate,bbm,amscd,graphicx,color}
\usepackage[british]{babel}
\usepackage{tikz}
\usepackage{tikz-cd}
\usepackage{caption}

 \newtheorem{thm}{Theorem}[section]
 
 \newtheorem{lem}[thm]{Lemma}
 
 \theoremstyle{definition}
 \newtheorem{defn}[thm]{Definition}
 \theoremstyle{remark}
 \newtheorem{rem}[thm]{Remark}
 \newtheorem*{ex}{Example}
 \numberwithin{equation}{section}
  \newtheorem*{preparation*}{\textbf{Notation of the theorem}}

    \usepackage{enumitem}
    
    \newcommand{\etalchar}[1]{$^{#1}$}

    \newcommand{\cV}{\mathscr{V}}
    \newcommand{\Aut}{\mathrm{Aut}}
    
    \newcommand{\Autc}{\Aut_{\mathrm{c}}}
    \renewcommand{\and}{\mbox{and}}

    \newcommand{\Tr}{\mathrm{Tr}}
    
    \newcommand{\mtr}[1]{\mathrm{#1}}
    \newcommand{\mtf}[1]{\mathfrak{#1}}

    \newcommand{\A}{\mathcal{A}}
    \newcommand{\dif}[1]{\mathrm{d}#1}

    \newcommand{\re}{\mathbb{R}}

    \newcommand{\dervpar}[2]{\frac{\partial #1}{\partial #2}}
    
    \newcommand{\dervfunc}[2]{\frac{\delta #1}{\delta #2}}

    \newcommand{\si}{\sigma}
    \newcommand{\G}{\mathcal{G}}
    
    \newcommand{\B}{\mathcal{B}}
    \newcommand{\C}{\mathbb{C}}

    \newcommand{\ii}{\mathrm{i}}
    \newcommand{\ee}{\mathrm{e}}
     
    \newcommand{\inv}{^{-1}} 
    \newcommand{\mtc}[1]{\mathcal{#1}} 
    \newcommand{\Z}{\mathbb{Z}}
    \newcommand{\N}{\mathbb{N}}

    \newcommand{\where}{\mbox{where}\,\,}

    \newcommand{\mtb}[1]{\mathbb{#1}}

    \newcommand{\F}{\mathcal{F}}

    \newcommand{\R}{\mathcal R}

    \newcommand{\Df}{\mathbb{\mathcal{D}}}

    \newcommand{\Sym}{\mathfrak{S}}
     
    \newcommand{\hp}[1]{^{(#1)}}

    \renewcommand{\phi}{\varphi}

    \newcommand{\Grph}[1]{\textsf{Grph}_{#1}}

     \newcommand{\dvGrph}[1]{\Grph{#1}^{\amalg,\mathrm{cl}}}

    \newcommand{\suml}{\sum\limits}

    \newcommand{\fder}[2]{\frac{\delta #1}{\delta #2}}
    
    \newcommand{\bJ}{\bar J}
    \newcommand{\kthree}{\raisebox{-.2\height}{\includegraphics[height=2ex]{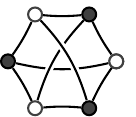}}}
    
    \newcommand{\kthreecik}{\raisebox{-.2\height}{\includegraphics[height=1.8ex]{graphs/3/Item6_K33.pdf}}}

    \newcommand{\GDmelon}{G\hp{2}_{\raisebox{-.33\height}{\includegraphics[height=2.2ex]{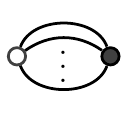}}}}
    \newcommand{\Gmelon}{G\hp{2}_{\raisebox{-.33\height}{\includegraphics[height=1.8ex]{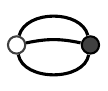}}}}

    \newcommand{\Gcc}{G\hp{4}_{\! \raisebox{-.2\height}{\includegraphics[height=1.8ex]{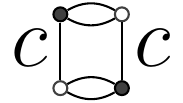}}}}

    \newcommand{\vi}{\raisebox{-.22\height}{\includegraphics[height=1.8ex]{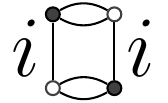}}}
     
     \newcommand{\vcv}{\raisebox{-.22\height}{\includegraphics[height=1.8ex]{graphs/3/Item4_Vcv.pdf}}}
    \newcommand{\vuno}{\raisebox{-.322\height}{\includegraphics[height=2.3ex]{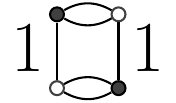}}}
    \newcommand{\vdos}{\raisebox{-.322\height}{\includegraphics[height=2.3ex]{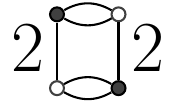}}}
    \newcommand{\vtres}{\raisebox{-.322\height}{\includegraphics[height=2.3ex]{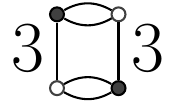}}}
    \newcommand{\vunoup}{\raisebox{-.25\height}{\includegraphics[height=2.3ex]{graphs/3/Item4_V1v.pdf}}}

    \newcommand{\vunito}{\raisebox{-.322\height}{\includegraphics[height=1.66ex]{graphs/3/Item4_V1v.pdf}}}

    \newcommand{\Gcmm}{G\hp{4}_{\raisebox{-.33\height}{\includegraphics[height=1.8ex]{graphs/3/Item2_Melon}}
    |\raisebox{-.33\height}{\includegraphics[height=1.8ex]{graphs/3/Item2_Melon}}}}
    
    \newcommand{\Gsmmm}{G\hp{6}_{\raisebox{-.33\height}{\includegraphics[height=1.8ex]{graphs/3/Item2_Melon}}
    |\raisebox{-.33\height}{\includegraphics[height=1.8ex]{graphs/3/Item2_Melon}}
    |\raisebox{-.33\height}{\includegraphics[height=1.8ex]{graphs/3/Item2_Melon}}}}

    \newcommand{\Gsmc}{G\hp{6}_{\raisebox{-.33\height}{\includegraphics[height=1.8ex]{graphs/3/Item2_Melon}}
    |\raisebox{-.34\height}{\includegraphics[height=1.8ex]{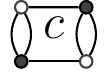}}|}}
     \newcommand{\Cubo}{ \raisebox{-.32\height}{\includegraphics[height=6ex]{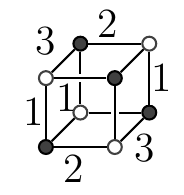}}}

    \newcommand{\Sint}{V}
    \newcommand{\fderJ}[1]{\fder{}{J{\phantom{a}}}_{\!\!\!\!\!{#1}}}
    \newcommand{\fderbJ}[1]{\fder{}{\bJ{\phantom{a}}}_{\!\!\!\!\!{#1}}}

    \newcommand{\jj}{[J,\bJ]}

    \newcommand{\mm}{\mathrm{m}}

    \newcommand{\xb}{\mathbf{x}}
    \newcommand{\yb}{\mathbf{y}}
    \newcommand{\zb}{\mathbf{z}}
    \newcommand{\Xb}{\mathbf{X}}
    
    \renewcommand{\sb}{\mathbf{s}}
    
    \newcommand{\uni}{\mtr{U}}
    \renewcommand{\sec}[1]{Sect. \ref{#1}}

    \usepackage{multicol,array}
    \makeatletter
    \def\moverlay{\mathpalette\mov@rlay}
    \def\mov@rlay#1#2{\leavevmode\vtop{%
       \baselineskip\z@skip \lineskiplimit-\maxdimen
       \ialign{\hfil$\m@th#1##$\hfil\cr#2\crcr}}}
    \newcommand{\charfusion}[3][\mathord]{
        #1{\ifx#1\mathop\vphantom{#2}\fi
            \mathpalette\mov@rlay{#2\cr#3}
          }
        \ifx#1\mathop\expandafter\displaylimits\fi}
    \makeatother

    \newcommand{\FM}{\mathrm{FM}}
    \newcommand{\hf}{\mathfrak{h}}
    \newcommand{\Calc}{\mathscr{C}(\hf)}

    \newcommand{\FMh}{\FM(\hf)}
    \newcommand{\cequal}{ \stackrel{\mbox{\scalebox{.6} {c}}}{\sim}}
    
    \usepackage{stackrel}
    \newcommand{\bcequal}{ \stackrel[{\,\!}^{\mbox{\scalebox{.6}{c}}}]{}{\sim}  }

    \usepackage{multirow}
    \newcommand{\logo}[4]{\raisebox{-.#4\height}{\includegraphics[height=#3 ex]{graphs/3/Logo#1_#2.pdf}}}

    \newcommand{\meloncik}{\raisebox{-.23\height}{\includegraphics[height=1.8ex]{graphs/3/Item2_Melon}}}
    
    \newcommand{\balita}{ \raisebox{.53\height}{\hspace{2pt}\tikz{ \fill[black] (0,0) circle (1.3pt);}}\,}
    
    \newcommand{\meloncito}{\logo{2}{Melon}{2}{13}}
    
    \newcommand{\J}{\mathbb{J}}

\newcommand{\itemB}{\item}
\newcommand{\itemW}{\item}
 \usepackage{rotating}
 \usepackage{pdflscape}

\calclayout

\begin{document}


%Graph calculus and the disconnected-boundary Schwinger-Dyson equations of quartic tensor field theories 

\title[Graph calculus and disconnected-$\partial$ SDE in quartic TFT]
 {Graph calculus and the  
 \\ disconnected-boundary Schwinger-Dyson
 \\ equations of quartic tensor field theories      
}

%----------Author 1
\author[C. I. P\'erez-S\'anchez]{Carlos I. P\'{e}rez-S\'{a}nchez}

\address{%
Mathematisches Institut der Westf\"alischen  % Einsteinstra\ss e 62,
% 48149 
 Wilhelms-Universit\"at \\
 Einsteinstra\ss e 62, 
 48149 M\"unster, Germany \\
 \& \\
Faculty of Physics, University of Warsaw$^{\star}$\\ 
 ul. Pasteura 5, 
 02-093 Warsaw, Poland\footnote{${}^\star$ Current affiliation.}
}

\email{cperez@fuw.edu.pl}

% \thanks{Thanks line}
%----------Author 2
% \author{A Second Author}
% \address{The address of\br
% the second author\br
% sitting somewhere\br
% in the world}
% \email{dont@know.who.knows}


\date{June 2019}
 
%%%--------------------------------------------------------------------

\begin{abstract}
 
    Tensor field theory (TFT) focuses on quantum field theory aspects of random tensor models, a quantum-gravity-motivated generalisation of random matrix models.
        The TFT correlation functions have been shown to
        be classified by graphs that describe the geometry of the boundary states, the so-called boundary graphs.
        These graphs can be disconnected, although the correlation functions are themselves connected. 
        In a recent work, the Schwinger-Dyson equations for an arbitrary albeit connected boundary were obtained. 
        Here, we introduce the multivariable graph calculus in order to derive the missing equations for all 
        correlation functions with disconnected boundary, thus completing the Schwinger-Dyson pyramid for quartic melonic (`pillow'-vertices) models in arbitrary rank. 
        We first study finite group actions that are parametrised by 
        graphs and build the graph calculus on a suitable quotient of the  
        monoid algebra $\A[G]$ corresponding to a certain function space $\A$ and  
        to the free monoid $G$ in finitely many graph variables;
        a derivative of an element of $\A[G]$ with respect to a graph yields its corresponding group action on $\A$.
        The present result and the graph calculus have three potential applications: the non-perturbative 
        large-$N$ limit of tensor field theories, the solvability of the theory
        by using methods that generalise the topological recursion to the TFT setting
        and the study of 
        `higher dimensional maps' via Tutte-like equations. 
        In fact, we also offer a term-by-term comparison between
        Tutte equations and the present Schwinger-Dyson equations.
    
\end{abstract}

%%% ----------------------------------------------------------------------
\maketitle
%%% ----------------------------------------------------------------------
\newpage
  \fontsize{11.2}{19}\selectfont  
     \tableofcontents
    \fontsize{11.1}{14.5}\selectfont 
\thispagestyle{plain}
  \section[sec]{Introduction and motivation}%\except{toc}{Introduction and motivation}}
    
  The quest for laws of physics near the Planck scale 
  leads some quantum gravitologist and quantum cosmologists to 
    replace the smooth space-time paradigm 
with new geometrical structures that are suitable for said energy scale. 
Those new structures include discretisation of space-time (e.g. causal dynamical triangulations \cite{CDTall}), 
the algebraisation of space-time (e.g. noncommutative geometry \cite{CC,MarcolliNCGcosmology}),
just to name some\footnote{See for instance \cite{MapQG} for a thorough classification.}.
Already the sole description of a space-time by a single mathematical object is expected to require
therefore novel geometrical ideas.
\par 
If one adopts a path-integral approach, the 
exploration of the quantum theory of space-time
requires additionally a multi-geometry description, to which `off-shell' geometries 
also contribute. Each of these geometries $\xi$ is weighted
via $ \exp(\ii S(\xi)/\hbar)\dif\xi$ by a `classical' action $S(\xi)$, 
bounded to 
resemble the Einstein-Hilbert action in the classical limit,
 in which $\xi$ starts looking like a Riemannian or Lorentzian manifold. 
\par 
Random tensors \cite{Ambjorn,GurauRyan,RTM_QG}  and related theories construct 
these kinds of measures $\dif \xi$, offer precisely a built-in 
description of both \textit{random} and \textit{discrete} geometry in arbitrary dimensions, and 
therefore constitute a tool to test models of background independent quantum gravity  (see e.g. \cite{Eichhorn:2018ylk}). 
Interest in the study of Euclidean quantum field theory (QFT) aspects of random tensors 
leads to \textit{tensor field theory} (TFT) \cite{4renorm,3Dbeta,wgauge,RivasseauVignes2019}, the matter of this article. 
Usually TFT fits in a `QFT + $ \epsilon$' framework, that is to say a
conservative modification of QFT, which one can 
pursue in the perturbative or non-perturbative approaches.
\par 

Non-perturbative TFT deals with the geometry of boundary states.
Single geometries in TFT are represented by certain decorated
graphs called \textit{coloured graphs}; this decoration is precisely the information that 
allows the construction of PL-manifolds from graphs.
Bulk-geometry graphs ---Feynman diagrams of TFT--- have a 
colour more than the graphs that triangulate the 
boundary geometries, which are called therefore 
\textit{boundary graphs} ($\partial$-graphs, for short). The two-fold purpose of this article is to define
in abstract way a calculus with coloured boundary-graph variables,
and shortly afterwards, to apply this construct to a 
particular problem in non-perturbative TFT.
\par 

The single-variable graph calculus has been used there as a toolkit for non-perturbative field tensor theory, leading to the full 
Ward-Takahashi identity \cite{fullward}.  
%    The fact that this calculus is graph-valued reflects the  
%  fact that correlation functions a
%  represent triangulations of boundary states.
% In fact, 
The single-variable graph calculus
allows to define each correlation function as a graph derivative of the 
free energy. 
Boundary graphs turn out to classify the correlation functions
of tensor field theories; these obey 
    analytic\footnote{We write `analytic' as opposed to algebraic SDE
      for expectation values. We conceive tensor field theory as a
      discretisation (therefore, 0-dimensional) of a $D$-dimensional
      quantum field theory. The $2k$-point correlation functions are
      thus functions of $\Z^{k\times D} \to \C$ which in the continuum
      limit pass to functions $\re^{k\times D}\to \C$, and render the SDE
       integro-differential equations.} Schwinger-Dyson Equation
    (SDE). \par 

Each and every analytic SDE for a \textit{connected }correlation function corresponding to a
\textit{connected}, but otherwise arbitrary, 
boundary graph was presented in \cite{SDE} in terms of 
a general formula that relates a 
given correlation function with its neighbourings (relative to the number of points)
 in terms of simple graph operations. In order to obtain these results,
    the single-variable graph calculus was a useful tool, which however 
   does not assist any longer in the 
  the derivation of SDE for connected correlation functions with disconnected boundary. This derivation
  requires a multivariable graph calculus, the variables being the
    different boundary components.
    
    We introduce graph-group actions as the
    basis of the multivariable calculus, and study
    their generating functionals. 
    Concretely, we obtain formulae for
    the graph derivative of products of functionals, i.e.
        the corresponding Leibniz rule that generalises
        \begin{equation} \label{eq:Leibniz}
     \partial^{\boldsymbol\alpha} (F\cdot G) =
    \sum_{0\leq \boldsymbol \gamma\leq \boldsymbol \alpha}
     \binom{ \boldsymbol \alpha}{ \boldsymbol \gamma } 
     \partial^{\boldsymbol \gamma}F \cdot 
     \partial^{\boldsymbol {\alpha}-\boldsymbol \gamma}G \quad 
    (\mbox{in multi-index notation})\,,
    \end{equation}
    when one replaces usual partial derivatives $\partial_{\mu}$ (say, with 
    $x^\mu$ a coordinate of $\re^n$, and $F,G$ real-valued smooth 
    functions there) by  derivatives $\partial_g=\partial/\partial g $ 
    with respect to a graph $g$.  The formula for graphs takes a different form,
    but reduces, as it should, to \eqref{eq:Leibniz}  when 
    one replaces functionals with functions and 
    simultaneously considers trivial group actions.\par
    We prove that this abstract structure underlies tensor models functionals and
    use it to find a general formula for the SDE of the quartic
    `pillow'-model, for the \textit{connected} correlation functions
    with arbitrary \textit{disconnected boundary} (abbr. disconnected-$\partial$). This is the
    missing piece that complements the connected-boundary SDE-pyramid obtained in
    \cite{SDE}. To have it complete is important for the analysis of
    the non-perturbative large-$N$ limit of tensor field
    theories. Moreover, although it is not clear which recursion
    should generalise the topological recursion \cite{EynardTopologicalRecursion}, it is clear that the 
    disconnected-$\partial$ correlation functions play an important role\footnote{For instance, higher dimensional analogue of the `pair of pants'
    being represented by a correlation function with three
    melonic boundary components.}.
%     ($\Gsmmm$ is the higher dimensional analogue of the `pair of
%     pants', which describes $2$-bordisms $S^1\to S^1 \amalg S^1$).
       
   \begin{figure}\vspace{-1cm}
    \includegraphics[width=8.5cm]{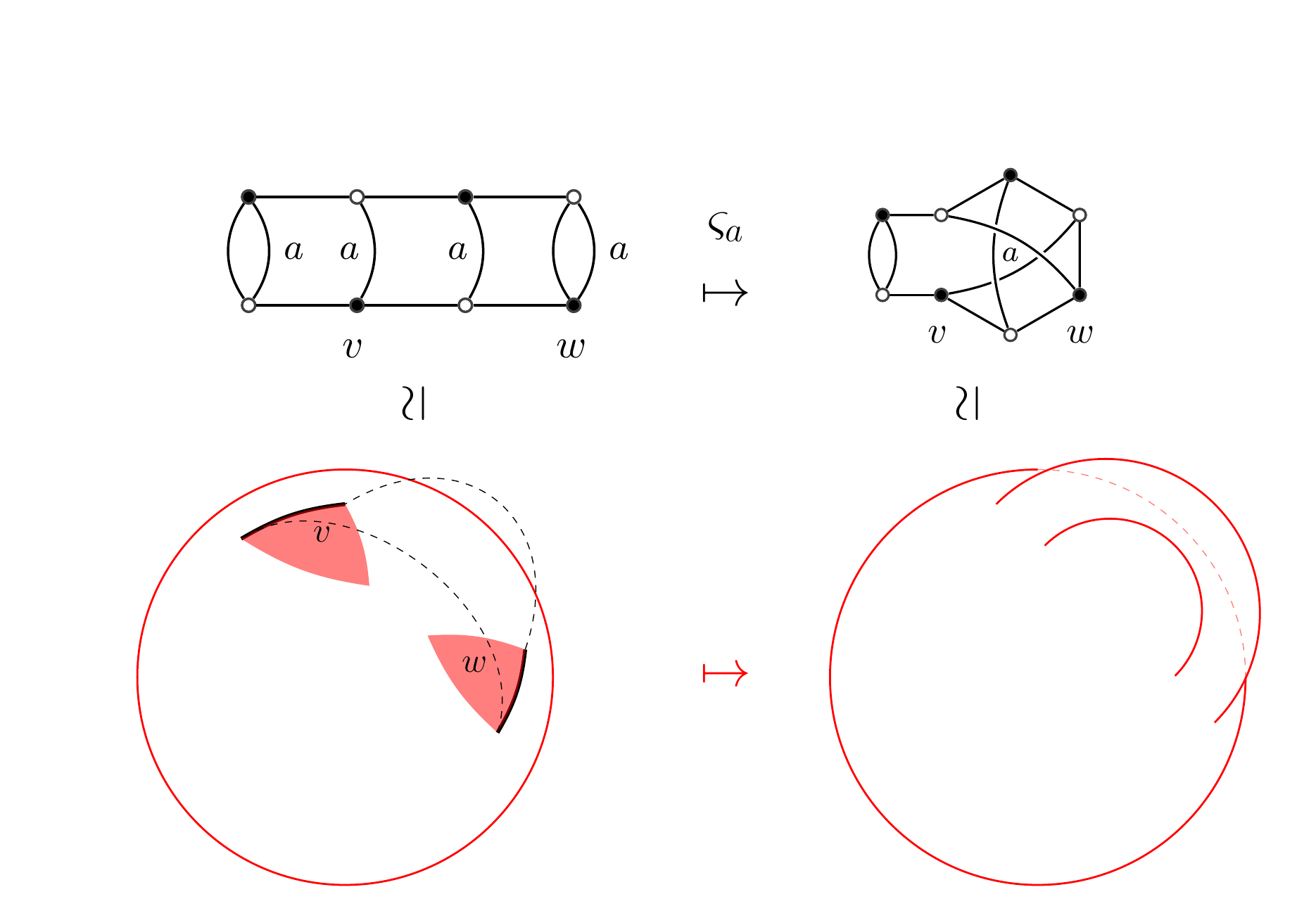}
    \caption{An example of $a$-coloured edge swap increasing the genus of a connected graph}
    \label{fig:swap_genus}
   \end{figure}
   Taking graph-derivatives can be understood 
   as the tensor model counterpart of `taking residues' 
    in matrix models (cf. \cite[Ch. 1-2]{CountingSurfaces})
    \[ \mathcal T_{l_1\ldots l_\kappa}= \displaystyle 
     \mathop{\mathrm{Res}}\displaystyle_{x_1\to \infty } \ldots  \mathrm{Res}_{x_\kappa\to \infty }
     \big [x^{l_1}_1 \cdots x^{l_\kappa}_\kappa W_\kappa(x_1,\ldots,x_\kappa)\big]\]
    in order to project the  free energy $W_\kappa$ (with $\kappa$ boundaries) of a matrix model 
    onto the generating function $ \mathcal T_{l_1\ldots l_\kappa}$ of random maps with 
    $\kappa$ marked faces with fixed perimeter lengths ($l_1,\ldots,l_\kappa$).  
    After the well-known equivalence \cite[Thm. 2.5.1]{CountingSurfaces} between
    Tutte equations \cite{tutte_1962,tutte_1963} for the enumeration of random maps
    and the loop equations for a suitable matrix model \cite{Migdal} (also summarised here in 
    in Sect. \ref{sec:Parallel}), 
    we can state the present result 
    as the basis to obtain Tutte-like equations    
    for the higher-dimensional
    analogue of the generating function of random maps $\mathcal T_{l_1\ldots l_\kappa}$.
    This improvement over \cite{SDE} ---where only the 
    equations analogous to a single-boundary correlator,
    say $\mathcal T_{l_1}$'s cousin, are presented---
    requires new developments: A straightforward generalisation of the proof
    given in \cite{SDE} is impossible due to the absence 
    of the multi-variable graph calculus.  
    Although it would be possible to state the main result 
    solely in a TFT context, the full notation would be a burden in the proof.
    We state some results in lighter notation and offer a shorter proof, at 
    expenses of introducing some new concepts. 
    The output of the main theorem is also a set of new graph operations 
    that extend those used to 
    describe the connected-$\partial$ SDE in \cite{SDE}. Also the  
    \textit{edge swap}\footnote{
    This operation has been studied 
    in the literature of Graph-Encoded Manifolds \cite{linsmulazzani} and by the 
    Crystallisation Theory \cite{PaolaRita}} $\varsigma_a$ (see Fig. \ref{fig:swap_genus}), as unary operation on a connected graph, is extended 
    to a binary operation implying two connected components; $\varsigma_a$ is then interpreted, 
    following \cite{cips}, as their connected sum.     
    \\
    
    This article is divided in an abstract part (\sec{sec:backround},
    whose main results are lemmas \ref{thm:Borel} and \ref{thm:Leibniz}), and
    a TFT-part.
    The next section explains 
     in detail the following 
    mnemonics: for  coloured graphs $g$ and $h$,
    \begin{equation} \label{eq:GroupAction}
    \dervpar{g}{h} = \delta(h,g) \cdot \mbox{group action of } G(g)  \,.
    \end{equation}
    Here $ \delta(g,h)=1$ if the graphs $g$ and $h$ are isomorphic 
    or otherwise $\delta(g,h)$ vanishes, and  $G(g)$ is a group
    determined by $g$.  In Section \ref{sec:back} we make the
    connection between the two first sections and
    model-independent TFT.
    In Section \ref{sec:teorema}, our quartic model is detailed and the results of the
    previous sections are applied to the main problem, namely to find the
    SDE for \textit{connected} correlation functions with arbitrary
    \textit{disconnected boundary} with arbitrary number of connected
    components.  Section \ref{sec:foursixSDE} gives
    explicitly some of the SDE for $4$, $6$-point functions for
    rank-$3$ theories. We highlight in Section \ref{sec:TutteSDE}, 
    preparing an important future task to apply our results, the term-to-term parallel between Tutte equations and those presented here, as well as
    analogies in the derivation of both sets of equations. 
    Concretely, we compare the new operations
    on the boundary 
    graphs of TFT (Table \ref{tab:MultiTM}) with their matrix models counterpart (Table \ref{tab:MultiMM}).
    The operations and terms presented in the SDE of \cite{SDE} 
    are only the counterpart of those matrix models SDE
    presented in Table \ref{tab:SingleMM}.
    The conclusions 
    and outlook are given in Section \ref{sec:Outlook}, discussing a potential application
    related to the higher-dimensional analogue of the topological recursion.
    The useful coefficients that encode
    the insertion of the $4$-point functions into the $2$-point 
    function and the $6$-point functions into the $4$-point functions
   are given in the Appendix \ref{sec:App}.

    \section{Graph calculus}  \label{sec:backround}
  
    In this section we explain what we mean by graph calculus.
    For naturality reasons\footnote{This is in line with other theories (as a matter of fact
    Topological Quantum Field Theories) that need the 
    empty $D$-manifold. Also, 
    it is technically advantageous and not the first time it is considered,
    see for instance \cite{krajewskireiko}.} we consider the empty graph  $\varnothing$ as 
    coloured and add it to the set of (possibly) disconnected, closed,
    regularly edge-$D$-coloured, vertex-bipartite graphs
    (`$D$-\textit{coloured graphs}') $\dvGrph{D}$, to form $\mathsf G_D=\{\varnothing\} \cup \dvGrph{D}$.  Henceforth, all graphs are coloured, but other
    types of graphs could be used for the next constructs.  

%       thus, in the context of $D$-rank tensor models, graphs are $D$-coloured. 

%     No reference to tensor models until.
%    \par 

    \subsection{Single variable graph calculus} 
      \thispagestyle{plain}
  
    We regard $\mathsf G_D$ with a monoidal structure, the product and the unit 
    being given by
    \[ c_1 c_2=c_1\amalg c_2\,\quad \and \quad c\amalg \varnothing = c
      =\varnothing \amalg c\,,%, \qquad\qquad \mbox{for all }.
      \]
      respectively, for all $c ,c_1,c_2\in
      \mathsf G_D$. We choose to remember the order of the factors, so this
    product is generally non-commutative, $c_1c_2\neq c_2c_1$.
     
    \begin{defn}[System of graph-group actions]
      For a finite collection $H\subset \mathsf{G}_D$, consider the
      following structures:
    \begin{itemize}
      \itemB for each connected graph $c\in H$:
     \begin{itemize}
    \itemW a set $\mathscr V(c) $ is associated with $c$; 
    for the empty graph, $\mathscr V(\varnothing)=\{*\}$ is the singleton
    \itemW a finite group $G(c)$ and
    \itemW a group action  $G(c) \curvearrowright \mathscr V(c)$
     of $G(c)$ on $\mathscr V(c)$
      \end{itemize}
      \itemB if $g=c_1c_2\cdots c_n$ is a factorisation in connected
      components $c_i$, then $\mathscr V(g)$ satisfies\footnote{One 
       could relax this condition so that there exist 
       domains $\mathscr U(g)$ of $\mathscr V(g)$ compatible with 
       the $G(g)$-action and such that $\mathscr U(g) \subset \mathscr 
       U(c_1)\amalg \ldots \amalg \mathscr U(c_n)$. 
        However, we keep the natural condition 
        \eqref{eq:natural}
        }
     \begin{equation}
     \label{eq:natural}
\mathscr V(g) = \mathscr V(c_1)\amalg \cdots\amalg \mathscr V(c_n)\,     
     \end{equation}
      \end{itemize}
        The collection $\{\mathscr V(g), G(g) \}_{g\in H}$
    is a \textit{system of graph-group actions}. \par
   If additionally, 
%       there is a family 
%       \begin{itemize}
%       \itemB 
      for each graph 
      in $g\in H$,  one has functions
      \[u_g:\mathscr V(g) \to \C \,\,(\mtr{or}\,\,\re)\,,\] 
%     \end{itemize} 
    then one says that 
    $\{u_g,\mathscr V(g), G(g) \}_{g\in H}$, or 
    more succinctly, $\{u_g\}_{g\in H}$     
    is a family of functions
    \textit{supported on} 
    $\{\mathscr V(g), G(g) \}_{g\in H}$. 
    \end{defn}
    We are interested in triples 
     $\{u_g,\mathscr V(g), G(g) \}_{g\in H}$ and formal sums 
     of the type $U=\sum_{h\in H} u_h h$, which we refer to as 
     their \textit{generating functionals}. In this case
     we say that (the set of graphs) $H$ \textit{spans} $U$. Why these are
     functionals instead of functions will become apparent 
     while addressing the applications. At this point also the 
     following terminology, inherited from the physical 
     significance, might seem mysterious: we call the elements of
    $\mathscr V(g)$ the \textit{momenta} of the graph $g$.
%     (notion sometimes used for the elements $\mathscr V(g)$ for $g$ disconnected too). 
    Notice that $u_\varnothing$ is a constant.
      
    \par
      As the last reference to
    TFT in this section, we clarify that the nature of these graphs is not important at this point; examples will 
    be presented in later sections. We just clarify the reader that is well-versed 
    with tensor models, that graphs treated here are not Feynman graphs, but boundary graphs of these 
      in a rank-$D$ TFT. In this context, functions
    $u_g$ are unknown\footnote{For instance, $u_g$ can be the
      correlation functions.}, and one derives equations that they
    should satisfy. Only after knowing solutions we would be able to fix a
    function space $u_g$ should belong to, which is for now
    unspecified. We vaguely refer then to them as `functions'. 
    \par
    
    Next, some words on notation. For a factorisation $g=c_1\cdots c_n$ in 
    connected graphs $c_j\neq \varnothing$, we let 
    $g/c_i$ be the graph with the $i$-th connected component
    deleted,
    \[g/c_i:=c_1\ldots \widehat{c_{i}}\ldots c_n=c_1\ldots
      c_{i-1}c_{i+1}\ldots c_n\,.\] 
    Notice that this deletion does not only care about the graph-class,
but also about its spot in the factorisation, which we can keep
track of thanks to the monoidal structure of $\mathsf G_D$. \par
For $g=c_1\cdots c_n$ as before, let $Y\in \mathscr V(c_r)$, $1\leq r\leq n$. Given a function $v_g:\mathscr V(g)\to \C$ 
we define  the insertion of $Y$ in the $r$-th argument of $v_g$
\[\iota_{Y}^r v_g:\mathscr V(c_1)\amalg \cdots \amalg \mathscr V(c_{r-1}) \amalg
 {\mathscr V(c_{r+1})} \amalg\cdots \amalg \mathscr V(c_n) \to \C \,,\] 
 using \eqref{eq:natural} by
    \[
    (\iota_Y ^r v_g) (X_1,\ldots, X_{r-1},X_{r+1},\ldots,X_n) =
    v_g(X_1, \ldots, X_{r-1},Y,X_{r+1},\ldots,X_n)\,,
    \]
    where $X_i\in   \mathscr V(c_i) $ for each $i\neq r$.
    
    \begin{defn}
      Let $H\subset \mathsf{G}_D$ span the functional
      $U=\sum_g u_g \,g$.  Given any connected graph
      $h\in\mathsf G_D$, $h\neq \varnothing$, and an arbitrary graph
      $g =c_1\cdots c_n$ factorised in connected components $c_i$, 
      we define $I(g,h)=\{ i \in |\![1,n]\!| \,|\, h= c_i \}$; 
      that is, $I(g,h)$ is the subset of numbers 
      that indexes the factors of $g$ that coincide with $h$. For
      $r\in I(g,h)$, we label by $h\hp r$ the appearance of $h$ in 
      the  $r$-th factor of $g=c_1\cdots c_n$.  We define the
      \textit{functional graph derivative} with respect to $h$
      (evaluated at $X$) as the functional
     \[
    \dervfunc{U}{h(X)} = \suml_{g \in H}\,  
    \suml_{r\in I(g,h)}
    \suml_{\sigma\in G(h)}\,  
    \iota_{\sigma(X)}^r  u_{g}\, \,(g/h\hp r)\,,
    \qquad   X \in \mathscr U (h) \subset \mathscr V(h)\,.
    \]
    The well-definedness follows from
    condition \eqref{eq:natural}, which 
    implies that in each case  $\iota_{\sigma(X)}^r  u_{g}$ is 
    indeed a function on $\mathscr  V(g/h\hp r)$.
    We stress that this derivative could be 
    defined in a proper domain $\mathscr U (h) $ of $ \mathscr V(h)$.
    Further, if $h$ occurs nowhere as a factor of $g$ the sum is empty, and
    thus $\delta U/\delta h \equiv 0 $.  The derivative with respect
    to $\varnothing$ is the coefficient of that graph,
    $\delta U/\delta \varnothing:=u_\varnothing \in \C$.
    \end{defn}
    
%     \begin{ex}
    To clarify this definition, consider a monomial functional, $U= u_{g} \, g$, with $g=h^n$ for some integer $n\geq 1$
    and $h$ a connected graph.
    By definition, one has
    \[
    \dervfunc{U}{h(X)} = \suml_{r=1}^n\suml_{\sigma\in G(h)}
    (\iota_{\sigma(X)}^r  u_{h^n}) h^{n-1}
    \, ,
    \]
    which one can rethink as
    \begin{align}
      \dervfunc{(h^n)}{h} & \nonumber
                            =\bigg(  \dervfunc{h}{h}h^{n-1}+h\dervfunc{h}{h}h^{n-2}+\ldots+ h^{n-1}\dervfunc{h}{h}  \bigg) \\
                          &=\bigg(  G(h)h^{n-1}+hG(h)h^{n-2}+\ldots+ h^{n-1}G(h) \bigg)
                            \,, \label{eq:omitarrow}
    \end{align}
    if
    \begin{equation} \label{eq:GroupActionGeneral}
      \bigg(\dervfunc{h}{h} u_h\bigg) (X)= \sum_{\sigma \in G(h)}
      (\sigma \cdot u_h)( X) = \sum_{\sigma \in G(h)} u_h(\sigma X) \,.
    \end{equation}

    To illustrate the action in slightly more generality, if in $u_{c_1\ldots c_{r-1}hc_{r+1}\ldots}$ none of the $c_j$ is isomorphic to $h$, then
    \begin{align} \label{eq:GroupActionGeneraldos}
     \bigg(\dervfunc{h}{h} u_{c_1\ldots c_{r-1}hc_{r+1}\ldots} \bigg) (X) & =
     \sum_{\sigma \in G(h)} \iota_{\sigma(X)}^r u_{c_1\ldots c_{r-1}hc_{r+1}\ldots}  \, .
     % &=
    \end{align}
%     . \\
    
   On this account, the useful symbolism to keep in mind is that 
the derivative of a graph with respect to itself is the group action of $G(h) $ on
    $ \mathscr V(h)$,
    \begin{equation} \label{eq:GroupActionGeneraluno}
    \mbox{}
     \dervfunc{h}{h}=   G(h)\curvearrowright \{\mbox{functions }\mathscr V (h) \to \C \}\,.
    \end{equation}
    
    In the sequel, we will often abuse on notation and write this equality without 
    the curved action-arrow, as we already did above in eq. \eqref{eq:omitarrow}. 
    The moral is that each factor $h$ occurring in a term of the type 
    \[u_{c_1 \ldots c_{r-1 }h c_{r+1}\ldots c_n } (c_1 \ldots c_{r-1 }h c_{r+1}\ldots c_n) \] 
    is a potential $G(h)$-orbit of the $r$-th 
    argument $u_{c_1 \ldots c_{r-1 }h c_{r+1}\ldots c_n}$. \par
    For iterated derivatives with respect to $h$,
    one can see by induction that for
    $n\geq 2$, the iteration of $n$ graph derivatives
    applied to $h^n$ yields
    \begin{align}
     \dervfunc{^n \,(h^n)}{h \, \delta h \cdots  \, \delta h} (X_1,\ldots,X_n)
     = \sum_{\mu\in \Sym(n)} 
     \sum_{(\sigma^1,\dots,\sigma^n) 
     \in 
     G(h)^n
     }
     \iota^{1}_{\sigma^{\mu(1)}(X_{1})}
     \cdots 
      \iota^{n}_{\sigma^{\mu(n)}(X_{n})}\,.
    \end{align}
    If $G_i(h)$ is the $i$-th factor of the group $G(h)^n$, a more
    transparent notation of last equation is
    \begin{equation}
          \dervfunc{^n \,(h^n)}{h \, \delta h \cdots  \, \delta h} 
     = \sum_{\mu\in \Sym(n)} 
    %  \sum_{(\sigma^1,\dots,\sigma^n) 
    %  \in 
    %  (G(h))^n
    %  }
    G_{\mu(1)}(h)\,
    G_{\mu(2)}(h)\,
    \cdots
    G_{\mu(n)}(h)
    \,, \label{eq:GroupAction_n}
    \end{equation}
    where the group $G_{\mu(i)}(h)$ acts on the $i$-th factor of the
    set $\mathscr V(h)\amalg \cdots \amalg \mathscr V(h)$.  The group
    corresponding to the $n$-th derivative of the $n$-th power of a
    graph $h$ with respect to itself is
    \begin{equation}\label{eq:wr}
     \dervfunc{^n \,(h^n)}{h \, \delta h \cdots  \, \delta h} 
     =
     G(h)\wr \Sym(n) \,.
    \end{equation}
    In this case, the wreath product $G(h)\wr \Sym(n)$ is the
    semi-direct product $G(h)^n \rtimes_\psi \Sym(n)$, with the obvious
    action $\psi$ of $\Sym (n )$ on the $n$ copies of $G(h)$. 
    \par
    To give further detail, given a generating system of graph-group actions $\{\mathscr V(g),G(g)\}_{g\in H}$
    and $h\in H$,  consider a function $F:\mathscr V(h)\amalg \cdots \amalg \mathscr V(h) \to \C $. 
    An element $\Omega=(\boldsymbol \sigma;\mu)=(\sigma^1,\ldots,\sigma^n;\mu)$ of the group 
    in eq. \eqref{eq:GroupAction_n} acts as follows:
    \begin{equation}  \label{eq:encoords}
    (\Omega \cdot F)(X_1,\ldots,X_n)=F\big(\sigma^1(X_{\mu({1})}),\ldots,\sigma^n(X_{\mu(n)})\big)\,.
    \end{equation}
    By departing from eq. \eqref{eq:encoords},
    the composition with another element $\Xi=(\boldsymbol\tau,\nu)$ in the group \eqref{eq:GroupAction_n} 
    is easily proven to yield $\Xi \circ \Omega= 
    (\boldsymbol \tau \psi_{\nu}(\boldsymbol \sigma ); \nu \mu   )$, 
    where 
    \[\psi:\Sym(n)\to \Aut(G(h)^n) \qquad \mu\mapsto [\psi_\mu:\,(\sigma_i)_{i=1}^n \mapsto (\sigma_{\mu(i)})_{i=1}^n\, ]\,,
    \]
    which is the product of $ G(h)^n \rtimes_\psi \Sym(n)$, as claimed.
%     \end{ex}

    \subsection{Examples of graph-group action systems}\label{sec:examplesGGAS}   
    
    Roughly stated, a multivariable graph calculus (of $n$ graph variables)
    consists of generating functionals of functions $u_g$ 
    supported on a system
     of graph-group 
    actions $\{ \mathscr V(g), G(g)\}_{g\in H}$ that are 
 spanned by a finite set $H$. We take $H \subset \FM(\{h_1,\ldots,h_n\} )$,  the free monoid
 generated\footnote{We recall
    that the \textit{free monoid}  generated by $\mathfrak h=\{h_1,\ldots, h_n\}$ 
    is in this case the following set 
    $\FM(\mathfrak h)=\{l_1\cdots l_m : m\in \mathbb Z_{\geq 0} \,\and \, l_i\in \mathfrak h\}$ 
    endowed with the concatenation operation; containing 
    the empty graph, i.e. the empty word.} by 
    $n$ non-isomorphic graphs $\hf=\{h_1\ldots,h_n\}$. For a
    multivariable graph calculus the key property is that 
    the graph-group  actions $G(h_i)$
    are pairwise independent, that is for each $i,j=1,\ldots,n$, 
    \begin{equation} 
    \label{eq:independence}
    \dervfunc{h_i}{h_j}= \delta^{i}_j \, G(h_i)\,,  \qquad\quad h_i,h_j\in\hf\,.  
    \end{equation}
    For the special element $g=h_1^{\alpha_1}\cdots h_n^{\alpha_n}$ the 
    restriction imposed by eq. \eqref{eq:independence} implies
    \begin{equation} \label{eq:groupactionindependence}
      \dervfunc{g}{g} = G(g) = G(h_1)\wr \Sym(\alpha_1)\times
      G(h_2)\wr \Sym(\alpha_2) \times \cdots \times G(h_n)\wr
      \Sym(\alpha_n) \,.
    \end{equation}
%     \subsubsection{}
    Before formally defining multivariable
    graph calculus, the next examples are just meant to 
    illustrate last action \eqref{eq:groupactionindependence},
    rather than the role of the graphs in 
     graph-generated actions, and therefore 
     can be skipped (to Sect. \ref{sec:mutivgraphcalc}).

           \begin{ex}
    Let $\zeta_n\neq 1$ denote a $n$-th root of unit ($n\geq 2$),
    and consider the system of graph-group actions with 
    a single graph $g$. Let $G(g)$ be the group spanned by $\zeta_n$ 
    $\mathscr V (g)= \C$ by multiplication.  
    Then the functional graph derivative of $g$ with respect to itself
    on the identity $\mathrm{id}_{\C}$ vanishes identically: 
    \[
    \bigg(
    \dervfunc{g}{g} \bigg) \mathrm{id}_{\C} \equiv 0\,. 
    \]
    The $G(g)$-orbit of the function $f_n:z\mapsto z^n$ yields
     \[
    \bigg(
    \dervfunc{g}{g (z)} \bigg) (f_n)= n\cdot z^n\,. 
    \]

     \end{ex}
      
      \begin{ex}
      Consider a finite set $H\subset \FMh $ 
      and the following  graph-group actions system 
      \[ G(g)= \Sym( |g^0 |),\qquad \mathscr V(g)=M_{|g^0|\times |g^0| }(\re)\,.\]
      Here   $|g^0 |$ is the number of vertices of $g$. The action
      of the symmetric group on the matrices permutes 
      columns (or rows). Then the orbit of the determinant $\det: \cV(g) \to \re $
      vanishes identically. This follows from considering, for an arbitrary matrix $X=(X_{ab})\in \cV(g)$,
      \begin{align*}
             \bigg(\dervfunc{g}{g (X)} \bigg) \det (\balita) & = \suml_{\sigma\in\Sym(|g^0|)}\det (X_{\sigma(a)b}) \\
             & =  \sum_{\sigma\in A} \det (X_{\sigma(a)b}) +
              \sum_{\sigma\in \Sym(|g^0|) - A} \det (X_{\sigma(a)\,b}) \\
              &=  \sum_{\sigma\in A} \det (X_{ab}) 
             -\sum_{\sigma\in \Sym(|g^0|)-A } \det (X_{ab})=0\,,
     \end{align*}
      where $A$ is the alternating subgroup; its complement 
      in the symmetric group consists of odd-degree permutations, whence
      the common minus sign in the last line. Both have the same order, which explains
      why the sum vanishes independently of $X$. 
      \end{ex}

      \begin{ex} 
      Let $K$ be a finite group that accepts (cf. \cite{FinIrreps} for a criterion)
      faithful irreducible representations. Consider $n$ of them $\pi^i:K\to \mathrm{End}(W_i)$, and set 
      $G(h_i)=K$ for each $i=1,\ldots,n$.
      Define for each $i$ the momenta of $h_i$ as the matrix space $\mathscr V(h_i)=\pi^i(K)$  (since $K$ is finite, irreps are finite-dimensional).
     The group $K$ acts on  $ \cV(h_i)$  by
       \[
      \cV(h_i)\ni\pi^i(m) \stackrel{u  }{\mapsto} \pi^i(u)\pi^i(m)=\pi^i(um)\,, \qquad (m,u\in K)\,.
       \]
      Consider the following functions $f^i$  defined in terms of
      the characters $\chi^i(u)=\Tr(\pi^i(u))$, $u\in K$,
      \[f^i: \mathscr V(h_i) \to \C\,,
      \quad 
      f^{i}(X_i) = \chi^i(m)^* \chi^i(m)\,, \qquad  (X_i=\pi^i(m))\,.
      \]
      Then for $X_i=\pi^i(m)$, the following holds:
         \begin{align*}
       \bigg(\dervfunc{h_i }{h_i(X_i)}\bigg)f^{i}  & = 
       \sum_{u\in K}  (u\cdot  f^{i})(\pi^i(m)) \qquad \mbox{($i$ is fixed)} \\
       &=\sum_{u\in K}  \chi^i(um)^* \chi^i(um)\,\\
       &=\sum_{u\in K}  \chi^i(u)^* \chi^i(u) = |K|\,.
%        &= |K|\,.
      \end{align*}
      Fix any $g\in \hf$ and let $\pi$ be the associated
      representation. Define for any $h_i,h_j\in \hf$, and 
      $X=\pi(m)$,
      \[
      F^{ij}: \mathscr V(g) \to \C\,,\quad 
      F^{ij}(X) =  \chi^i(m)^* \chi^j(m)\,.
      \]
      For $X=\pi(m)$,
      \begin{align*}
       \bigg(\dervfunc{g  }{g (X)}\bigg)F^{ij}  = 
       \suml_{u \in K}  \chi^i(u m)^* \chi^j(u m)
        = 
       \suml_{u' \in K}  \chi^i(u')^* \chi^j(u')
       = \delta_{ij} |K|\,.
      \end{align*}
      Last equality is due to Schur orthogonality.
      
      \end{ex}

    \begin{ex} Let $D\in \Z_{\geq 1}$. 
     Let two non-isomorphic graphs $H=\{g,h\}\subset \mathsf G_D$ 
     parametrise the system of graph-group actions
     given by $\{\mathscr V (l),G(l)\}_{l\in H}$, being 
     \begin{align*}
     \mu\in G(g)&=\Sym(D)   &&& Z=(z_i)_{i} \in 
     \mathscr V(g) &= \C ^D && & \mu&:(z_i) \mapsto (z_{\mu(i)}) \\
     \tau\in G(h)&=\Z_2  &&& \varepsilon\in\mathscr V(h)&=\Z_2  &&& \varepsilon&:s\mapsto \tau \varepsilon 
     \end{align*}
    where $ \Z_2$ is written multiplicatively $\{-1,1\}$. 
    Let $F:  \mathscr V(g)\amalg \mathscr V(g) \amalg  \mathscr V (h)\to \C$
    be given by, say,
\[F(Z, W,\varepsilon)=\frac{\varepsilon}{2 (D!)^2}  \ee^{{- |W| +\varepsilon |Z| }}
=: c(D) \cdot \varepsilon \cdot \ee^{{- |W| +\varepsilon |Z|} } \,.
 \]     
Then the functional graph derivative of $g^2f$ with respect to itself
yields the following group-orbit, when applied to $F$:
 \allowdisplaybreaks[1]
 \begin{align*}
& \dervfunc{^3 \, (g^2h)}{g(Z_1)\,\delta g(Z_2)\,\delta h(\varepsilon)} F \\
 & = 
 \bigg(\dervfunc{^3 \, ggh}{g \,\delta g\,\delta h} F \bigg)(Z_1, Z_2,\varepsilon)
 \\
  & = \Big( \big[ G(g) \wr \Sym(2) \times G(h)\big] \cdot F\Big)(Z_1, Z_2,\varepsilon)
  \\
  & =\suml_{\mu\in \Sym(2)}\suml_{\tau\in\Z_2} 
  \suml_{\sigma  \in \Sym(D)} \suml_{\rho\in\Sym(D)}
 F \big(\sigma (Z_{\mu(1)}),\rho(Z_{\mu(2)}),\tau \varepsilon \big) 
 \\
  & = \suml_{\tau\in\Z_2} 
  \suml_{\sigma  \in \Sym(D)} \suml_{\rho\in\Sym(D)}
\Big\{F \big(\sigma(Z_1),\rho(Z_2),\tau \varepsilon \big) + F \big(\sigma (Z_2),\rho(Z_1),
\tau \varepsilon \big)\Big \}
  \\
   &= 2 c(D)\cdot (D!)^2\cdot\varepsilon \cdot \big( \ee^{\varepsilon |Z_1|-|Z_2|} + \ee^{\varepsilon |Z_2|-|Z_1|}  \\
   & \qquad\quad\qquad\quad\qquad\,- \ee^{-\varepsilon |Z_1|-|Z_2|} -\ee^{-\varepsilon |Z_2|-|Z_1|}\big) \\
   &=  \varepsilon \big( \ee^{-|Z_1|} \sinh |Z_2|  + \ee^{-|Z_2|} \sinh |Z_1| \big)\,.
 \end{align*}
  \allowdisplaybreaks[0]
We have used the invariance under the action of 
two copies $\Sym(D)$, which contributed a factor $(D!)^2$.

    \end{ex}
    
    \subsection{Multivariable graph calculus} \label{sec:mutivgraphcalc}
    Let $\hf=\{h_1^{}\,\ldots, h_n^{}\} \subset \mathsf{G}_D$
    be a set of connected, non-isomorphic graphs.
    For the basis of the multivariable calculus the free monoid $\FM(\{h_1,\ldots,h_n\} )$ is too `verbose',
    and not each one of its elements has
    the ordered form $h_1^{\alpha_1}\cdots h_n^{\alpha_n}$. 
    This could in principle be solved by taking
    the free commutative monoid  $\FM^{\mtr{ab}}(\hf)$ instead,
    which, however, tuns out to be overly restrictive (for our aims). 
    A mild compromise between these two alternatives ---the free monoid
    and its abelianisation--- is
    to allow to permute letters in an arbitrary word,
    as to make use of 
    the action \eqref{eq:groupactionindependence}, and then 
    in some sense undo the changes. Next definition  introduces 
    precisely such reordering.

    \begin{defn} 
   Given a finite set of graphs
   $\mathfrak h=\{h_1,\ldots,h_n\} $,
   the \textit{degree} $|g|$ of an element $g$ in 
   $\FM(\mathfrak h) $
   is the number of factors of $g$, i.e. the number 
   of connected components $g$ consists of.
    We let $\Sym(|g|)$ act by permuting the 
   factors of $g$, $g\mapsto \sigma (g)$; notice that 
   $\Sym(|g|)$ left-acts naturally  as
     $\sigma \cdot f=f\circ \sigma\inv$ on functions 
   $f:\mathscr V(g)\to \C$. 
   Given a family of functions $\{u_g\}_{g\in H}$
   supported on a system of graph-group actions 
   $\{\mathscr V(g), G(g)\}_{g\in H\subset \FMh }$,
   and given a $g\in H$,
   we declare the pairs $(u_g,g)\sim (\sigma\cdot u_g, \sigma (g))$ 
   equivalent  for each $\sigma\in\Sym(|g|)$. The notation 
   we choose for this equivalence, called \textit{reordering}, is 
   \begin{equation} \label{eq:relation}
u_g g \sim u_h h \quad  \mbox{ if and only if }\quad u_h=\sigma \cdot u_g\,\,\and\,\,h=\sigma (g)
\mbox{ for certain }  \sigma\in \Sym(|g|)\,.  
   \end{equation}
   
   \end{defn}

   \begin{defn}
   Given a finite set of connected non-isomorphic graphs 
   $\mathfrak h=\{h_1,\ldots,h_n\} \subset \mathsf G_D$, a 
    system of
   graph-group actions
   $\mathcal S=\{\mathscr V (h),G(h)
   \}_{h\in{\mathfrak h}}$
   is said to be \textit{independent}
   if eq. \eqref{eq:independence} holds. When
    the context is clear, we just say that `$\mathfrak h$ is 
    independent', or that $\mathcal S$ is.
   \end{defn}

    \begin{defn}
   A \textit{multivariable graph calculus} 
   $ \mathscr C(\mathfrak{h})$ or a \textit{graph calculus} with variables 
   $\mathfrak h$ consists of two objects: 
   \begin{itemize}
    \itemB 
the choice of an independent system of graph-group actions $\{\mathscr V (h),G(h)
   \}_{h\in{\mathfrak h}}$ for a finite set $\mathfrak h \subset \mathsf G_D$ and
   
   \itemB the set  
   of finite formal sums in elements of $g \in \FM(\mathfrak{h})$ 
   having each of these a function of the form $v_g:\mathscr V(g)\to \C$  as coefficient,
   modulo reordering. 
   That is, 
   \begin{align}
   \mathscr C(\mathfrak{h})= \big\{  \,\textstyle\sum_{g } v_g \,g \,\,\big| \,
   v_g \equiv 0 \mbox{ for almost all } g\in \FM(\mathfrak{h}) 
     \big\} \,\big/\sim
   \end{align}
   where $\sim $ is the linear extension of relation \eqref{eq:relation},
   abusing on the same symbol.

   \end{itemize}
     \end{defn}

    \subsection{Algebraic structure}
    We now explore the structure of a
 graph calculus $\mathscr C(\hf)$ with 
 variables $\hf=\{h_1,\ldots,h_n\}$.
 The elements of $\mathscr C(\hf)$,
 called also \textit{functionals}, have a non-unique
  representation, since 
  $\sum_g v_g \,g=\sum_{\tilde{g} }v_{\tilde g} \,\tilde g$
  where $\tilde g=\tau_g(g) $ and $v_{\tilde g}=\tau_g \cdot (v_g)$ for an arbitrary $\tau_g \in \Sym(|g|)$.
  For sake of computability, it will be helpful to 
  be able to fix representing elements $g$ that span
  a functional, and subordinate the order of 
  the arguments of the coefficient-functions to that choice. 
  \par 
  We write 
  $g \cequal h$ for any  $g,h\in \FMh$ 
  if $g=h $ in the free commutative monoid $\FM^{\mtr{ab}}(\hf)$ spanned by $\hf$.
   In other words,
   $g \bcequal h$ if and only if
   $h$ and $g$ match in $\FMh$ up to a rearranging $\sigma \in \Sym(|g|)$, i.e. if $\sigma(g)=h$.
   
   \begin{defn} Given  a family of functions $\{v_l\}_l$ supported on $\{\mathscr V(l),G(l)\}_{l\in H}$ where $  H\subset \FMh$. Let $g,h \in H$ be such that 
   $h \bcequal g$. We define for the \textit{reordering 
   of a function} $v_h:\mathscr V(h)  \to \C $ \textit{with respect to} $g$  by 
   \[\langle v_{h}\rangle_g= \sigma \cdot  v_g\,, \mbox{if } \sigma (g)=h \mbox{ as elements of }\FMh, \]
   being $\sigma$ the
   rearranging element $\sigma \in \Sym(|g|)=\Sym(|h|)$. 
   \end{defn}
   
    We shall drop the subindex $g$ in $\langle \,\balita\, \rangle_g$ when the context is
    clear.  If one factors $g$ as $g_1g_2$ with respect to an
    `abelianised' product, an element
    $\sigma \in \Sym (|g|) $ serves as correction, so that
    $\sigma (g_1g_2)=g$.  Their rearranging yields $\langle u_{g_1}t_{g_2}\rangle_g \,=\sigma \cdot (u_{g_1}t_{g_2})$ for suitably chosen functions $u_{g_1}$ and $ t_{g_2}$.
    In general, the collection of graphs $g_1,\ldots,g_r$ is not required to be connected.
    If the context is clear, we pick 
    this rearranging element in a smaller group $\sigma \in \Sym (r)$
    that only permutes the arguments of $\mathscr V(g_i)$.

    \begin{defn} 
    Denote by $H_1  H_2$ the subset 
    $\{g_1  g_2 \,|\, g_a\in H_a\}$ in the 
    free commutative monoid $\FM^{\mtr{ab}}(\hf)$ spanned by an independent set of graphs $\hf$.
     Given     two functionals in $\mathscr C(\hf)$, 
    $U=\sum_{h\in H_1} u_{h} \,h \,$ and % \sum_{g\in G_1} u_{g}
                                          % \star \J(g)\,, \\
    $T= \sum_{h\in H_2} t_{h} \, h$, we define their product $V=(U\cdot T)$
    as the functional 
    \[V = U\cdot T =\sum_{g\in H_1  H_2 \subset \FM^{\mtr{ab}}(\hf)} v_g \,g \,,\] 
    whose coefficients $v_g$ are given by the `ordered convolution'
    \begin{equation} \label{eq:product}
    v_g=\suml_{ \substack{(g_1,g_2)\in H_1\amalg H_2  \\ g_1g_2 \bcequal g} } \langle  u_{g_1}t_{g_2}  \rangle_g\,.
    \end{equation}
    \end{defn}
    
    \begin{lem}
     This product on $\mathscr C (\hf)$ is commutative.

    \end{lem}
\begin{proof}
 Let $U= \sum_{j\in J} u_j \, j$ and $ T=\sum_{l\in L} t_l\,l$ be in $\Calc$.
 Given $g\in \FMh$ with $g\cequal jl$ for some $j\in J$ and some $l\in L$ we show
 that the components $v_g$ of $V=U\cdot T$ and $\tilde v_g$ of $\tilde V=T\cdot U$
satisfy $v_g\,g= \tilde v_g\,g$ in $\Calc$. It suffices to exhibit an element 
$\nu\in \Sym (|g|)$ that satisfies $\langle u_jt_l \rangle_g\,g= \langle t_lu_j \rangle_g\, \nu (g)$. This $\nu$ will be next constructed. \par We have the freedom to assume that 
$g =h_1^{\alpha_1}\cdots h_n^{\alpha_n}$. Since $j,l\in \FMh$,
\begin{equation}
 j\bcequal h_1^{\theta_1}h_2^{\theta_2}\cdots h^{\theta_n}_n
\qquad \and \qquad
l\bcequal h_1^{\lambda_1}h_2^{\lambda_2}\cdots h^{\lambda_n}_n
\label{eq:thetaslambdas}  
 \end{equation}
 for some $0\leq \lambda_i,\theta_i\leq \alpha_i$ 
 that satisfy $\alpha_i=\theta_i + \lambda_i$, $i=1,\ldots,n$. 
 In the notation introduced above, $\#\,I(j,h_i)=\theta_i$ and 
 $\#\,I(l,h_i)= \lambda_i $.
We begin by assuming that the relations above are equalities,
\begin{equation}
 j= h_1^{\theta_1}h_2^{\theta_2}\cdots h^{\theta_n}_n
\qquad \and \qquad
l= h_1^{\lambda_1}h_2^{\lambda_2}\cdots h^{\lambda_n}_n
  \label{eq:thetaslambdasp} 
\,, 
 \end{equation} 
and restore towards the end the more general form \eqref{eq:thetaslambdas}.
 Let $|j|=\theta_1+\ldots+\theta_n$
and $|l|=\lambda_1+\ldots+ \lambda_n$
be the orders of $j$ and $l$.
We define first 
 $\sigma \in \Sym(|g|)$
as the $(|j|,|l|)$-shuffle determined by  
\[
\sigma= \raisebox{-.45\height}{
\includegraphics[width=8cm]{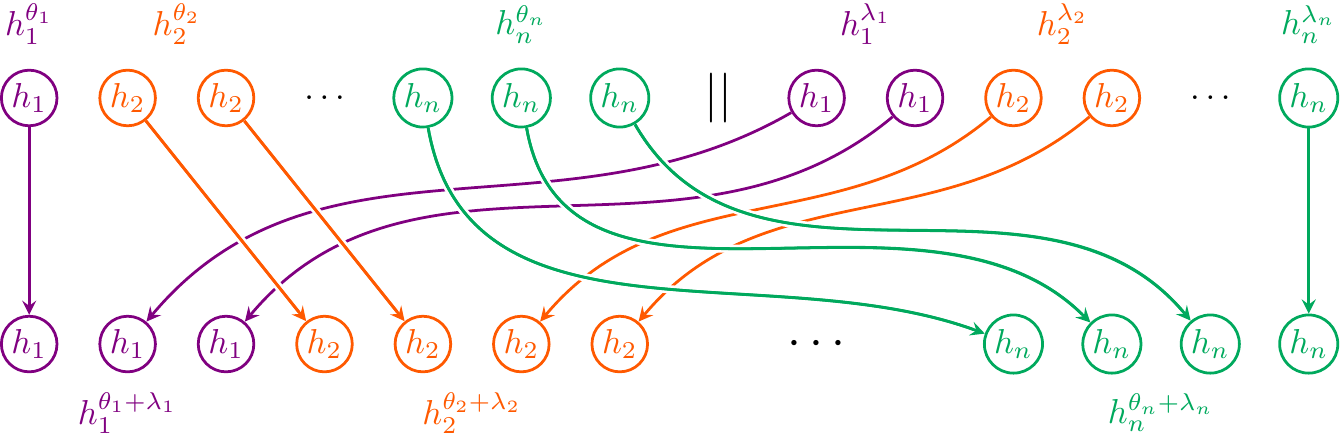}
} \in \Sym(|g|)
\]
Each $h_i$ before double bar in the first row is a factor of $j$;
after the double bar, the $h_i$'s represent the factors of
$l$. The lower is a factorisation of $g$. 
Thus the diagram states that $\sigma(jl)= g$. Analogously, 
we can define a $(|l|,|j|)$-shuffle $\rho$ 
that satisfies $\rho(lj)=g$. This is depicted in 
the following diagram, in which we represent 
$l$ to the left of the double bar and $j$ to the right.

\[
\rho= \raisebox{-.45\height}{
\includegraphics[width=8cm]{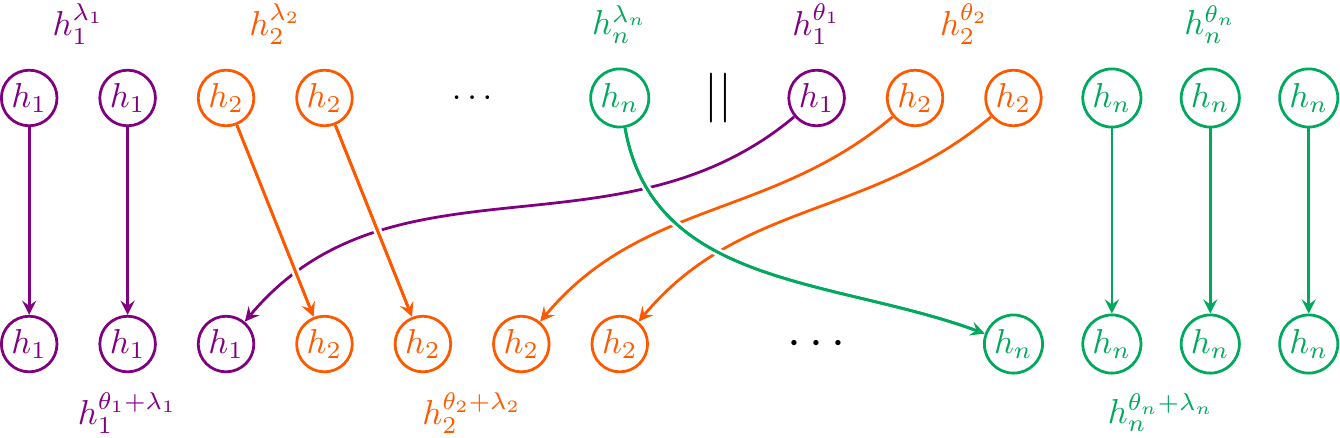}
} \in \Sym(|g|)
\]
One has $\rho(lj)=g=\sigma(jl)$ but this still 
does not guarantee that $\rho\inv\sigma(v_{jl})=\tilde v_{lg}$. In order 
to correct this, we define certain permutations $\tau_i$ that are constant
everywhere except in the elements pertaining a particular $h_i$ for fixed $i$.
This embeds $\Sym(\#\,I(g;h_i))\subset \Sym(|g|)$.  
Notice first that for each such element $\mu\in\Sym(\#\,I(g;h_i))$ 
\[
w_g \,g =(\mu\cdot w_g) \,\mu(g)=(\mu \cdot w_g) \,g\,
\]
holds in $\mathscr C(\hf)$
for any $w_g:\mathscr V(g)\to \C$. 
For each $i=1,\ldots,n$, define $\tau_i$ as certain
permutation $\beta_i$ (given below) in the range $[\alpha_{i-1}+1,\alpha_{i}]$
and constant outside it:
\[\tau_i(x)=\begin{cases} 
   \beta_i(x) &   0<  x-s_i < \alpha_i      \\ 
x & \mbox{otherwise} \end{cases}\,,\]
where $s_{i}:=(\alpha_1+\ldots+\alpha_{i-1})$
\[\beta_i(x)=\begin{cases}
          x+ \lambda_i& s_i<x \leq s_i+ \theta_i \,,\\
           x-\theta_i   & \theta_i <x \leq \theta_i + \lambda_i =\alpha_i\,.
             \end{cases} 
\]
Then the sought-after $\nu$ is 
\[
\nu= \rho\inv\circ(\tau_1\circ \cdots\circ \tau_n)\circ \sigma\,,
\]
which by construction satisfies 
\begin{equation}
(
\nu\cdot v_{jl})\,\nu(jl)= \tilde v_{lj} \,lj \,.  \label{eq:fin}
\end{equation}
The map $\nu$ is thus given by
\[
\scalebox{1.5}{$\nu$}\, =\quad \raisebox{-.48\height}{\includegraphics[width=.75\textwidth]{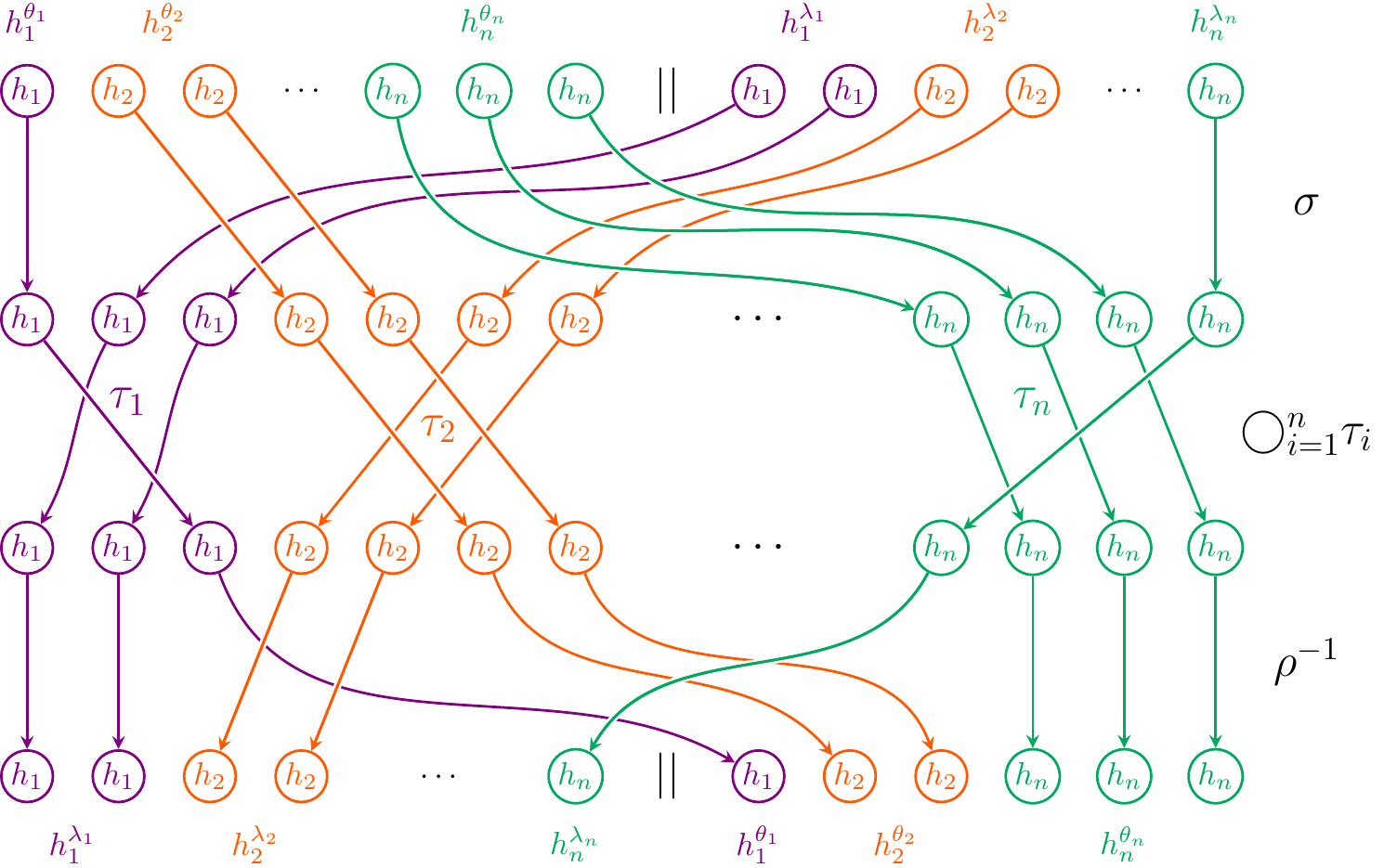}}
\]

We now come back to the strongest, original statement, in which
$l$ and $j$ have the  form \eqref{eq:thetaslambdas},
instead of \eqref{eq:thetaslambdasp}. This means 
that there are permutations $\gamma\in\Sym (|j|) \subset \Sym(|j|+|l|)$ and $\delta\in \Sym(|l|)\subset \Sym(|j|+|l|)$ with 
\begin{equation}
 \gamma(j)= h_1^{\theta_1}h_2^{\theta_2}\cdots h^{\theta_n}_n
\qquad \and \qquad
\delta(l)= h_1^{\lambda_1}h_2^{\lambda_2}\cdots h^{\lambda_n}_n\,.
\end{equation}
Then we correct $\nu$ by these two elements:
\[
\nu= \rho\inv\circ(\tau_1\circ \cdots\circ \tau_n)\circ \sigma \circ (\gamma,\delta)
\]
which satisfies, in the most general case, eq. \eqref{eq:fin}. 
The statement follows by linear extension of it.\end{proof}

    \begin{ex}
     To illustrate this notation, consider the sets $H=\{fg,f^2\}$
      and  $K=\{fg,g^2\} $  of coloured graphs and let
      $U=\sum_{e\in K} u_e e$ and $T=\sum_{e\in H} t_e e$. If $X_A$ are momenta of $f$
      and $Z_A$ of $g$ (for $A=1,2$), we pick a particular graph
      $l=f^2g^2$, and define the following permutations in $\Sym(4)$
          \begin{equation}
 \raisebox{-4ex}{\includegraphics[width=.63\textwidth]{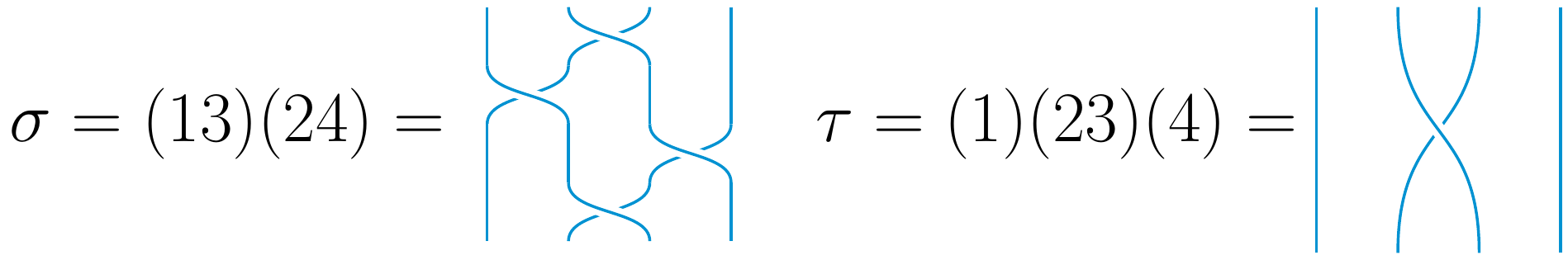}}
\end{equation}
      The coefficient $v_l=v_{f^2g^2}$ is given by
    \begin{align*}
    v_{f^2g^2}(X_1,X_2,Z_1,Z_2) &=
    \big(\, \langle u_{fg} t_{fg}\rangle_l
    \,+\, 
    \langle u_{g^2}t_{f^2} \rangle_l \,\big)
     (X_1,X_2,Z_1,Z_2) \\
     &=
    \big(\tau\cdot (u_{fg} t_{fg}) \big) (X_1,X_2,Z_1,Z_2) \\
    & \quad + 
    \big(\sigma\cdot(u_{g^2}t_{f^2})\big)
     (X_1,X_2,Z_1,Z_2) \\
    &=  u_{fg} t_{fg} (X_1,Z_1,X_2,Z_2) +
    u_{g^2}t_{f^2}(Z_1,Z_2,X_1,X_2)\,.
    \end{align*}
%     We have used permutation notation $(ij)$ for the elements 
%     in  ($i\neq j=1,\ldots,4$).

    \end{ex}

    Structures appearing in the graph 
    calculus resemble the monoid ring. Given a commutative unit ring 
    $R$ and a monoid $M$, the monoid ring \cite[Ch. II]{langalgebra}
    is built by formal finite sums
    in $M$ with coefficients in $R$,
    \[
    R[M]=
%     \{
    \{ \textstyle\sum_m r_m \, m\,|\, r_m\in R, m\in M \,\&\, r_n \neq 0 \mbox{ for finitely many $n\in M$} \}\,,
%     \}
    \]
    and endowed with the convolution product. 
    The structure of the graph calculus generated by 
    $n$ variables $\mathfrak h= \{h_1,\ldots,h_n\}$ 
    requires to define, instead of the ring $R$, 
    the collection $\A$ of algebras of functions  
    \[
    \A = \prod_{g\in \FM(\mathfrak h)} A_g,\qquad\where 
     A_g=    \{ \mbox{functions }
 \mathscr V(g)\to \C  \}\,,
    \]
    and then to consider the following restricted version of the monoid algebra
    \[
    \A\big[\FM(\mathfrak h)\big]:= \Big\{ \textstyle\sum_{g\in \FM(\mathfrak h)}
    u_g \, g \,\Big |\,  u_g \in  A_g\,,\, u_h=0 \mbox{ for almost all } h\in \FM(\mathfrak h)
    \Big\}    \,.
    \]
    Then we see that 
    \[
    \mathscr C(\hf)= %\A\big[\FM(\mathfrak h)\big]^{\Sym} =
    \A\big[\FM(\mathfrak h)\big] /\!\sim\,.
    \]
    
    Given functionals
    $U=\sum_{h \in H} u_h \, h$ , $ T=\sum_{l \in L} t_l\,l$ in $\mathscr C({\hf})$
    one defines their sum by
     \begin{align*}
     U+T=\sum_{m\in H\cup L} \Big( \sum_{h\cequal m, h\in H} \langle u_h\rangle_m +\sum_{l\cequal m, l\in L} \langle t_n\rangle_m \Big) m \,;
      \end{align*}
 for a scalar $x\in \C$,  the functional $xU$  is defined by componentwise multiplication by $x$, 
      \begin{align*}
    xU=\sum_{h\in H} (x\cdot u_h) \,h \,.
     \end{align*}
The in-depth study of the structure of $\mathscr C(\hf)$ is beyond the scope of this paper. Here, we will prove the properties 
that are useful to in later sections.
  
  %%%% hier bin ich geblieben....

    \begin{defn}[Coloured Borel transformation]

    Let  $V=\sum_{h\in H} v_g \,g$ be a generating functional
    of graph-group actions. Then define its \textit{coloured Borel transformation} by
    \[ \mtr{B_{c}}(V):=\sum_{g\in H\subset   \mathsf G_D } \frac{1}{|G(g)|} v_g\, g\,.\]
    \end{defn}
         If $H\subset \mathsf G_D \ni h$, the following set  
    \[H_h:=\{ g\in \mathsf G_D \,\,|\,\, hg \cequal j \mbox{ for some element $j$ of }H \}\,,\]
    will be relevant for the 
     next lemma.
    
    \begin{lem} \label{thm:Borel} Consider the 
    generating functional of a system of graph-group actions
      $V=\sum_{g\in H} v_g\, g$ belonging to a graph calculus 
      in (connected graph) variables $\mathfrak h = \{h_1,\ldots, h_n\} \subset \mathsf G_D$. 
      Suppose that the coefficients $ v_{c_1\cdots c_a\cdots c_b\cdots c_p }$ 
      obey the following rule under the transposition $(ab)\in \Sym(p)$ of graphs:
   \begin{align}   v_{c_1\ldots c_a\ldots c_b\ldots c_p }=(ab)\cdot v_{c_1\ldots c_b\ldots
        c_a\ldots c_p} \label{eq:invar} \end{align}   
        for each $c_a,c_b\in \mathfrak h$,
        where $a$ and $b$ denote the number of factor (connected component) 
        where $c_a$ and $c_b$ are located, respectively.  Moreover, assume that
      $ v_{c_1\cdots c_r\cdots c_p}$ is invariant under $G(c_r)$ for each
      factor $c_r$, $1\leq r \leq p$. Then, one has for $h=c_r$,
    \[
    \dervfunc{\mtr{B_c}(V)}{h}=\sum_{g\in H_h}\frac{1}{|G(g)|} v_{hg}\,g
    \,,
    \]
    which means that  for each $X\in \mathscr V(h)$,
    \[
    \dervfunc{\mtr{B_c}(V)}{h(X)}=\sum_{g\in H_h}\frac{1}{|G(g)|} (\iota^{1}_{X}v_{hg}) \,g =
    \sum_{g\in H_h}\frac{1}{|G(g)|}  v_{hg}(X, \balita ) \,g
    \,,
    \]
    being $\iota_X^1$ the insertion of momentum $X$ of $h$ at the first
    argument of the function $v_{hg}$.
    % \[
    % \dervfunc{\mtr{B_c}(V)}{h(X)}=\sum_{g\in H_h}\frac{1}{|\Autc (g)|} (\iota^{h}_{X}v_{hg}) \,g
    % \]
    \end{lem}
    
    \begin{proof}
By assumption,  one can bring every element $g\in H$ to the form
      $g=h_1^{\alpha_1}\cdots h_n^{\alpha_n}$, and change the coefficients $v_g$
      accordingly without alteration.   
      Notice that each $g\in H$ can be further factorised as $h^\ell f$, where
      $\delta f /\delta h\equiv 0$, and $\ell$ depends on $g$. Through
      direct computation,
    \begin{align}
    \dervfunc{\mtr{B_c}(V)}{h(X)} & = 
    \sum_{g=h^\ell f \in H} \sum_{r=1,\ldots,\ell} \sum_{\sigma\in G(h)}\label{eq:ele}
    \frac{1}{|G(h^\ell f)|} \iota_{\sigma (X)}^rv_{h^\ell f} h^{\ell-1}f
    \\ 
    &= \nonumber 
    \sum_{g=h^\ell f \in H}   \sum_{\sigma\in G(h)}
    \frac{\ell}{|G(h^\ell f)|} \iota_{\sigma (X)}^1v_{h^\ell f} h^{\ell-1}f 
    \\
    & =  \nonumber
    \sum_{g=h^\ell f \in H}   
    \frac{\ell \cdot |G(h)|}{|G(h^\ell f)|} \iota_{ X}^1v_{h^\ell f} h^{\ell-1}f\,. 
    \end{align}
    First, we used the invariance \eqref{eq:invar}, i.e. 
    $\iota^m_X v_{h^\ell f}=\iota^p _X v_{h^\ell f}$ for each
    $1\leq p,m \leq \ell$; then, the invariance under
    $G(h)$. Throughout, we can assume $\ell\geq 1$, since this is required for a summand 
    in the first equality  of eq. \eqref{eq:ele} not to vanish.  Also, since
    \[
    G(h^\ell f)=G(h^\ell) \wr \Sym(\ell) \times G(f)\,,
    \]
    the orders of the groups should satisfy
    \[
    |G(h^\ell f)|=\ell\,!\,\cdot\,|G(h)|^\ell \cdot |G(f)|=
    \ell |G(h)| \cdot |G(h^{\ell-1}f)|\,.
    \]
    Hence, after cancellation one gets
    \[
    \dervfunc{\,\mtr{B_c}(V)}{\,h(X)}  = 
    \sum_{g\in H_h}
    \frac{1}{|G(g)|} \iota_{X}^1v_{hg} \,  g\,.
     \qedhere \]
    \end{proof}

    \begin{defn}
      The \textit{graph derivative} $\partial U/\partial h$ of a
      generating functional of group actions $U$ is given by the
      coefficient $ v_\varnothing$ of the empty graph $\varnothing$ of
      the functional derivative of $U$ with respect to $h$, 
      $\delta U/\delta h=:\sum_g v_g g$, to wit
    \[
    \dervpar{U}{h}= \bigg(\dervfunc{U}{h}\bigg)_\varnothing\,.
    \]
    Parenthetically, the difference in notations for `partial
    derivative' and `functional derivative' does
    not intend to mirror any difference between
    multivariable ordinary  and functional calculi.
    \end{defn}
    The next result is simple and useful at the same time:
    
    \begin{lem}[Graph calculus Leibniz product rule]\label{thm:Leibniz}
    % \naranja{(Tensormodelkontext hinzuf\"ugen)}
    Consider a graph calculus $\mathscr C(\hf)$ and let 
    $ J,L\subset \FMh $ span functionals  
    $U $ and $T$ in $\mathscr C(\hf)$,
    \begin{align}
    U &=\sum_{j\in J} u_{j} \,j \,, \qquad % \sum_{g\in G_1} u_{g} \star \J(g)\,, \\
    T = \sum_{l\in L} t_{l} \, l \,.%\sum_{g\in G_2} t_{g} \star  \J(g)\,.
    \end{align}
    Then  the graph derivative of the product $V=U\cdot T =\sum_g v_g\,g $ is 
    \[
    \dervpar{(U\cdot T)}{g}=
    \suml_{\Omega\in G(g)} \Omega \cdot v_g
    =
    \suml_{\Omega\in G(g)}\suml_{\substack{\, (j,l)\in J\amalg L  \\ j l\bcequal g }}
    \Omega \cdot  \langle u_{j}t_{l}\rangle_g \,.
    \]
    \end{lem}
    
    \begin{proof} 
      We compute directly the derivative of the product $U\cdot T$
      with respect to $g$ from the $\varnothing$-coefficient of 
     $\delta{(U\cdot T)}/{\delta g}$:
             \allowdisplaybreaks[1]
     \begin{align*}
       \dervpar{(U\cdot T)}{g} & = \bigg(\dervfunc{(U\cdot T)}{g}\bigg)_{\varnothing} \\ 
       & = \bigg(\dervfunc{}{g} \suml_{f\in J L\subset \FM^{\mtr{ab}}(\hf)}\suml_{\substack{\, (j,l)\in J\amalg L  \\ jl\sim f }}
    \langle u_{j}t_{l}\rangle_f\, f \,  \bigg)_{\varnothing}\\ 
     & =    \suml_{f\in J L\subset \FM^{\mtr{ab}}(\hf)}\dervfunc{f}{g}\suml_{\substack{\, (j,l)\in J\amalg L \\ jl\sim f }}
     \langle u_{j}t_{l}\rangle_f   \\
       & = G(g) \curvearrowright  \bigg[\suml_{\substack{\, (j,l)\in J\amalg L \\ jl\bcequal g }}
    \langle u_{j}t_{l}\rangle_g   \bigg] \\
      & =  \suml_{\Omega \in G(g)}\suml_{\substack{\, (j,l)\in J\amalg L  \\ jl\bcequal g }}
    \Omega \cdot  \langle u_{j}t_{l}\rangle_g   \,.
     \end{align*}
 For the second equality we inserted the coefficients explicitly, according to the 
     definition of the product. The fourth equality holds by graph independence,
     eq. \eqref{eq:independence}, guaranteed for a graph calculus.
     \end{proof}
    
    From now on let $2k(g)=|g^0|$ denote the number of vertices of $g\in \mathsf{G}_D$. 
    Our canonical example of system of graph-group actions have the form
    \begin{equation} \label{eq:canonicalSGGA}
    \{u_g, \mathscr V(g)= M_{D\times k(g)}(\mathbb Z),  G(g)=\Autc (g) \}_{g\in H }\,.
    \end{equation}
    Due to the rigidity of a coloured graph,
    each automorphism\footnote{There are 
    more than one definition of `automorphism of a coloured 
    graph'. The one used here is introduced in \cite{fullward}. In this setting,
    an automorphism of a coloured graph is a graph-morphism that
    preserves the colouring of the edges and bipartiteness of the vertex-set 
    in a strict way (not up to a permutation of colours as the factor 
    $1/3$ in the action of the quartic rank-$3$ model in \cite{On} suggests). 
    That is, edges of colour $a$ have to be mapped to 
    edges of colour $a$; black (resp. white) vertices to black (resp. white) vertices.} 
    $\pi$ of a connected graph $g$
    is determined by a permutation $\sigma$ 
    of the black (or white) vertices of $g$; we write 
    $\pi=\hat \sigma$ for such $\sigma \in \Sym(k(g))$.
    The action of the coloured automorphisms $\hat \sigma \in \Autc (g)\subset \Sym(k(g))$  on $M_{D\times k(g)}(\mathbb Z)$ is by permutation of the 
    matrix columns, $
    \yb^i \mapsto \yb^{\sigma(i)}$, $i=1,\ldots ,k$.
    As a notational remark, we will often write $k$ instead 
    of $k(g)$, as we just did, if the context is clear.
    
    \subsection{Three limit cases and examples}
    The previous lemma implies the Leibniz multivariable rule.  Before 
    elaborating on it, for the case of graph-group actions by
    automorphisms, $G(g)=\Autc(g)$,
    it will be helpful to exhibit this group action
    on a function $v_g$ in three limit cases, according to 
    the graph type of $g$. \par 
    Consider $\{\mathscr V(g),\Autc(g)\}_{g\in \hf}$, 
     an independent system
    graph-group actions,
    $\hf=\{h_1,\ldots,h_n\}\subset\mathsf{ G}_D$. 
    % A power $h^n$ of a graph denotes for us the disjoint union $h\amalg\cdots\amalg h$.  
    Let $g=h_1^{\alpha_1} \cdots h_n^{\alpha_n}$
    and let $U $ and $T$ graph-generated functionals by $J$ and 
    $L$, respectively,\footnote{
    Of course, one could just take the union of the 
    both spanning sets of graphs, if they do would not a priori coincide.} 
    being these subsets of the 
    the monoid generated by $\{h_1,\ldots,h_n\}$. 
    Then according to Lemma \ref{thm:Leibniz},
    \begin{align} \label{eq:wreath}
    \dervpar{(U\cdot T)}{g}=  \suml_{\Omega\in\Autc(g)}\suml_{\substack{(j,l)\in JL \subset \FM ^{\mtr{ab}}(\hf)\\ 
    j \cdot l \bcequal g}}\Omega \cdot  \langle u_{j} t_{l} \rangle_g\,.
    \end{align}
      This relation holds in a subdomain $\mathscr U(g)=\F_{D,k(g)} \subset \mathscr V(g)= M_{D\times k(g)(\Z)}$ to be justified later (see Sect. \ref{sec:GCtoTFT}):
          \begin{align*}
    % \boldsymbol{\Delta}_{D,k}=
    \mathcal{F}_{D,k}:=   
    \{(\yb^1,\ldots,\yb^k) \in M_{D\times k} (\Z) \,| &\, y^\alpha_{c}\neq y^\nu_c \mbox{ for all } 
    c=1,\ldots,D 
    \\  & \, \, 
      \,\mbox{ and } 
    \alpha,\nu=1,\ldots, k, \alpha\neq \nu\}\,.
    \end{align*}
    The canonical action of $\Autc (g)$ obviously 
    restricts to this set $\mathscr U(g)=\F_{D,k(g)}$. 
    Recalling that
    $\Autc(g)=
    \Autc(h_1^{  \alpha_1}\amalg h_2^{  \alpha_2}\amalg \ldots\amalg h_n^{  \alpha_n})= \prod_{i=1}^n \Autc (h_i) \wr \Sym(\alpha_i)$, first we elaborate on three simple cases: 
    \begin{itemize}
     \itemB \textit{\textbf{Case I}: if $n=1$}. Then $g=h^\alpha$,  any
      $\Omega\in\Autc (g)=\Autc (h) \wr \Sym(\alpha)$ is given by
      $\boldsymbol \sigma=(\sigma^1,\ldots,\sigma^\alpha)\in
      \Autc(h)^\alpha$ and $\mu\in\Sym (\alpha)$,
      yielding for eq. \eqref{eq:wreath}  
     \[
     \Omega\inv \cdot (v_g)(X_1,\ldots,X_\alpha)= v_g(\sigma^{1}X_{\mu(1)},\ldots,\sigma^{\alpha}X_{\mu(\alpha)}) \,.
     \]
     Here each $X_A=(\xb^1_A,\ldots,\xb^k_A)\in \F_{D,k(h)}$ and the action of the automorphism 
     group $\tau \in \Autc (h)$
     is given by  $\tau(X_A)=(\xb^{\tau(1)}_A,\ldots,\xb^{\tau (k)}_A)$.
    \itemB \textit{\textbf{Case II}: if $n\neq 1 $ but $\alpha_A=1$ for all
     $A=1,\ldots,n$}. In this case, $g=h_1h_2\cdots h_n$.  Then
     $\Autc(g)=\prod_i\Autc(h_i)\ni
     \Omega=(\sigma_1,\ldots,\sigma_n)$, which acts like
     \[
     \Omega\inv \cdot (v_g)(X^1,\ldots, X^n)=v_g(\sigma_1 (X^1),\ldots,\sigma_n  (X^n))\,.
     \]
      \itemB \textit{\textbf{Case III}: If $g=h^\alpha_1\cdots h_n^{\alpha_n}$, but 
     all automorphisms $\Autc(h_i)$ are trivial}. Then 
     \[
    \qquad \Autc(g)=\prod_i \Autc(h_i)\wr \Sym(\alpha_i)= 
     \Sym(\alpha_1)\times \cdots\times \Sym(\alpha_n)\ni (\mu_1,\ldots,\mu_n)\,.
     \]
     We use now multi-index notation
     $\boldsymbol \alpha=(\alpha_1,\ldots,\alpha_n)$  and
     $\boldsymbol \gamma=(\gamma_1,\ldots,\gamma_n)$ and abbreviate
     $u_{\gamma_1,\ldots,\gamma_n}= u_{h_1^{\gamma_1}\cdots
       h_n^{\gamma_n}}$, and similarly for
     $t_{\gamma_1,\ldots,\gamma_n}$.  One can rewrite then
 
     \begin{align*}
     \qquad\dervpar{\,(U\cdot T)}{\,g} &=   \suml_{\mu_1 \in \Sym(\alpha_1)} \cdots  \suml_{\mu_n \in \Sym(\alpha_n)} 
     \suml_{\substack{(\gamma_i,\nu_i)\\ \gamma_i,\nu_i\geq 0 \\ \gamma_i+\nu_i=\alpha_i \\ i=1,\ldots,n}} (\mu_1,\ldots, \mu_n)^*\langle u_{\gamma_1,\ldots,\gamma_n} t_{\nu_1,\ldots,\nu_n}\rangle_{g}
       \\
       &= \suml_{\mu_1 \in \Sym(\alpha_1)} \cdots  \suml_{\mu_n \in \Sym(\alpha_n)}
      \\ 
      & \hspace{2.6cm}\suml_{\substack{ 0\leq \gamma_i \leq \alpha_i \\ 
     i=1,\ldots,n }} (\mu_1,\ldots , \mu_n)^*\langle u_{\gamma_1,\ldots,\gamma_n} 
     t_{\alpha_1-\gamma_1,\ldots,\alpha_n-\gamma_n}\rangle_{g}  
    %   &= \\
     \end{align*}
    in multi-index notation as
      \begin{align*}
     \dervpar{\,(U\cdot T)}{\,g} &= \suml_{\boldsymbol \mu \in \prod_{i_1}^n\Sym(\alpha_i)} \,\,
     \suml_{\boldsymbol \gamma \leq \boldsymbol \alpha } (\boldsymbol \mu )^*\langle u_{\boldsymbol\gamma} t_{\boldsymbol \alpha -\boldsymbol\gamma}\rangle_{g} \,.
     \end{align*}
      For constant functions $u_{\gamma_1,\ldots,\gamma_n}, t_{\gamma_1,\ldots,\gamma_n}$, this should reduce to the  multivariable product formula. Indeed,
      \begin{align*}
       \dervpar{\,(U\cdot T)}{\,g}
      &=\suml_{\boldsymbol \gamma \leq \boldsymbol \alpha }  \binom{\alpha_1}{\gamma_1}
      \cdots \binom{\alpha_n}{\gamma_n}  u_{\boldsymbol\gamma} t_{\boldsymbol \alpha -\boldsymbol\gamma} 
     \end{align*}
     which is just the Leibniz rule eq. \eqref{eq:Leibniz}. 
    % %  \naranja{Figure out, whether the requirement of having constant functions is
    %  not overly restricting (more general formula?)}
    
    \end{itemize}
    Conveniently, lower-case (super)indices ($i=1,\ldots,n$) of momenta label the
    graph type, whereas upper-case (sub)indices indicate the number of copy
    ($A=1,\ldots,\alpha_i$) of the $i$-th graph type.  \par We describe now
    the action of $\Autc(g)$ on a general function
    $v_g:M_{D\times k(g)} (\C)$.   $v_g:M_{D\times k(g)} (\Z)\to \C$. \par Let
    $\Xb=((X^1_1,\ldots,X^1_{\alpha_1})\ldots(X^n_1,\ldots,X^n_{\alpha_n}))\in M_{D\times k(g)} (\Z)$,
    being $X_A^i\in M_{D\times k(h_i)} (\Z)$ the momentum of the $A$-th copy of
    $h_i$, for $A=1,\ldots,\alpha_i$. Picking an element
    $\Omega=( \boldsymbol\sigma,\boldsymbol\mu)\in\Autc(g)$, with
    \begin{align}
                 (\boldsymbol\sigma_1,\ldots,\boldsymbol \sigma_n) &\in \prod_{i=1}^n\Aut(h_i)^{\alpha_{i}} \,,  \quad                                                                                                     \boldsymbol\sigma_i=(\sigma_i^1,\ldots,\sigma_i^{\alpha_i})\,,
    \\
    \boldsymbol \mu=(\mu_1,\ldots,\mu_n)&\in \prod_{i=1}^n\Sym(\alpha_i) \,,
    \end{align}
    the following holds:
    \begin{equation}
     (\Omega\inv \cdot v_g)(X_A^i) % v_g \big((\sigma_{i}^{\mu\inv_i(A)})  X_{ A }^i)_{A}^i\big)
     =v_g \big((\sigma_{i}^{A}  X_{\mu_i(A)}^i)_{A}^i\big)
     \,,
    \label{eq:mainformula}
    \end{equation}
    which is short-hand notation for
    \begin{align*}
      &\Omega \inv \cdot v_g(\Xb) \\ & =  v_g \big((\sigma^{1}_{1}
      X^1_{\mu_1(1)},\sigma^{2}_{1}X^1_{\mu_1(2)}
      \ldots,\sigma^{\alpha_1}_{1}
      X^1_{\mu_1(\alpha_1)}),  
      \ldots,
      (\sigma_n^{1}X^n_{\mu_n(1)},\ldots,\sigma_n^{\alpha_n}
      X^n_{\mu_n(\alpha_n)})\big) \, .\end{align*} 
      Usually, it is summed all over
    $\Omega\in\Autc(g)$, which being a group, allows us to choose
   whether we put the inverse in \eqref{eq:mainformula}.
    % \end{rem}

    \begin{ex}
    For   $e=\meloncik$, $f=\vuno$ and $g=\kthree$,
     let the subsets
    \[
    H=\{\varnothing, e^2, f^2 
    \} \qquad \and\qquad L=\{  e^3f,e^2 g ,e^2 f^2 g \}
    \]
    span the functionals $U = \sum_{h\in H} u_h h$ and
    $ T=\sum_{l\in L} t_l l$, and consider their product $V=U\cdot T$.
    According to the action \eqref{eq:mainformula}, one has the
    following formulae:
    \begin{itemize}
      \itemB For $b=e^5f$, \vspace{-.5cm}
     \begin{align}
     \dervpar{V}{b} & = \sum_{\Omega\in\Autc(b)} \Omega \cdot  \langle u_h t_l \, \rangle_b  \nonumber\\
      & =\sum_{\Omega\in\Autc(e^5f)} \Omega \cdot \langle u_{e^2} t_{e^3f} \, \rangle_b  \nonumber\\
      & = \sum_{(\sigma,\epsilon)\in{\Sym(5)\times \Z_2}} (\sigma,\epsilon)\cdot (u_{e^2} t_{e^3f})\,,
     \end{align}
     since
     $\Autc(e^5f)=\Autc (\meloncik) \wr \Sym( 5) \times \Autc (\vuno)
     \wr \Sym (1)=\{1\}\wr \Sym( 5) \times \Z_2 \wr \{1\}$. Also the
     ordering $\langle \balita \rangle_b$ is trivial. This in turn means that for the
     five momenta $X_A \in\F_{1,D=3}=\Z^3$ of $e^5$ and the momentum
     $Z=(\zb^1,\zb^2)\in \F_{2,3}$,
    \[
    \dervpar{V}{e^5|f}(X_A,Z)= \sum_{(\sigma,\epsilon)\in{\Sym(5)\times \Z_2}} (u_{e^2} t_{e^3| f})({ X_{\sigma (A)}},\zb^{\epsilon(1)},\zb^{\epsilon(2)})
    \]
    in abstract notation, or displaying the graphs:
    \begin{align*}\qquad\quad
    \dervpar{V}{\meloncik^5|\vuno}(X_A,Z)&= \suml_{(\sigma,\epsilon)\in{\Sym(5)\times \Z_2}} 
    u_{\meloncik|\meloncik} (X_{\sigma(1)},X_{\sigma(2)}) \\
    & \qquad \times t_{\meloncik|\meloncik|\meloncik| \vunito}(X_{\sigma(3)},X_{\sigma(4)},X_{\sigma(5)},\zb^{\epsilon(1)},\zb^{\epsilon(2)})\,.
    \end{align*}
     \itemB For $b'=e^2f^2g$, 
     \begin{align}
     \dervpar{V}{b'} & = \sum_{\Omega\in\Autc(e^2f^2g)} \Omega \cdot \langle u_\varnothing t_{e^2f^2g} \, \rangle + \Omega \cdot \langle  u_{f^2}t_{e^2 g} \,\rangle \nonumber \,.
    %   & =\sum_{\Omega\in\Autc(e^5f)} \Omega \cdot \langle u_{e^2} t_{e^3f} \,\big\rangle  \nonumber\\
    %   & = \sum_{(\sigma,\epsilon)\in{\Sym(5)\times \Z_2}} (\sigma,\epsilon)^*(u_{e^2} t_{e^3f})\,, \nonumber
     \end{align}
     For $X_A$ momenta of $e^2$, $Z_A$ momenta of $f^2$ and total
     momentum $\Xb= (X_1,X_2,Z_1,Z_2,W)$, one sums over elements
    $ \Omega=(\epsilon,(\sigma^1,\sigma^2;\mu),\tau ) \in \Autc(e^2f^2g)=\Sym (2)\times(\Z_2 \wr \Sym(2))\times \Z_3$, which yields for $(\partial V/\partial b')(\Xb)=\big(\partial V/ \partial(\meloncik^2|\vuno^2|\kthree)\big)(\Xb)$
    the expression
     \allowdisplaybreaks[1]
     \begin{align*}
     &  \suml_{\Omega } 
     \big[\Omega \cdot ( u_\varnothing t_{e^2f^2g} ) + 
     \Omega \cdot ( [(13)(24)]^* u_{f^2}t_{e^2 g}  ) \big](\Xb) \nonumber  \\
     & =  u_\varnothing \cdot\bigg\{\suml_{ (\epsilon,(\sigma^1,\sigma^2;\mu),\tau )} 
     [ t_{e^2f^2g}(X_{\epsilon(1)},X_{\epsilon(2)}),({\sigma^{1} Z_{\mu(1)}},\sigma^2 Z_{\mu(2)} ,\tau (W)]\bigg\} \\
    &\qquad\quad + \suml_{ (\epsilon,(\sigma^1,\sigma^2;\mu),\tau )}  u_{f^2} ({\sigma^{1} Z_{\mu(1)}},\sigma^2 Z_{\mu(2)})  t_{e^2 g}
    (X_{\epsilon(1)},X_{\epsilon(2)}) ,\tau (W) ))\nonumber   
    %   & =\sum_{\Omega\in\Autc(e^5f)} \Omega \cdot \langle u_{e^2} t_{e^3f} \,\big\rangle  \nonumber\\
    %   & = \sum_{(\sigma,\epsilon)\in{\Sym(5)\times \Z_2}} (\sigma,\epsilon)^*(u_{e^2} t_{e^3f})\,, \nonumber
    %  \end{align}
    %   \begin{align*}
    %  \dervpar{V}{(\meloncik^2|\vuno^2|\kthree)}(\Xb)&  =  u_\varnothing
     \\
     &= u_\varnothing\cdot \bigg\{\suml_{ (\epsilon,(\sigma^1,\sigma^2;\mu),\tau )} 
     [ t_{\meloncik^2|\vunito^2|\kthreecik}(X_{\epsilon(1)},X_{\epsilon(2)}),({\sigma^{1} Z_{\mu(1)}},\sigma^2 Z_{\mu(2)}) ,\tau (W)] \bigg\}\\
    &\qquad\quad+\suml_{ (\epsilon,(\sigma^1,\sigma^2;\mu),\tau )}  u_{\vunito^2} ({\sigma^{1} Z_{\mu(1)}},\sigma^2 Z_{\mu(2)}) \times  t_{\meloncik^2 | \kthreecik}
    (X_{\epsilon(1)},X_{\epsilon(2)}) ,\tau (W) )\nonumber  \,.
     \end{align*}
      \allowdisplaybreaks[0]
    One should still insert the explicit momenta
    $Z_A=(\zb^{1}_A,\zb^{2}_A)\in\F_{2,3}$, 
    $W=(\mathbf{w}^1,\mathbf{w}^2,\mathbf{w}^{3}) 
    \in \F_{3,3}$,
    $\tau( W)=
     (\mathbf{w}^{\tau(1)},\mathbf{w}^{\tau(2)},\mathbf{w}^{\tau(3)})$, 
     and $\sigma^{A}Z_{\mu(A)}= (\zb^{\sigma^{A}(1)}_{\mu(A)},\zb^{\sigma^{A}(2)}_{\mu(A)})$.
     
    \end{itemize}
    
    \end{ex}

    \section{Tensor models}\label{sec:back}
    
    In this section, we implement the graph calculus for TFT.     

    \subsection{Tensor Field Theory}

The main idea to use a 
the Ward-identity \cite{DGMR} to decouple the Schwinger-Dyson equations (at a planar sector)
and to obtain a master integral equation for the 2-point functions for matrix
models \cite{gw12} can been extended to the TFT  setting.
Some progress along these lines has been made for complex tensor field theory 
and consists in the further study \cite{fullward} of the Ward-Takahashi identity of Ousmane-Samary \cite{DineWard} in order to descend the SDE tower \cite{SDE}
and eventually find closed equations. This led lately to the 
large-$N$ limit \cite{cinco} of the connected-$\partial$ SDE.
    We treat a TFT as inspired by 
    group field theory \cite{asymptotic,3Dbeta,Carrozza:2012uv,wgauge}.

    Unlike matrix models, where there is a canonical way of forming a scalar, for 
    tensor models a specific trace $\Tr_\B$, indicating how to contract the indices,
    should be specified. These traces $\Tr_\B$ are indexed by $D$-coloured graphs $\B$,
    where $D$ is the rank of the tensors $\phi_{x_1\ldots x_D}, \bar\phi_{x_1\ldots x_D}$. 

    The graphical representation of these traces derives from to the independence
    of the imposed transformation rules under the action of $\uni(N)$ 
    on the spaces corresponding to each index $x_a$ of 
    $\phi_{x_1\ldots x_D}, \bar\phi_{x_1\ldots x_D}$, for $a=1,
    \ldots,D$, deemed \textit{colouring.} That is to say, to form 
invariants only indices of identical colour can be contracted.
    
    Therefore, a trace corresponding to a quartic interaction would be, say, 
    $\vuno$ formed by colour-wise contracting
    the indices  with deltas, as follows:\[\Tr_{\vunito}(\phi,\bar\phi)= \sum_{\mathbf x,\mathbf y}
    \phi_{x_1y_2 x_3}  \bar\phi_{x_1y_2y_3}
    \phi_{y_1y_2 y_3}  \bar\phi_{y_1x_2x_3}
    \,.\] 
    The actual index-set of the tensors is $\{1,\ldots, N\}$,
    but thinking of $N$ as a large integer, we typically write
    these sums over $\mathbb N$ or $\Z$, or omit the domains in the sums. 
    Although orthogonal groups \cite{On}, compact symplectic groups \cite{CarrozzaSp} and mixed symmetries
    \cite{reviewTanasa} define other classes of tensor models,
    we restrict our discussion to the $\uni(N)$-tensor models we just introduced.
    
    A tensor model is thus determined by a dimension $D$ (the rank of the tensors)
    and an action $S[\phi,\bar\phi]$ given by a finite sum of traces 
    indexed by connected $D$-coloured graphs. 
    The partition function is given by 
    \begin{align} \label{C_measure}
 Z[J,\bar J]&=Z[0,0]\frac{\int\Df[\phi, \bar\phi] \,\ee^{\Tr{(\bar J\phi)}+
\Tr{(\bar\phi J)}-N^{D-1}S[\phi,\bar\phi]}}{\int\Df[\phi, \bar\phi]\,\ee^{-N^{D-1} S[\phi,\bar\phi]}}\,, \\
 \Df[\phi, \bar\phi]&:= \prod\limits_{\mathbf {x} } N^{D-1}\frac{\dif\varphi_{\mathbf x} \dif\bar\phi_{\mathbf{x}}}{2\pi \ii} \,. \nonumber
\end{align} 
Its logarithm, $W\jj=\log Z\jj$ is the \textit{free energy} and
generates the connected correlation functions, which as pointed out
before, are classified by boundary 
(possibly disconnected $D$-coloured) graphs. 
    
    \subsection{From graph calculus to tensor field theory} \label{sec:GCtoTFT}
    For a deeper exposition and motivation of the terminology and proofs of the results  
    exposed in this section, we refer to 
    \cite{fullward}.
     
    Some objects of interest in tensor models are functionals
    generated by graphs (e.g. the free energy). 
    By this we mean 
    expansions in graphs with functions (or distributions) as coefficients.
    For a graph $g\in H \subset \mathsf G_D$ we recall that
    $2k(g)$ denotes its number of vertices. One is interested in collections
    of functions 
    \[\{u_g: \Z^{D\times k(g)}\to \C\}_{g\in H}\] and in
    their generating functionals
    \[U[J,\bJ] = \sum_{g\in H \subset \mathsf G_D} u_{g} \star
      \J(g)\,, \qquad \where u_{g} \star \J(g) := \sum_{\Xb\in
        \Z^{D\cdot k}}(u_g(\Xb)) \J(g) (\Xb)\,. \] Here
    $\J(g)(\Xb)=\prod_{\alpha=1}^{k(g)} J_{\xb^\alpha}
    \bJ_{\yb^\alpha}$, where $\{\yb^\alpha\}_{\alpha}$ are determined
    by $g$ and $\Xb$ through  $g_*(\Xb)=(\yb^1,\ldots,\yb^{k(g)})$.  The
    induced map $g_*$ is defined as follows.  The $D$-tuple $\xb^\alpha$
    (resp. the $\yb^\alpha$) indexes white (resp. black) vertices in a
    graph. Then $g_*: M_{D\times k(g)} (\Z) \to M_{D\times k(g)} (\Z)$
    is given by
    $\Xb=(\xb^1,\ldots,\xb^{k(g)})\mapsto
    g_*(\Xb)=(\yb^1,\ldots,\yb^{k})$, where $y^\alpha_c =x^\nu_c$ (for
    $\alpha=1,\ldots,k$) if and only in the graph $g$ there exists a
    $c$-coloured edge starting at $\xb^\alpha$ and ending at
    $\yb^\nu$. As before, $\Xb$ 
    is called momentum,  but also each one of these arguments $\xb^\alpha$
    is referred to as (\textit{entering})
    \textit{momentum} of the white vertex $J_{\xb^\alpha}$. Similarly
    $\yb^\nu$ is the (\textit{outgoing}) \textit{momentum} at the
    black vertex $\bar J_{\yb^\nu}$; the terminology relies on Figure \ref{fig:combinatorics}.
    Although an ordering of the vertices is assumed,
    notice that $\J(g)$ is independent of it. 
    
    \begin{figure}\centering
\begin{minipage}{.5\textwidth}\centering
\includegraphics[width=.634\textwidth]{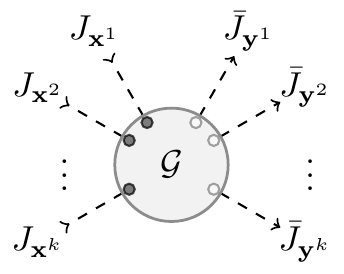}
\end{minipage} 
% \subfloat[An example with $\B=K_\mtr{c}(3,3)$]
\centering
\begin{minipage}{.46\textwidth}
{
\includegraphics[width=.936\linewidth]{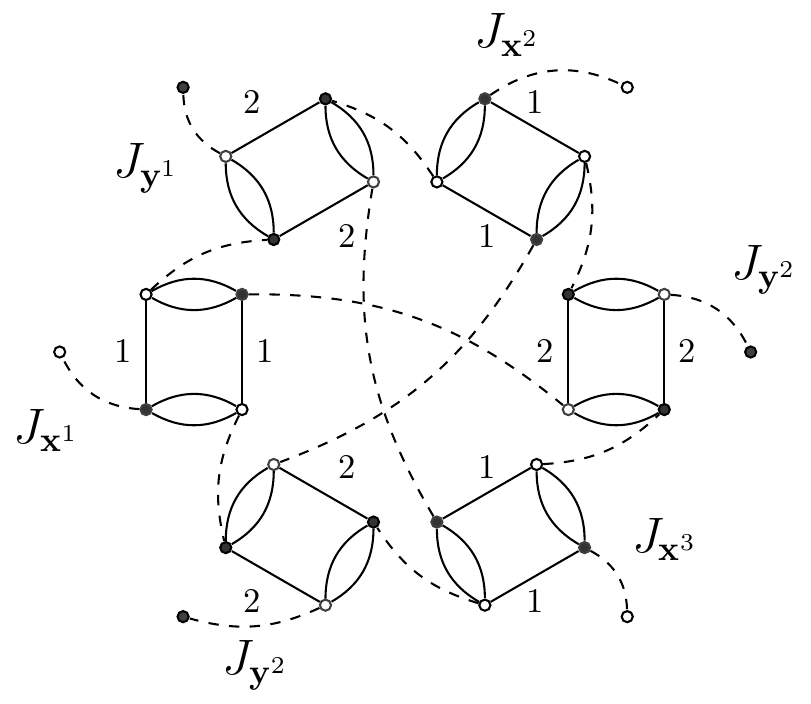}}\end{minipage}
\caption[Induced map $g_*$ and momenta transmission]{\textit{Left.} On the interpretation of the induced map $g_*$.
One takes any representing graph $\G$ such that $\partial \G=g$. 
For the quartic melodic model (pillow-interactions) we can do so
because it has been proven in \cite{fullward} that the 
spectrum of boundary states is full, which is to say, the boundary graphs
are all of $\mathsf G_D$. Assuming 
an ordering on the white and on the black vertices, the 
components $\yb^\alpha$ of 
$g_*(\xb^1,\ldots,\xb^k)=(\yb^1,\ldots,\yb^k)$ 
are determined by this picture, seeing the $\xb^\mu$'s
as independent momenta entering the graph, 
with output $g_*(\xb^1,\ldots,\xb^k)$. \textit{Right}. 
If $g= \kthree$, and we choose the numeration
 $ x^\alpha_1=y^\alpha_1$ for $\alpha=1,2,3$;
 there  that $\G$ satisfying $\partial \G=\kthree$ is shown.
 The map 
 $(\kthree)_*(\xb^1,\xb^2,\xb^3)=(\yb^1,\yb^2,\yb^3)= ( (x^1_1, x_2^2 ,x_3^3)^t,
  (x^2_1 , x_2^3 , x_3^2)^t, ( x^3_1 , x_2^1 , x_3^2 )^t) $
  is determined by following the $0a$ momenta lines
 for each colour $a$
 \label{fig:combinatorics}}
\end{figure} 
%       \begin{figure} 
%             \includegraphics[width=.34\textwidth]{figures_tables/combinatorics}
%             \caption{How a closed graph $g$ induces a 
%             induces the map $g_*:(\xb^1,\ldots,\xb^{k(g)})\mapsto (\yb^1,\ldots,\yb^{k(g)})$ is based on this figure. In tensor models, 
%             $g$ is a boundary graph, that is, $g=\partial 
%             \G$ for a Feynman graph of a certain model. The 
%             boundary map $\partial $ encodes momentum transmission in 
%             each one of the $D$ colours. }
%             \label{fig:combinatorics}
%         \end{figure}
 Consider the canonical system of graph-group actions
  introduced in eq. \eqref{eq:canonicalSGGA} \[ 
    \{u_g,   M_{D\times k(g)}(\mathbb Z),  \Autc (g) \}_{g\in H }
\,.\] 
 A very important domain where the graph derivatives shall be defined 
 is the subspace 
 $  \mathcal F_{D,k(g) } $ of $ \Z^{D\cdot k(g)}\simeq  M_{D\times k(g)}(\mathbb Z)$
  consisting of points outside all the
    coloured diagonals, i.e.
    \begin{align*}
    % \boldsymbol{\Delta}_{D,k}=
    \mathcal{F}_{D,k}:=   
    \{(\yb^1,\ldots,\yb^k) \in M_{D\times k} (\Z) \,| &\, y^\alpha_{c}\neq y^\nu_c \mbox{ for all } 
    c=1,\ldots,D 
    \\  & \, \, 
      \,\mbox{ and } 
    \alpha,\nu=1,\ldots, k, \alpha\neq \nu\}\,.
    \end{align*}
% %     This space is essential for us. 
For a \textit{connected} graph $g$, 
the elements of $\Autc(g)$ are a lift $\hat \sigma$ of an
element $\sigma$ of the symmetric group $\Sym{(k(g))}$ (see 
discussion below eq. \eqref{eq:canonicalSGGA}, or \cite{fullward} for details).  Defining
    \begin{equation}
     \label{eq:PartialGraphDerivative}
    \dervpar{U\jj}{\,g(\Xb)}= \prod_{\alpha=1}^{k(g)}\dervfunc{}{J_{\xb^\alpha}}\dervfunc{}{\bJ_{\yb^\alpha}} {U\jj}\bigg|_{J=\bar J=0}\,,\quad  \Xb=(\xb^1,\ldots,\xb^{k(g)}) \in \mathcal F_{D,k(g)}\,,
    \end{equation}
    one can give this derivative the 
    the meaning of eq. \eqref{eq:GroupAction} as a permutation of the
    arguments of $u_g$, that is  
    \begin{align}
    \dervpar{U\jj}{\, g (\mathbf X)}&= 
    \suml_{\hat\sigma \in \Autc( g)} (\sigma \cdot u_g )(\mathbf{X}),\label{eq:sumovergroup}
    ~\where~ \\
    (\sigma \cdot u_g  )(\xb^1,\ldots,\xb^{k(g)}):\!&= 
    u_g(\xb^{\sigma\inv(1)},\ldots,\xb^{\si\inv(k(g))})\,, \quad \label{eq:dropinverse}
    \end{align}
    for all $\Xb=(\xb^1,\ldots,\xb^{k(g)}) \in \mathcal F_{D,k(g)}  $. 
    This statement is \cite[Lem. 4.1]{fullward}.
    
    \begin{rem}
      Notice that in eq. \eqref{eq:sumovergroup} it is summed over the
      group $\Autc(g)$. We are therefore entitled to drop the inverse in the RHS in
      \eqref{eq:dropinverse}. However, if the sum is not over all the
      group, we will keep the right `orientation' of the action, for
      the convention in single terms $(\sigma \cdot u_g )$ is
      important 
      in that case.
    \end{rem}

    For graphs $h,g$, functions $u_g$ and functionals $U\jj$ we
    abbreviate the usual notation as follows:  \begin{equation} \label{eq:abreviar} U=U\jj\,,\quad u_g
      g:=u_{g} \star \J(g)\, \quad\and \quad g h:=\J(g)\J(h)\,,
    \end{equation}
    and treat the latter as a product of graphs. This product 
    is not considered commutative, since the star $\star$ 
    implies an ordering in the arguments of a function $u_{gh}$,
    which need not satisfy $u_{gh}=u_{hg}$.
    % however $u_{gh} gh=u_{hg} hg$ is imposed in order for the 
    % arguments of the functions to match the order of the product.
    Now we exhibit the relation to the generating functional of
    group actions.  With the product defined above, consider the
    functional $V^{\leq n}$ that generates the \textit{connected} correlation 
    functions of TFT 
    \[
     V^{\leq n}=\sum_{\B\in \mathsf G_D^n} G_\B \,\B\,.
    \]
     (with $G_\varnothing=0$.) Here $n$ is a large integer, and $\mathsf G_{D}^n$ is a finite set
    of coloured graphs whose elements $\mtc D$ satisfy
    \begin{equation}
    \label{eq:trunca}
    \#(\mbox{vertices of } 
    \mtc D
    )   \leq 2n\,.
    \end{equation}
    In tensor models, the subindices of the functions corresponding to
    the graph $\A\amalg \B$ (also written as juxtaposition
    $\A\B$) are rather denoted by $G_{\A | \B}$.  These particular functions 
    $G_{\mathcal D}$ satisfy invariance under $\Autc(\mathcal D)$;
    this means invariance under the 
    $\Autc$-groups of the connected components of $\mathcal D$ and, for any isomorphic connected components $\A$ and $\B$  of $\mathcal D$,
    $G_{\ldots|\A|\ldots |\B|\ldots } =G_{\ldots|\B|\ldots |\A|\ldots
    }$.
   
    The truncation \eqref{eq:trunca} would declare vanishing 
    all the floors above the $n$-th floor of the SDE-tower, 
    but we can increase $n$ at desired accuracy. It bounds
    any graph $\mathcal D$ appearing in an non-identically 
    vanishing correlator $G_\mtc{D}$ to
    have $n$ components at most.  In rank $3$,
    since the canonical (optimal in number of vertices) 3-coloured 
    graph   of genus $g$ has $4g+2$ vertices, this truncation bounds
    the genus through $ 2g+1 \leq n$, making higher-genera boundary states
    vanish.  We write the infinite sums keeping in mind that we mean
    their $n\to \infty$ limit.
%     \begin{equation}
%     \label{eq:limit} 
%      \#(\mbox{Vertices of } \mtc D) \ll 2n.
%     \end{equation}
    The coloured Borel transform of $V^\infty=\lim_{n\to \infty} V^{\leq n} $ 
    is called the \textit{free energy}:
    \begin{equation}
    \label{eq:freeenergy} 
     W=\mtr{B_c}(V^\infty)=\sum_{\mtc D\in \mathsf G_D} 
     \frac{1}{|\Autc(\mtc D)|} G_{\mtc{D}}\,\mtc D\,.
    \end{equation}
    This equation holds in `$N=1$-units' and this assumption is innocuous within
    the scope of this article. 
    However, if one plans to proceed perturbatively in $1/N$,
    the realistic case that drops this simplification ought to be addressed.
    Adding the power counting conjectured in \cite{cinco} 
    that scales $G_{\mathcal{D}} \to N^{\gamma(\mathcal D)}G_{\mathcal{D}}$,
    where $\gamma(\mathcal D)$ is certain  factor already determined for 
    the 2-pt and 4-pt functions, 
    would help analysing the convergence of $W$ (see Sect. \ref{sec:Outlook}).\par 
    This free energy functional (but not only this) corresponds to a  
    system of graph-group actions that has the following constituents:
    \begin{itemize}
      \itemB For each graph $\B$, $\mathscr V(\B)=M_{D\times k(\B)}(\Z)$.
      \itemB For a disconnected graph $\mtc D=
      \amalg_p \B_p$ one has
     \[\mathscr V(\mathcal D)=M_{D\times k(\mathcal D)}(\Z) =
     \prod_p \mathscr V(\B_p)\,.\]
         \itemB There is an action of each $\Autc(\mathcal{ \B}_p)$ on
     $M_{D\times k(\B_p)}(\Z)$. Since
     $\sigma \in \Autc(\mathcal{ \B}_p)\subset \Sym(k(\B_p))$,
     precomposition by a function by $\sigma$  permuting the columns
     of $M_{D\times k(\B_p)}(\Z) $  gives this action.
     \end{itemize}
      For connected graphs $\B_i$ and $\mathcal B_j$, 
     the following holds \cite[Lemma 4]{fullward}:
     \[
     \dervfunc{\B_i}{\B_j}=\delta_{ij} \Autc(\B_i)   \,,
     \]
     on the domain $\F_{D,k(\B_j)} \subset M_{D,k(\B_j)}$. For general disconnected graphs $\mathcal D=\B_1^{\alpha_1}\amalg \cdots\amalg \B_m^{\alpha_m}$
       \begin{equation} \label{eq:groupactionindependenceAut}
      \dervfunc{\mathcal D}{\mathcal D} = \Autc(\mathcal D) = \Autc(\B_1)\wr \Sym(\alpha_1)
      \times \cdots \times \Autc(\B_n)\wr
      \Sym(\alpha_m) \,
    \end{equation}
    on $\mathcal F_{D,k(\mtc D)}\subset M_{D\times k(\mathcal D)}(\Z)  $. 
    One can then operate with functionals, now without needing to evaluate the
    sources at 0. That is,
    \[
    \dervfunc{U\jj}{g(\Xb)} = \prod_{\alpha=1}^{k(g)}\dervfunc{}{J_{\xb^\alpha}}\dervfunc{}{\bJ_{\yb^\alpha}}  U\,.
    \]
        This object is, unlike the function $\partial U/\partial g$, a
    functional that can be graph-derived again (without getting something
    a trivial result) and get again a functional generated by graphs. Contrast this with eq. \eqref{eq:PartialGraphDerivative}, whose result is a function. 
%     We now derive
%     equations for the free energy coefficients.
    \par
    
    Another important functional in the next derivation is the so-called $Y$-term
    that emerged in the derivation of the Ward-Takahashi identity
    \cite{fullward} and which encodes all the pertinent insertions 
    of $2p$ point functions into $2p-2$ point functions for all $p$.
    
    The  expression to order six is given in
    \cite[Lemma 4.1]{SDE}, but this article only will evoke 
    the $Y$-term up to order four, located in Appendix \ref{sec:App}.  For this paper, it is sufficient 
    to additionally know the expression 
    \[
    Y\hp c_{x}= \sum_{\B \in \mathsf{G}_D} \mathfrak{f}\hp c_{\B,x} \,\B
    \,,
    \]
    where $c\in \{1,\ldots, D\}$ is a colour and $x\in\Z$.  It is
    important to notice that unlike $W$ (for which we set
    $G_\varnothing=0$), there is a non-vanishing constant term
    $\mathfrak{f}\hp c_{\varnothing,x}$ in $Y\hp c_x$. Each 
    function coefficient $\mathfrak{f}\hp c_{\B,x}$ of a 
    graph $\B$ denotes a triple propagator contraction of vertices from $\vcv$ 
    with vertices of the graph $\mtc C$ having $2k(\mtc C)=2+2 k(\B)$ vertices and such that the whole contraction's boundary is $\B$.
    Although $G_{\mtc C}$ is symmetric with respect to 
    action of $\Autc (\mtc C)$, the resulting insertion 
    need not to be $\Autc (\mtc B)$-symmetric. 
     \par 
     In order to derive any of the $2p$-point SDE\footnote{Here we mean the melonic quartic model. The
     bound on the vertices is model dependent and justified
     in the statement of Theorem \ref{thm:SDEDisconnected}.}  
     we shall employ the graph calculus  $\mathscr C(\hf)$  with variables
     $\hf = \{$connected, closed, $D$-coloured graphs with $\leq 2(p+1)$ vertices$\}$, 
     being the system of graph-group actions 
     the canonical one given by 
     automorphism groups (see \eqref{eq:canonicalSGGA} above for details).

    \section{Disconnected-boundary Schwinger-Dyson equations} \label{sec:teorema}
    
    The next section introduces the model whose SDE are found in Section \ref{sec:main}.
    
    \subsection{The quartic melonic tensor field theory }
  The $\phi^4_{D,\mtr{m}}$-theory is the model with quartic interaction vertices
    $\Sint[\phi,\bar\phi]=\lambda \sum_{a=1}^D{\Tr_{ 
        V_a}}(\phi,\bar\phi)$. These vertices are sometimes
    called pillows, since the graphs they correspond to 
    have that appearance:
    \begin{equation} \label{eq:DefVa}
    V_a=\logo{4}{VaD}{7.9}{4} \qquad a=1,\ldots,D\,.
     \end{equation}
     We analyse this theory with an abstract Laplacian
 $E:\Z^D\to \re_{\geq 0}$ as propagator, $S_0[\phi,\bar\phi]=\Tr_{\meloncik}(\bar\phi,E \phi)=\sum_\xb \bar\phi_\xb E_\xb \phi_\xb$,
 assumed here to satisfy the following technical assumption:
 for each colour $c$, the difference 
    \begin{equation}
\label{eq:kinetic}
    E(t_c,s_c):=E_{p_1\ldots t_c\ldots p_D} - E_{p_1\ldots s_c\ldots p_D}\, 
    \end{equation}
    is independent of the fixed momenta in colours different from $c$.
    Such kind of technical conditions permit to exploit the 
    Ward-Identity and are common. In matrix field theory 
    \cite[Thm. 2.3]{gw12}, this is analogous to the assumed injectivity of $n\mapsto E_n$
    for the generalised matrix Laplacian $E_{\mtr{matrix}}=\mtr{diag}{(E_n)_{n\in \N}}$ there. 
    \subsection{Main result}\label{sec:main}
    We prepare\footnote{We come back to the usual notation for graphs or `bubbles'.}
    now some notions and notations needed for 
    the formulation of the main result. 
    Let $\R$ be a connected graph, $2r$ its number of vertices,
    and $X\in\F_{r,D}$. 
    Given a colour $c$ and a numeration $w^1,\ldots, w^r$ of
    the black vertices of $\R$ we set, for any $1\leq \rho \leq r$,
    \[\mathrm{Br}\hp{2}(\R,\rho,c):=  
      \{\tau\in\{1,\ldots,r\} \,|\, \varsigma_c(\mtc R,w^\rho,w^\tau)
      \mbox{ is disconnected}\}\,. \] Thus, $\tau$ belonging to this
    set indexes the vertices $w^\tau$ at which a swap of the colour
    $c$ edge at $w^\tau$ and $w^\rho$ disconnects $\mtc R$. 
    In other words, $\mathrm{Br}\hp{2}(\R,\rho,c)$ indexes $c$-coloured 
    edges of $\mathcal R$ that, paired with the 
    $c$-coloured edge at $w^\rho$, form a $2$-bridge. 
    If $\mathrm{Br}\hp{2}(\R,\rho,c)\neq \emptyset $, then 
    $\R$ must be, in particular, 3-edge connected.  
    
   \begin{ex}
   The coloured utility  graph $\kthreecik$ satisfies $\mathrm{Br}\hp{2}(\kthreecik,\rho,c)=\emptyset$
for any colour $c$ and vertex $\rho$ (`no edge-swap separates it'). On the other hand, if  $\{\rho,\mu\}=\{1,2\}$ label 
the two black vertices of the pillow $\vuno$, then \[\mathrm{Br}\hp{2}(\vuno,\rho,c=1)=\{\mu\} \mbox{ and }\mathrm{Br}\hp{2}(\vuno,\rho,c)=\emptyset \mbox{ for } c=2,3\,. \] 
    \end{ex}

 \begin{preparation*}
     Let $\mathcal D$ be a $D$-coloured 
      graph with $2d$ vertices. By \cite[Thm. 1]{fullward} $\mathcal D$ is 
      a boundary graph of the 
      $\phi^4_{D,\mtr{m}}$-theory and $G_{\mathcal D} \equiv \hspace{-2.3ex}/\hspace{1ex} 0 $.  Given any
      $\Xb\in\mathcal F_{d,D}$, we select an outgoing momentum $\sb=\yb^\beta$
      listed in
     \begin{equation}\label{eq:listmomemta}
      (\yb^1,\ldots,\yb^\beta,\ldots,\yb^d)=\mtc D_* (\Xb)\in \mathcal F_{d,D}\,.
     \end{equation}
     This $\sb$ determines both a connected graph component $\R$ of
     $\mtc D$, and an $r$-tuple $X_0\in \mathcal F_{r,D}$ of momenta,
     being $2r$ the number of vertices in $\R$, by asking that $\sb$
     appears listed in the $r$-tuple $\mtc R_*(X_0)$, particularly.  
     Different choices of the distinguished momentum variable $\sb$ ---say $\sb=\yb^{\beta_1}$
     and $\sb=\yb^{\beta_2}$---
     lead to a different SDE 
     whenever the respective $\beta_1$-th and $\beta_2$-th black vertices lie on different 
     connected components; or, less obviously, when such vertices do lie in the same 
     connected graph, $\mathcal R$, but they are not related
     by a non-trivial element of $\Autc(\R)$. In particular,
     if the distinguished 
     connected component $\R$ has no symmetries, 
     then there are exactly half as many SDE's as vertices of $\mathcal R$, to wit $k(\R)$ equations (for that connected component).
     
     For the rest of components of $\mtc D$ we write $\mathcal Q$
     $(\mtc D=\mtc R\amalg \mathcal Q)$. We factorise $\mathcal Q$ in
     $\alpha_u$ copies of pairwise non-isomorphic connected graphs
     $\{Q_u\}_{u=1,\ldots,n}$:
   \[\mathcal Q= Q_{1}^{\alpha_{1}}\amalg\ldots \amalg Q_{n}^{\alpha_{n}}\,, \qquad\qquad (Q_i^{\alpha_i}=\amalg_{A=1}^{\alpha_i} Q_i ).\]
      We split the $d$-tuple $\mathbf X$ into momenta $X_0$ of $\mtc R$ and
      momenta $\mtb X$ of $\mtc Q$, so that  $(X_0,\mathbb X)=\Xb$, up
      to reordering.  For $\tau \in \mathrm{Br}\hp{2}(\R,\beta,c)$ we can
      therefore write \[\varsigma_c(\R;\beta,\tau)= \mtc R'\amalg \mtc R'',\]
      and accordingly split the momentum $X_0$ in momenta $X_0'$ of
      $\mtc R'$ and $X_0''$ of $\mtc R''$,
      $X_0=(X_0',X_0'')$. Furthermore, for any factorising pair of
      graphs
    \[ \mtc C, \B  \in \mathsf G_D \mbox{ with } \mtc C \amalg \B \cequal \mtc Q\,,
    \]
    we define the two functions $H\hp{c,\tau}_{\mtc C,\mtc B}$ and
    $ I\hp c_{\mathcal C,\mathcal B}$ by the following products:
    \begin{align*}
    H\hp {c,\tau}_{\,\mathcal C,\mathcal B}&= \,\,\, \big\langle\, G_{\mtc R' | \mtc C }(X_0'; \balita ) 
    \times  G_{\mtc R''| \mtc B }(X_0''; \balita ) \, \big\rangle_{\mtc Q}\,, \qquad (\tau \in \mathrm{Br}\hp{2}(\R,\beta,c))\\
     I\hp c_{\mathcal C,\mathcal B}&= \frac{1}{|\Autc (\B)|}\,
    \big\langle \,\mathfrak f\hp c_{\mathcal C,s_c} \, \times 
     \,G_{\mathcal R | \B}(X_0; \,\balita) \big\rangle_{\mtc Q}\,,
    \end{align*} 
    where the reorderings refer only to the graph components of the
    graphs in the pair $(\mtc C, \B)$.  The pivotal term, appearing 
    in each equation, is
    \[\mathfrak f\hp c_{\varnothing,s_c}=Y\hp c_{s_c}[0,0]=
    \suml_{\mathbf q_{\hat c}}
    \GDmelon(s_c ,\mathbf q_{\hat c}) \,.\] 
    \end{preparation*}
    \begin{thm} \label{thm:SDEDisconnected}
    For the $\phi^4_{D,\mathrm{m}}$-theory with kinetic term \eqref{eq:kinetic}, 
    the Schwinger-Dyson equations for the disconnected graph $\mtc D$
    read, for the particular vertex choice $\sb=\yb^\beta$, as follows:
            \allowdisplaybreaks[1]
    \begin{align}
    % &\!\!\!\bigg(1+\frac{2\lambda}{E_{\sb}} \suml_{a=1}^D \suml_{\mathbf q_{\hat c} }\GDmelon (s_c,\mathbf q_{\hat c})\bigg)
    \nonumber
     &G_{\mtc D}\hp{2k}(\Xb)  \\
    & = \nonumber
    % \frac{\delta_{1,k}}{E_\sb}+
    \frac{(-2\lambda)}{E_{\sb}} \suml_{c=1}^D \Bigg\{ 
    \suml_{\hat \sigma \in \Autc({\mtc D})} \sigma \cdot  \mathfrak{f}\hp{c}_{{\mtc D},s_c}(\Xb) 
    %   & \qquad\qquad \qquad 
    %  \,\,\,\,\,
     \\ \nonumber
     &    +
     \suml_{\rho \neq \beta }  \frac{1 }{  E(y^\rho_c,s_c)} \bigg[\dervpar{\,W[J,\bJ]}{\varsigma_c({\mtc D} ; \beta,\rho) }  (\Xb)
    - \dervpar{\,W[J,\bJ]}{\varsigma_c({\mtc D} ; \beta,\rho) } (\Xb|_{s_c \to y_c^{\rho}})\bigg] \nonumber
     \\
    &  \label{eq:SDEs} 
    - \suml_{b_c}\frac{1}{E(s_c,b_c)} \big[ G_{\mtc D}\hp{2k}(\Xb) - 
    G_{\mtc D}\hp{2k}(\Xb|_{s_c \to b_c}) \big]    
    \\ \nonumber 
    &    + \suml_{\tau\in \mathrm{Br}\hp{2}(\R;\beta,c)}\bigg[\suml_{\substack{
    (\mathcal C,\B)\in \mathsf G_D \amalg  \mathsf G_D
    \\ (\mathcal C\amalg \B) \bcequal \mathcal Q}} 
    \suml_{\Omega\in \Autc( \mathcal Q)} 
    \frac{(\Omega \cdot  H\hp{c,\tau}_{\mathcal B,\mathcal C}) (\mtb X)-(\Omega \cdot  H\hp{c,\tau}_{\mathcal B,\mathcal C}) (\mtb X|_{s_c\to y^\tau_c})}{E(y^\tau_c,s_c)}\Bigg]
    % \\
    % &   \quad \qquad+ \suml_{\tau\in \mathrm{Br}\hp{2}(\R;\beta,c)} \, \nonumber
    % \suml_{\substack{
    % (\mathcal C,\B)\in \mathsf G_D \amalg \mathsf G_D
    % \\(\mathcal C\amalg\B) \mtr{\, factorises\,  } \mathbf Q}}  \suml_{\Theta \in \Autc( \mathbf Q)} 
    % (\Theta^* A_{\mathcal M,\mathcal N}) (\mtb X)
    \\
    \nonumber 
    &      + \suml_{\substack{
    (\mathcal C,\B)\in \mathsf G_D \amalg \mathsf G_D
    \\ (\mathcal C\amalg \B) \bcequal \mathcal Q}}  \suml_{\Omega \in \Autc( \mathcal Q)} 
    (\Omega \cdot I\hp{c}_{\mathcal B,\mathcal C}) (\mathbb X)
    \Bigg\}   \,.
    \end{align}
    
    \end{thm}
      \allowdisplaybreaks[1]
    
    \begin{proof}
     By definition, % \naranja{Entscheiden, ob $\partial^d$ oder $\partial$ in Z\"ahler und Nenner.}
    \[G_{\mtc D} = \dervpar{ \,W}{\mathcal
        D} =\bigg[\dervfunc{^{ } }{\mathcal
        Q^{ }} \bigg( \dervfunc{\,W}{ \mtc R}\bigg) \bigg] \Bigg |_{J=\bar J=0} \,. \] 
        Spelled out, this means that
    \begin{equation} \label{eq:redundant}
    G_{\mtc D}(\Xb)=\bigg\{\Big[
    \prod_{i=1}^n \dervfunc{^{\alpha_i}}{\mtc Q_1^i\,\delta \mtc Q_2^i\,\cdots \,\delta\mtc Q_{\alpha_i}^i} (\mathbb X)
    \big(\prod\limits_{\alpha =1}^r \fderJ{\xb_0^\alpha}\fderbJ{\yb_0^\alpha} \big) \Big] W\jj \bigg\}\bigg|_{J=\bar J=0} \,.
    \end{equation}
    
    To compute $G_{\mtc D}$, one needs to start then 
    with a functional derivative of $W$ with respect to 
    a source, which we choose to be $\bar J_{\mathbf s}=\bar J_{\yb^\beta}$. 
    The partition function $Z\jj=\exp(W\jj)$ has been shown \cite{fullward} to
    satisfy
     \begin{align} \label{eq:WnachJ_s}
    \fder{ W  }{\bJ_\sb} =
     \frac{1}{E_{\sb}} 
     \Bigg\{ J_\sb - \ee^{-W} \bigg(\dervpar{\Sint(\phi,\bar\phi)}{\bar\phi_\sb}\bigg)\bigg|_{\substack{\phi \to \delta/\delta \bar J \\ \bar\phi \to \delta/\delta  J}} 
    \,\,\ee^{+W} \Bigg\} \,.
    \end{align}
%     The term in parenthesis we denote henceforth by $(\partial\Sint(\phi,\bar\phi)/\partial{\bar\phi_\sb})^{\mtr{sour.}}$.
    The colour-$c$-WTI leads to  
    \begin{align} \label{eq:SintPostWard}
    &\bigg(\dervpar{\Sint(\phi,\bar\phi)}{\bar\phi_\sb}\bigg)\bigg|_{\substack{\phi \to \delta/\delta \bar J \\ \bar\phi \to \delta/\delta  J}} Z[J,\bJ]
     \\ &= 2\lambda \sum_{c=1}^D(A_c(\sb) -  B_c(\sb)  + C_c(\sb) + D_c(\sb) +F_c(\sb) )\,,
     \end{align}
    where each of the summands is given by
    \begin{align*}
    % \mbox{with }\qquad 
    A_c(\sb) & = Y\hp c_{s_c} [J,\bJ] 
     \cdot
     \fder{Z[J,\bJ]}{\bJ_\sb} \,, & &
     \\ \nonumber  B_c(\sb) & = \sum_{\mathbf b}
     \frac{ J_{\mathbf b_{\hat c} s_c}}{E(b_c,s_c)}
     \fder{^2Z[J,\bJ]}{\bJ_{\mathbf s_{\hat c}b_c} \delta J_\mathbf b}\,,   
     &  (\mathbf s_{\hat c} b_c &=(s_1\ldots,s_{a-1},b_c,s_{c+1}\ldots,s_D))
    \\
    C_c(\sb) & =  \sum\limits_{b_c}
     \frac{1}{E(b_c,s_c)}   
     \fder{Z[J,\bJ]}{\bJ_{\sb}} \,, & &
     \\ \nonumber  D_c(\sb) & =  \sum\limits_{\mathbf b}
     \frac{\bJ_{\mathbf b}}{E(b_c,s_c)} 
     \frac{\delta^2Z [J,\bJ]}{\delta \bJ_{\sb_{\hat c}b_c} \delta \bJ_{\mathbf b_{\hat c}s_c}} \,,
     &(\mathbf b_{\hat c} s_c &=(b_1\ldots,b_{c-1},s_c,b_{c+1}\ldots,b_D) )
    \\  F_c(\sb) & =     
     \fder{Y\hp c_{s_c} [J,\bJ]}{\bJ_\sb} 
     \cdot
     Z[J,\bJ]\,, &&
    \end{align*}
    for any $c=1,\ldots,D$. 
    In this new derivation, it is convenient to work with 
    \begin{align}
     \ee^{-W }A_c(\sb) & = Y\hp c_{s_c} [J,\bJ] 
     \cdot
     \fder{W[J,\bJ]}{\bJ_\sb} \,,
    \\      \ee^{-W} B_c(\sb) & = \sum_{\mathbf b}
     \frac{ J_{\mathbf b_{\hat c} s_c}}{E(b_c,s_c)}
     \bigg[
     \fder{^2W[J,\bJ]}{\bJ_{\mathbf s_{\hat c}b_c} \delta J_\mathbf b}
     +
      \fder{W[J,\bJ]}{  J_\mathbf b}  \fder{W[J,\bJ]}{\bJ_{\mathbf s_{\hat c}b_c}} \bigg]
     \,,
    \\
     \ee^{-W}  C_c(\sb) & =  \sum\limits_{b_c}
     \frac{1}{E(b_c,s_c)}   
     \fder{W[J,\bJ]}{\bJ_{\sb}} \,, 
     \\
      \ee^{-W } D_c(\sb) & =  \sum\limits_{\mathbf b}
     \frac{\bJ_{\mathbf b}}{E(b_c,s_c)} \bigg[
     \frac{\delta^2W [J,\bJ]}{\delta \bJ_{\sb_{\hat c}b_c} \delta \bJ_{\mathbf b_{\hat c}s_c}} 
     + \frac{\delta W [J,\bJ]}{\delta \bJ_{\sb_{\hat c}b_c}  }
     \frac{\delta W [J,\bJ]}{  \delta \bJ_{\mathbf b_{\hat c}s_c}}
     \bigg]\,,
    \\
    \ee^{-W} F_c(\sb) & =     
     \fder{Y\hp c_{s_c} [J,\bJ]}{\bJ_\sb} \,.
    \end{align}
    Notice the presence of the product of derivatives in the $B_c$ and
    $D_c$ terms.  These are called $B_c\hp{\mtr{prod}}$ and
    $D_c\hp{\mtr{prod}}$, respectively.  The other two terms (which
    already appeared on the SDE's for connected boundaries) containing
    a double derivative are denoted by $B_c\hp{\mtr{dd}}$ and
    $D_c\hp{\mtr{dd}}$, respectively.
    \par
    
    Next, we use the freedom to numerate momenta starting with the
    component $\mathcal R$, 
    \[
    X_0=(\xb^1_0,\ldots,\xb^r_0)
    \in \mathcal{F}_{D,r},\,\,\mtc R_*(X_0)=(\yb^1_0,\ldots,\yb^r_0), \]
    and $\sb= \yb^\beta$. For each item in the list
    \[
    (\mathfrak{m},M)\in\{ (\mathfrak{a},A), (\mathfrak{b},B), (\mathfrak{c},C), (\mathfrak{d},D),(\mathfrak{f},F)\}\,, 
    \]
    we define 
    the following functions:
    \begin{equation}
     \label{eq:goticas}
    \mathfrak m_c{(\Xb;\sb;\mathcal{D})}:= 
    % \prod_{i=1}^n \dervfunc{^{\alpha_i}}{\mtc Q_1^i\,\delta \mtc Q_2^i\,\cdots \,\delta\mtc Q_{\alpha_i}^i} 
    \dervfunc{^{\alpha}}{\mathcal Q^{\alpha}(\mathbb X)}
    \prod\limits_{\substack{\alpha \neq \beta \\ 0\leq \alpha\leq r \\ \nu=1,\ldots,r} }\fder{}{\bJ_{\yb^\alpha_0}} \fder{}{J_{\xb^\nu_0}}   
    \big[ \ee^{-W\jj}M_c(\sb) \big]\bigg|_{J=\bJ=0} \,.
    \end{equation}
    The splitting of $\mathfrak b$ and $\mathfrak d$ into terms $\mathfrak{b}\hp{\mtr{dd}}, \mathfrak{b}\hp{\mtr{prod}},  
    \mathfrak{d}\hp{\mtr{dd}},\mathfrak{d}\hp{\mtr{prod}}$, respectively, still makes sense.
     We now determine all coefficients, beginning with the easiest. \\    
    The $\mathfrak c_c$ and $\mathfrak f_c$ terms 
    are readily computed:
    \begin{align*}
     \mathfrak c_c(\Xb;\sb;\mathcal{D})& = \suml_{b_c} \frac{1}{E(b_c,s_c)} G_\mathcal{D}\hp{2k}(\Xb)\,, \\
     \mathfrak f_c(\Xb;\sb;\mathcal{D})&=\suml_{\hat\pi \in \Autc(\mathcal{D})}\pi^*\mathfrak f _{\mathcal{D}} \hp{c}(\Xb)
     \,.
    \end{align*}
    The term $ \mathfrak c_c$ itself is not finite, but a term arising
    from one of the three functions will render it finite.  The three
    remaining $\mathfrak m_c$-functions involve derivatives of
    products of functionals.  We first observe that the $2r-1$
    derivatives in the sources complete the graph derivative
    $\delta W / \delta \mtc R$, so the $\mathfrak a_c$ yields
    \begin{align}
      \mathfrak a_c(\Xb;\sb;\mathcal{D}) &=  \nonumber
      \dervfunc{^{\alpha}}{\mathcal Q^{\alpha}(\mathbb X)}
    \prod\limits_{\substack{\alpha \neq \beta \\ 0\leq \alpha\leq r \\ \nu=1,\ldots,r} }\fder{}{\bJ_{\yb^\alpha_0}} \fder{}{J_{\xb^\nu_0}}   
     \bigg(Y\hp c_{s_c} [J,\bJ] 
     \cdot
     \fder{W[J,\bJ]}{\bJ_\sb} \bigg)\bigg|_{J=\bar J=0}
      \\
     & =   
      \dervfunc{^{\alpha}}{\mathcal Q^{\alpha}(\mathbb X)} 
     \bigg(Y\hp c_{s_c} [J,\bJ] 
     \cdot
     \fder{W[J,\bJ]}{\mtc R(X_0)} \bigg) \bigg|_{J=\bar J=0}
     \,.
    \end{align}
    Notice that, after evaluation in the sources at 0,
    
    \[\dervpar{}{\A} \cdots \dervpar{}{\A'} \dervpar{}{\mtc R(X_0)} W[J,\bJ]= G_{\mtc R|\mtc A|\cdots |\mtc A'}(X_0,\balita)\]
    holds
    for any boundary graphs $\A,\ldots,\A'$. Using the formula for the graph derivative of 
    products (Lem. \ref{thm:Leibniz}) and subsequently Lemma \ref{thm:Borel}
    one deduces
    \begin{align*}
       \mathfrak a_c(\Xb;\sb;\mathcal{D}) &=  
       \suml_{\substack{   \B,\mtc C \\ \B\amalg \mtc C\bcequal \mtc Q}} 
       \suml_{\Omega\in\Autc(\mathcal{Q})} 
       \suml_{\sigma \in\Autc(\mathcal{R})} \\
       & \hspace{1.7cm}
       \frac{1}{|\Autc(\mtc R \amalg \mtc B)|}\Omega \cdot \big\langle G_{\mtc R|\B} (\sigma (X_0); 
       \balita )\times \mathfrak f_{\mtc C,s_c}\hp c \,\big\rangle_{\mtc Q} \\
       & =\suml_{\substack{   \B,\mtc C \\ \B\amalg \mtc C\bcequal \mtc Q}} 
       \suml_{\Omega\in\Autc(\mathcal{Q})} 
    %    \suml_{\sigma \in\Autc(\mathcal{R})}
       \frac{1}{|\Autc(\mtc B)|}\Omega \cdot \big\langle G_{\mtc R | \B} (X_0; 
       \balita )\times \mathfrak f_{\mtc C,s_c}\hp c \,\big\rangle_{\mtc Q}  \\
       &=\suml_{\substack{   \B,\mtc C \\ \B\amalg \mtc C\bcequal \mtc Q}} 
       \suml_{\Omega\in\Autc(\mathcal{Q})} \Omega \cdot  I\hp c_{\mathcal C,\B} (\mathbb X)\,.
    \end{align*}
    
    We compute now\footnote{The part of the calculation of the
      contributions of the double derivatives
      $\delta^2 W/\delta J\delta \bar J$ or
      $\delta^2 W/\delta \bar J\delta\bar J$ to $\mtf{b}_c$ and
      $\mathfrak{d}_c$ is shortened, due to the very similar
      derivation provided already in \cite{SDE}.}  the term
    $\mathfrak{d}_c$. We can split $D_c$ into the double derivative
    contribution 
    $D_c\hp{\mtr{dd}}$ and the product of single derivatives
    $D_c\hp{\mtr{prod}}$.  If we decompose
    $\delta^\alpha / \delta \mathcal Q^\alpha$ and the rest of
    derivatives implied in $\mathcal R$ into single functional
    derivatives,
    \begin{align} \nonumber 
    &
    \mathcal{O}(\bar J)+ 
    \prod\limits_{\substack{\alpha \neq \beta \\ \nu=1,\ldots,d} } 
    \fder{}{\bJ_{\yb^\alpha}} \fder{}{J_{\xb^\nu}} \big(\ee^{-W}D_c\hp{\mtr{dd}}(\sb) \big) 
    \\
    & =
     \prod\limits_{\substack{\alpha \neq \beta; \nu}} 
     \fder{}{\bJ_{\yb^\alpha}}\fder{}{J_{\xb^\nu}}\bigg[
     \sum_{\mathbf b}
     \frac{1}{E(b_c,s_c)}
     \bJ_{\mathbf b}
      \frac{\delta^2W [J,\bJ]}{\delta \bJ_{\sb_{\hat c}b_c} \delta \bJ_{\mathbf b_{\hat c}s_c}} 
     \bigg] 
    \\
    & = 
    \suml_{\substack{ \rho=1 \\ \rho\neq \beta}}^d
    \prod\limits_{\substack{ \beta \neq \alpha\neq \rho \\ \nu=1,\ldots, d}} 
    \fder{}{\bJ_{\yb^\alpha}}\fder{}{J_{\xb^\nu}}\bigg[
     \sum_{\mathbf b}
     \frac{ \delta^{\mathbf b}_{\yb^\rho} }{E(b_c,s_c)}
      \frac{\delta^2W [J,\bJ]}{\delta \bJ_{\sb_{\hat c}b_c}
      \delta \bJ_{\mathbf b_{\hat c}s_c}}
     \bigg]\, \nonumber \\
     & = 
    \suml_{\substack{\rho=1 \\ \rho\neq \beta}}^d \,
    \prod\limits_{\substack{\alpha; ( \beta \neq \alpha\neq \rho) \\ \nu=1,\ldots, k}} \fder{}{\bJ_{\yb^\alpha}}\fder{}{J_{\xb^\nu}}\bigg[
     \frac{ 1 }{E(y^\rho_c,s_c)}
      \frac{\delta^2W [J,\bJ]}{\delta \bJ_{\sb_{\hat c}y_c^\rho} \delta \bJ_{\mathbf y_{\hat c}^\rho s_c}}
     \bigg]\, .
     \nonumber
    %  \label{eq:swap_black}
    \end{align}
    This is, after evaluation at $\bar J=J=0$, all the (coloured) graphs obtained 
     from $\mathcal{D}$ (also implying the other connected components)
     after the colour-$c$ swapping at $\sb=\yb^\beta$ and $\yb^\rho$:
    \begin{figure}
    \includegraphics[height=3.0cm]{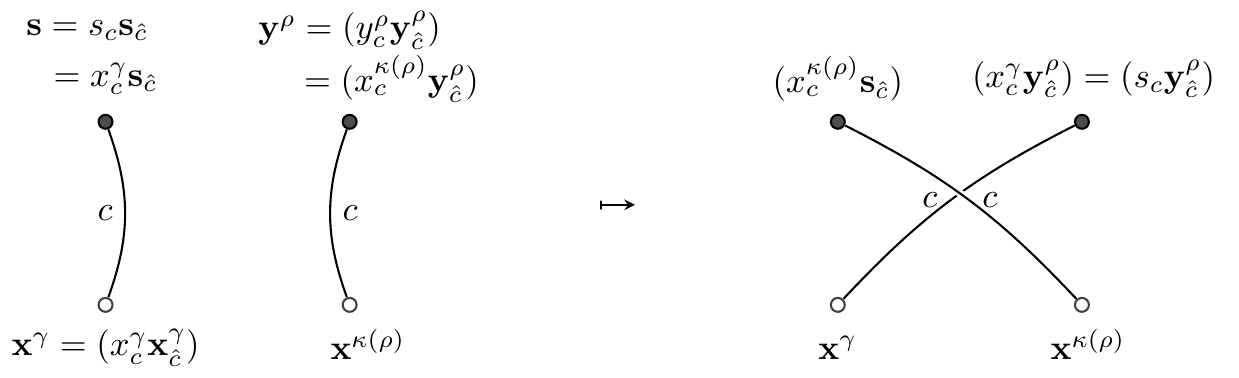}
    \caption{ \label{fig:swap}
    The swapping operation $\varsigma_c$; variables 
    are relevant in the proof}
    \end{figure}
    for $\rho\neq\beta$, i.e. running over black vertices 
    skipping $\bJ_\sb=\bJ_{\yb^\beta}$. Thus the derivatives on $D\hp{\mtr{dd}}_c$
    contribute 
    
    \[\suml_{\substack{\rho\neq \beta \\ \rho, \mtr{any\, black\, vertex}} } 
\frac{1}{  E(y^\rho_c,s_c)}  \dervpar{W[J,J]}{\varsigma_c(\mathcal{D} ;\beta,\rho)(\Xb) } 
    \mbox{ for }\rho\neq \beta.\] 
    \allowdisplaybreaks[1]
  To fully compute $D_c$, we add now the product of derivatives,
  $D_c\hp{\mtr{prod}}$.  Unlike $D_c\hp{\mtr{dd}}$, with
  $D_c\hp{\mtr{prod}}$ one can first see the effect of applying the
  rest of the derivatives of the graph $\mtc R$, that is those with
  momenta $(\xb_0^\alpha)_{\alpha}$ and
  $(\yb^\alpha_0)_{\alpha\neq \beta}$. This is due to the product of
  derivatives of $W$, which forces both factors to be derived with
  respect to momenta in the same graph in order not to vanish. Thus,
     \allowdisplaybreaks[1]
    \begin{align} \nonumber 
     &
    \prod\limits_{\substack{\alpha \neq \beta \\ \nu=1,\ldots,r} } 
    \fder{}{\bJ_{\yb^\alpha_0}} \fder{}{J_{\xb^\nu_0}} \big(\ee^{-W}D_c\hp{\mtr{prod}}(\sb) \big) 
    \,+\mathcal{O}(\bar J)
    \\
    &=
     \prod\limits_{\substack{\alpha \neq \beta ; \nu}} 
     \fder{}{\bJ_{\yb^\alpha_0}}\fder{}{J_{\xb^\nu_0}}\bigg[
     \sum_{\mathbf b}
     \frac{ \bJ_{\mathbf b}}{E(b_c,s_c)}
      \frac{\delta W }{\delta \bJ_{\sb_{\hat c}b_c}}
      \frac{\delta W  }{ \delta \bJ_{\mathbf b_{\hat c}s_c}} 
     \bigg] 
    \\
    & = 
    \suml_{\rho\neq \beta }^r
    \prod\limits_{\substack{\alpha\neq \beta (\alpha\neq \rho)\\ \nu=1,\ldots, r}} \fder{}{\bJ_{\yb_0^\alpha}}\fder{}{J_{\xb^\nu_0}}\bigg[
     \sum_{\mathbf b}
     \frac{ \delta^{\mathbf b}_{\yb^\rho} }{E(b_c,s_c)}
      \frac{\delta W  }{\delta \bJ_{\sb_{\hat c}b_c}}\frac{\delta W}{ \delta \bJ_{\mathbf b_{\hat c}s_c}}
     \bigg]\, \nonumber
     \\
     & = 
    \suml_{\rho\neq \beta }^r
     \frac{1 }{E(y^\rho_c,s_c)}
    \prod\limits_{\substack{\alpha\neq \beta (\alpha\neq \rho)\\ \nu=1,\ldots, r}} \fder{}{\bJ_{\yb_0^\alpha}}\fder{}{J_{\xb^\nu_0}}\bigg[ 
      \frac{\delta W  }{\delta \bJ_{\mathbf y^\beta_{\hat c}y^\rho_c}}
      \frac{\delta W}{ \delta \bJ_{\mathbf y^\rho_{\hat c} y^\beta_c}}
     \bigg]\,, \nonumber  
     \end{align}
        \allowdisplaybreaks[1]
     where in the last step we only used that
     $\sb=\yb^\beta=\yb^\beta_0$.  It is evident that these two
     derivatives in the square bracket form of a colour-$c$ edge swap
     at $\yb_0^\beta$ and $\yb_0^\rho$, but in order not to lead to a
     vanishing term, they also have to lie on a different component of
     a graph (as otherwise a graph derivative would be incomplete).
     It is therefore additionally required that the graph  is disconnected after the
     swapping at $\yb_0^\beta$ and $\yb_0^\rho$, that
     is that there are connected graphs $\mathcal R_1$ and $\mathcal R_2$,
     such that
     \[
     \varsigma_c(\R;\beta,\rho)=\mathcal R_1 \amalg \mathcal R_2
     \qquad\mathrm{if\,and\,only\,if} 
      \qquad \rho \in \mtr{Br}(\R,\beta,c)\,.
     \]
     Therefore
     \begin{align} \nonumber 
     &\dervfunc{^{\alpha}}{\mathcal Q^{\alpha}(\mathbb X)}
    \prod\limits_{\substack{\alpha \neq \beta \\ \nu=1,\ldots,r} } 
    \fder{}{\bJ_{\yb^\alpha_0}} \fder{}{J_{\xb^\nu_0}} \big(\ee^{-W}D_c\hp{\mtr{prod}}(\sb) \big) 
    \,+\mathcal{O}(\bar J)
    \\
    &
    =
    \suml_{\tau \in \mtr{Br}(\R,\beta,c) }
     \frac{1 }{E(y^\tau_c,s_c)}\dervfunc{^{\alpha}}{\mathcal Q^{\alpha}(\mathbb X)}
    \prod\limits_{\substack{\alpha\neq \beta (\alpha\neq \rho)\\ \nu=1,\ldots, r}} \fder{}{\bJ_{\yb_0^\alpha}}\fder{}{J_{\xb^\nu_0}}\bigg[ 
      \frac{\delta W  }{\delta \bJ_{y^\beta_{\hat c}y^\tau_c}}
      \frac{\delta W}{ \delta \bJ_{\mathbf y^\tau_{\hat c}y^\beta_c}}
     \bigg]\,. \nonumber
     \\
     &
     =
    \suml_{\tau \in \mtr{Br}(\R,\beta,c) }
     \frac{1 }{E(y^\tau_c,s_c)}\dervfunc{^{\alpha}}{\mathcal Q^{\alpha}(\mathbb X)}
     \bigg(
     \dervfunc{W}{\mathcal R_1(X_0')} \dervfunc{W}{\mathcal R_2(X_0'')} \bigg)
     \,. \nonumber  
     \end{align}
     At this place we use the multivariable graph calculus Leibniz rule (Lemma \ref{thm:Leibniz}), namely 
      \begin{align} \nonumber 
    & \dervfunc{^{\alpha}}{\mathcal Q^{\alpha}(\mathbb X)}
    \prod\limits_{\substack{\alpha \neq \beta \\ \nu=1,\ldots,r} } 
    \fder{}{\bJ_{\yb^\alpha_0}} \fder{}{J_{\xb^\nu_0}} \big(\ee^{-W}D_c\hp{\mtr{prod}}(\sb) \big) \bigg|_{J=\bar J=0}
    \\
    &=\suml_{\tau \in \mtr{Br}(\R,\beta,c) }
     \frac{1 }{E(y^\tau_c,s_c)}\suml_{\substack{\B,\mathcal C  \\ \B\amalg \mtc C \bcequal \mathcal Q^\alpha }}
    \suml_{ \Omega\in \Autc(\mathcal Q^\alpha)} \nonumber 
    \\ & \hspace{4cm} \nonumber 
    \Omega \cdot  \big[ 
    \langle\, G_{\mtc R_1|\B}(X_0',) \times G_{\mtc R_1|\mathcal C} (X_0'', )\,\rangle_{\mathcal Q^\alpha} \big](\mathbb X)
    \\
    &=
    \suml_{\tau \in \mtr{Br}(\R,\beta,c) }
     \frac{1 }{E(y^\tau_c,s_c)}\suml_{\substack{\B,\mathcal C  \\ \B\amalg \mtc C \bcequal \mathcal Q^\alpha }}
    \suml_{ \Omega\in \Autc(\mathcal Q^\alpha)} 
    (\Omega \cdot H\hp{c,\tau}_{\B,\mtc C})(\mathbb X)\,.\nonumber 
    \end{align}
    In summary, the $\mtf{d}_c$-term is given by
    \begin{align} \label{eq:d_term}
    \mtf{d}_c(\Xb;\sb;\mathcal D)
    &=\suml_{\substack{\rho\neq \beta \\ \rho, \mtr{any\, black\, vertex}} }  \frac{1}{  E(y^\rho_c,s_c)}  \dervpar{W[J,J]}{\varsigma_c(\mathcal{D} ;1,\rho) } (\Xb)
    \\
    &\qquad+
    \suml_{\tau \in \mtr{Br}(\R,\beta,c) }
     \frac{1 }{E(y^\tau_c,s_c)}\suml_{\substack{\B,\mathcal C  \\ \B\amalg \mtc C \bcequal \mathcal Q^\alpha }}
    \suml_{ \Omega\in \Autc(\mathcal Q^\alpha)} 
    (\Omega \cdot H_{\B,\mtc C})(\mathbb X)\,. \nonumber
    \end{align}
    
    As for the derivatives on $B_c(\sb)$, we 
    divide the derivation in two parts. One concerns the 
    double derivative, $ B_c\hp{ \mtr{dd} }$: 
    \begin{align}\label{eq:runningba}
    & \mathcal{O}(J)+\prod\limits_{\substack{\alpha \neq \beta \\ \nu=1,\ldots,d} }
    \fder{}{\bJ_{\yb^\alpha}} 
    \fder{}{J_{\xb^\nu}} B_c\hp{ \mtr{dd} }(\sb) \\ \nonumber 
    & =
     \prod\limits_{\substack{\alpha \neq \beta; \nu}} \fder{}{\bJ_{\yb^\alpha}}\fder{}{J_{\xb^\nu}}\bigg[
     \sum_{\mathbf b}
     \frac{1}{E(b_c,s_c)}
      J_{\mathbf b_{\hat c} s_c}
     \fder{^2  W[J,\bJ]}{\bJ_{\mathbf s_{\hat c}b_c} \delta J_\mathbf b} 
    \bigg]
    \\
    % & = 
    % \suml_{\beta=1}^d
    % \prod\limits_{\substack{\alpha \neq \beta; \nu\neq \beta}} \fder{}{\bJ_{\yb^\alpha}}\fder{}{J_{\xb^\nu}}\bigg[
    %  \sum_{\mathbf b}
    %  \frac{1}{E(b_c,s_c)}\delta^{s_c }_{x_c^\beta}
    %  \delta^{ {\mathbf b}_{\hat c}}_{ \xb^\beta_{\hat c}}
    %  \fder{^2W [J,\bJ]}{\bJ_{\mathbf s_{\hat c}b_c} \delta J_\mathbf b} 
    %   \bigg]\, \nonumber
    %     \\ 
    &=
    \nonumber
    \suml_{\theta=1}^d
    \prod\limits_{\substack{\alpha \neq \beta; \nu\neq \theta}} \fder{}{\bJ_{\yb^\alpha}}\fder{}{J_{\xb^\nu}}\bigg[ 
     \suml_{b_c}
     \frac{1}{E(b_c,s_c)}
     \delta^{s_c }_{x_c^\theta}
     \fder{^2 W[J,\bJ]}{\bJ_{\mathbf s_{\hat c}b_c} \delta J_{\xb^\beta_{\hat c}b_c }} 
      \bigg]
      \\ &=
    \prod\limits_{\substack{\alpha \neq \beta; \nu\neq \gamma}} \fder{}{\bJ_{\yb^\alpha}}\fder{}{J_{\xb^\nu}}\bigg[ 
     \suml_{b_c}
     \frac{1}{E(b_c,x_c^\gamma)}
     \fder{^2\,W[J,\bJ]}{\bJ_{\mathbf s_{\hat c}b_c} \delta J_{\xb^\gamma_{\hat c}b_c }} 
      \bigg] \,.  \nonumber 
    \end{align}
    since there is a single vertex
    $\xb^{\gamma}$, $\gamma=\gamma(c)$, with 
    $x_c^\gamma=s_c$, so 
    $\delta_{x_c^\theta}^{s_c}=\delta_{x_c^\theta}^{s_c}\delta_{\theta}^\gamma$.
    The term  
     $\delta^2 W[J,\bJ]/\delta \bJ_{\mathbf s_{\hat c}b_c} \delta J_{\xb^\gamma_{\hat c}b_c } $, 
     is selected by $\delta_{x_c^\beta}^{s_c}$ leads to 
     $\partial Z/\partial \mathcal{D}(\Xb)|_{x^\gamma_c\to b_c}$, after taking all the rest of derivatives, with 
     the single coordinate $x^\gamma_c$ being substituted by (the running) $b_c$.
     But when 
     \[b_c \in \{y_c^1,y_c^2,\ldots,\widehat{y_c^\beta},\dots,y^d_c\}
     \,,\]
     one does not have exactly a `graph derivative', since we are evaluating 
     it not in $\mathcal F_{D,d}$, but in one of its diagonals of colour $c$.
     A direct computation yields then a second contribution to $ \mathfrak  b\hp{\mtr{dd}}_c$ 
     (the third and fourth lines below):
%      \allowdisplaybreaks[2]
     \begin{align}  
     & \mathfrak b\hp{\mtr{dd}}_c(\Xb;\sb;\mathcal{D}) \label{eq:new_bterm} \\
     & =\suml_{b_c}\frac{1}{E(b_c,x^\gamma_c)}  \nonumber \\
     & \hspace{.8cm} \times
     G_\mathcal{D}\hp{2k}(\xb^1,\ldots,\xb^{\gamma-1}, 
     (x_1^\gamma,\ldots,x_{a-1}^\gamma,b_c,x_{a+1}^\gamma,\ldots x^\gamma_D), 
     \xb^{\gamma+1}, \nonumber
     \ldots,\xb^d)  \\
     & + \suml_{\rho>1} \frac{1}{E(x_c^{\kappa(\rho)},x^\gamma_c)} \nonumber
     \\ & \hspace{.8cm}
     \times \frac{\partial W[J,\bar J]}{\partial \,\varsigma_c(\mathcal{D};1,\rho)}  
     (\xb^1,\ldots, (x_1^\gamma,\ldots,x_{a-1}^\gamma,x_c^{\kappa(\rho)},x_{a+1}^\gamma,\ldots x^\gamma_D), 
     \ldots,\xb^d)  \,. \nonumber
     \end{align} 
        \allowdisplaybreaks[1]
     where $\kappa(\rho)$ is defined in Figure \ref{fig:swap}  (i.e. $x_c^{\kappa(\rho)}=y_c^{\rho}$).
     \par
     The last computation is $\mathfrak b\hp{\mtr{prod}}_c$, 
     \begin{align*}
     & \mathfrak b\hp{\mtr{prod}}_c (\Xb; \sb; \mathcal D)
     \\ 
     & =   
    \dervfunc{^{\alpha}}{\mathcal Q^{\alpha}(\mathbb X)}
    \prod\limits_{\substack{\alpha \neq \beta \\ 0\leq \alpha\leq r \\ \nu=1,\ldots,r} }\fder{}{\bJ_{\yb^\alpha_0}} \fder{}{J_{\xb^\nu_0}}   
     \bigg[\sum_{\mathbf b}
     \frac{ J_{\mathbf b_{\hat c} s_c}}{E(b_c,s_c)}
    \fder{W[J,\bJ]}{  J_\mathbf b}  \fder{W[J,\bJ]}{\bJ_{\mathbf s_{\hat c}b_c}} \bigg]\bigg|_{J=\bJ=0} \,.
     \end{align*}
     This computation is quite similar to the one for the
     $\mathfrak d\hp{\mtr{prod}}_c$-term, presented above, with the
     only difference that the evaluation is not at $\mathbb X$, but at
     $\mathbb X|_{s_c\to y^\tau_c}$.  The less trivial part in that
     derivation is to figure out, which non-zero contributions come
     from
     \begin{equation} 
     \fder{W[J,\bJ]}{  J_{b_c \xb^\gamma_c}}  \fder{W[J,\bJ]}{\bJ_{\mathbf s_{\hat c}b_c}}\,. 
     \end{equation}
    
     Because of the repetition of $b_c$ in both factors, it seems that
     the term vanishes after deriving it. However, $b_c$ runs, and it
     does so also through the particular $c$-coloured entries of the
     black vertices of $\mathcal R$,
     \[
     b_c \in
     \{y_{0,c}^1,y_{0,c}^2,\ldots,\widehat{y_{0,c}^\beta},\dots,y^d_{0,c}\}\,.\]
     Only if we also require that
     $b_c \in \{ y^\tau_c | \tau \in \mtr{Br}(\mathcal R, \beta , c
     )\}$, we guarantee that each one of those factors forms a graph
     derivative.  However, notice that momentum in the white vertex
     $\xb^\gamma=(s_c\xb^\gamma_{\hat c})$ has changed to
     $(y^{\tau}_c\xb^\gamma_{\hat c})$. Thus, one changes $X_0$ into
     $X_0\big|_{s_c\to y^\tau_c}$.  Therefore
     \begin{align*}
     \mathfrak b\hp{\mtr{prod}}_c (\Xb; \sb; \mathcal D)
     & =   \sum\limits_{\tau\in \mtr{Br}(\mathcal R, \beta,c)}
     \frac1{E(y^\tau_{c},s_c)}
    \dervfunc{^{\alpha}}{\mathcal Q^{\alpha}(\mathbb X)}  
     \dervfunc{W}{\mathcal R_1(X'_0)} \dervfunc{W}{\mathcal R_2(X''_0)}
     \bigg|_{\substack{ J=\bJ=0 \\  {s_c\to y^\tau_c}}} \,.
     \end{align*}
     We apply again Lemma \ref{thm:Leibniz} and find that 
      \begin{align*}
     \mathfrak b\hp{\mtr{prod}}_c (\Xb; \sb; \mathcal D)
     =
     \suml_{\tau\in \mathrm{Br}\hp{2}(\R;\beta,c)} \,
    \suml_{\substack{
    (\mathcal C,\B)\in \mathsf G_D \amalg \mathsf G_D
    \\ (\mathcal C\amalg \B) \sim \mathcal Q}}  \suml_{\Omega\in \Autc( \mathcal Q)} 
    \frac{ (\Omega \cdot  H_{\mathcal B,\mathcal C}) (\mtb X)|_{s_c\to y^\tau_c}}{E(y^\tau_c,s_c)} \,.
     \end{align*} 
    Due to eqs. \eqref{eq:SintPostWard}, one has
    \begin{align*}
    \dervpar{W\jj}{\mtc D(\Xb)} & = \prod_i \dervfunc{^{\alpha_i}}{\mtc Q_i(X_{1}^i)\cdots \mtc Q_i(X_{\alpha_i}^i)}  
    \bigg\{\prod\limits_{\substack{\alpha \neq \beta  \\ 
    \nu=1,\ldots,k} } \fder{}{\bJ_{\yb^\alpha}} \fder{}{J_{\xb^\nu}} \Big(
     (-2\lambda  E_\sb\inv)
\\& \qquad\!\!\times \suml_{c=1}^D
     \ee^{-W}\big[A_c(\sb) +  C_c(\sb)  
    + D_c(\sb) + F_c(\sb) -  B_c(\sb) \big]\Big)\bigg\}\!\bigg|_{\substack{\bJ=0 \\ J=0}} \\
    &= \frac{(-2\lambda)}{E_{\sb}}
     \sum_c \Big(\mathfrak{a}_c(\Xb;\sb;\mathcal{D}) +  \mathfrak{c}_c(\Xb;\sb;\mathcal{D})   +
     \mathfrak{d}\hp{\mtr{dd}}_c(\Xb;\sb;\mathcal{D})
      \\
     & \qquad\qquad\qquad\,\,  +\mathfrak{d}\hp{\mtr{prod}}_c(\Xb;\sb;\mathcal{D})   + \mathfrak{f}_c(\Xb;\sb;\mathcal{D})\\
     & \qquad \qquad\qquad \,\,-  \mathfrak{b}_c\hp{\mtr{dd}}(\Xb;\sb;\mathcal{D})
      -\mathfrak{b}\hp{\mtr{prod}}_c(\Xb;\sb;\mathcal{D})\Big) \,.
    \end{align*}
    This is precisely $G_{\mathcal D}(\Xb)$ and the result follows. \qedhere
    \end{proof}
    
    % $\Autc(\mathcal Q)=
    % \Autc(Q_1^{  \alpha_1}\amalg Q_2^{  \alpha_2}\amalg \ldots\amalg Q_n^{  \alpha_n})= \prod_{u=1}^n \Autc (Q_u) \wr \Sym(\alpha_u)$, which has order 
    % \[|\Autc(\mathcal Q)|= \alpha_1!\cdots \alpha_n! \times |\Autc (Q_1)|^{\alpha_1}
    % \cdots |\Autc (Q_n)|^{\alpha_n}\,.\]
    
    \section{Four and six point SDE with disconnected boundary}\label{sec:foursixSDE}
    
    Concrete SDE's for the rank-3 $\phi^4_3$-theory are presented next.
    Recall, the interaction in this case is
    $\lambda(\vuno + \vdos + \vtres)$.  We display some of the equations
    in traditional notation with explicit graphs, which allows to see immediately the
  graph operations. In other equations we use the simplification 
  clarified in Table \ref{table:notation}.

    \begin{table}[h]
    \centering
    \includegraphics[width=8cm]{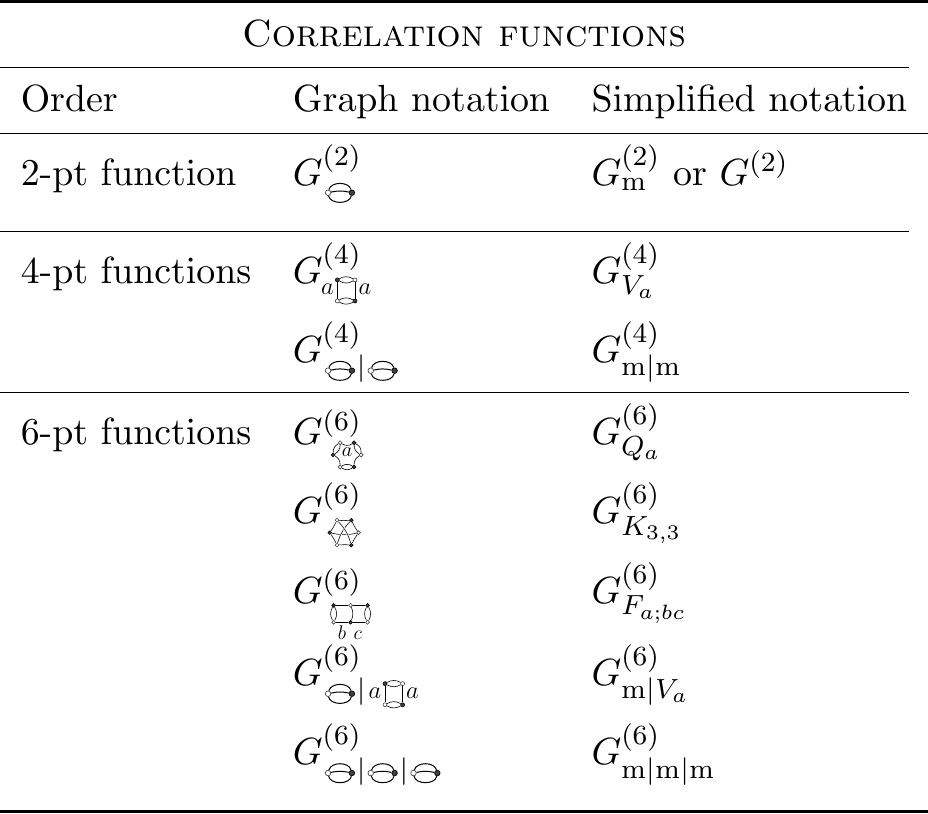}
    \vspace{.3cm}
    \caption{ Two notations for the correlation functions. Here
      $a,b,c$ are colours bound to satisfy
      $\{a,b,c\}=\{1,2,3\}$ \label{table:notation}. The subindex $\mathrm{m}$ 
      originates from `melon'.}
\end{table}
      \subsection{Schwinger-Dyson equations for \protect $G_{\protect \meloncik|\protect \vi }\hp6$ }
    
We single out the terms in the derivation of the SDE
for this case, which is the most complicated presented here.
The rest of the results are obtained in a similar and more
direct way.

  \thispagestyle{plain}
  
   There are two equations, depending on whether one chooses 
    $\sb$ (cf. Theorem \ref{thm:SDEDisconnected} above) as a component of the outgoing momentum in $\meloncik$ or in $\vi$.
      We choose this last vertex to be $V_1=\vuno$, for sake of clarity (since the model 
      is colour-invariant, SDE for the other colours are readily obtained from it).
      
      \begin{itemize}
        \itemB \textit{If $\sb=\xb$ is outgoing momentum of the graph $\meloncito$}. 
In the notation of the theorem, here $\mtc D=\mm |V_1$ being $\mtc R=\mm$,
since $\yb^1=\sb=X_0$ is the momentum of the black vertex of $\mm$.
Therefore $\mathcal Q=V_1$. The remaining momenta $\mathbb X$ equal $(\yb,\zb)$.
 The $I\hp c_{\mathcal{C},\B}$-terms are then computed as follows, from
 any of the factorisations $\mtc C,\B = (\varnothing,V_1)$ or
 $(V_1,\varnothing)$ and read
 \begin{align}
  I\hp c_{\mathcal C,\mathcal B}(\mathbb X)&= \frac{1}{|\Autc (\B)|}\,
  \langle \mathfrak f\hp c_{\mathcal C; s_c} \, \times
  \,G_{\mm| \B}(X_0; \,\balita) \rangle_{V_1} (\mathbb X)\, \\
  & =
  \begin{cases}
  \mathfrak f\hp c _{V_1;s_c}(\yb,\zb) G_{\mm}\hp 2(\xb) & (\mtc C,\B) = (V_1,\varnothing) \,,\\
  \frac12 \mathfrak f\hp c _{\varnothing;s_c} G_{\mm|V_1}\hp 6(\xb,\yb,\zb) & (\mtc C,\B) = (\varnothing, V_1)\,.
  \end{cases}
 \end{align}
On these functions $\hat\sigma\in \Autc(V_1)$ acts exchanging
$\yb $ with $\zb$; just as on the terms
coming from the derivative of the $Y$-term with respect to $(\mm | V_1)$:
 \begin{align}
  \sum_{\hat\sigma\in \Z_2} \sigma \cdot \mathfrak f\hp c _{\mm |V_1 ; x_c}(\xb,\yb,\zb)
  = \mathfrak f_{\mm |V_1 ; x_c}\hp c(\xb,\yb,\zb) +
  \mathfrak f_{\mm |V_1 ; x_c}\hp c(\yb,\xb,\zb)
  \,.
 \end{align}
Next, we obtain the second line in eq. \eqref{eq:SDEs} (the `swap-term').
We have chosen $\yb^1=\sb$ ($\beta=1$), so
\begin{subequations}
 \begin{align}
 \dervpar {W[J,\bJ]}{\varsigma_{c=1}(\mm | V_1; \beta=1,\rho=2,3)}(\xb,\yb,\zb)
 &= G\hp 6_{Q_1}(\xb,\yb,\zb)
 \\
  \dervpar {W[J,\bJ]}{\varsigma_{c=2}(\mm | V_1; \beta=1,\rho=2)}(\xb,\yb,\zb)&=
 G_{F_{3;21}}\hp 6(\xb,\zb,\yb)
  \\
  \dervpar {W[J,\bJ]}{\varsigma_{c=2}(\mm | V_1; \beta=1,\rho=3)}(\xb,\yb,\zb)&=
 G_{F_{3;21}}\hp 6(\xb,\yb,\zb)
  \\
  \dervpar {W[J,\bJ]}{\varsigma_{c=3}(\mm | V_1; \beta=1,\rho=2)}(\xb,\yb,\zb)&=
 G_{F_{2;31}}\hp 6(\xb,\zb,\yb) \\
  \dervpar {W[J,\bJ]}{\varsigma_{c=2}(\mm | V_1; \beta=1,\rho=3)}(\xb,\yb,\zb)&=
 G_{F_{2;31}}\hp 6(\xb,\yb,\zb)
 \end{align}
 \end{subequations}
 In this case the set $\mathrm {Br}(\meloncito,\rho,c)$ is empty, for any values of $\rho$ and $c$. Therefore,
 the sum over the $H$-terms vanishes. For each  $ (\xb,\yb,\zb)\in \mathcal F_{3,3}$, 
 \allowdisplaybreaks[1]
  \begin{align} \label{eq:SDEdifficult}
 & G\hp 6_{\mm| V_1}(\xb,\yb,\zb) \\
 & = \nonumber
 \bigg(\frac{-2\lambda}{E_\xb}\bigg) \times \Bigg\{\suml_{c=1}^3 \mathfrak f_{\mm |V_1 ; x_c}\hp c(\xb,\yb,\zb) +
  \mathfrak f_{\mm |V_1 ; x_c}\hp c(\yb,\xb,\zb) \\
 & + \frac{1}{E(y_1,x_1)} \big[ G\hp 6_{Q_1}(\xb,\yb,\zb) - G\hp 6_{Q_1}(y_1,x_2,x_3,\yb,\zb)\big] \nonumber\\
 & + \frac{1}{E(z_1,x_1)} \big[ G\hp 6_{Q_1}(\xb,\yb,\zb) - G\hp 6_{Q_1}(z_1,x_2,x_3,\yb,\zb)\big] \nonumber\\
  & +\frac1{E(z_2,x_2)}\big[G_{F_{3;21}}(\xb,\zb,\yb) -G_{F_{3;21}}(x_1,z_2,x_3,\zb,\yb) \big] \nonumber\\
  & +\frac1{E(y_2,x_2)}\big[G_{F_{3;21}}(\xb,\yb,\zb) -G_{F_{3;21}}(x_1,y_2,x_3 \yb,\zb) \big] \nonumber\\
  & +\frac1{E(z_3,x_3)}\big[G_{F_{2;31}}(\xb,\zb,\yb) -G_{F_{2;31}}(x_1,x_2,z_3,\zb,\yb) \big]\nonumber \\
  & +\frac1{E(y_3,x_3)}\big[G_{F_{2;31}}(\xb,\yb,\zb) -G_{F_{2;31}}(x_1,x_2,y_3 \yb,\zb) \big] \nonumber\\
  & -\suml_{c=1}^3 \bigg[
  \suml_{b_c} \frac{1}{E(x_c,b_c)}\big[ G\hp 6_{\mm|V_1} (\xb,\yb,\zb) - G\hp 6_{\mm | V_1} (\mathbf x_{\hat c}b_c,\yb,\zb) \big]
  \nonumber\\
  & + \big[( \mathfrak f\hp c _{V_1;s_c}(\zb,\yb)+ \mathfrak f\hp c _{V_1;s_c}(\yb,\zb)\big]\cdot G \hp 2(\xb)
  + \mathfrak f\hp c _{\varnothing;s_c} G_{\mm|V_1}\hp 6(\xb,\yb,\zb) \bigg] \nonumber \Bigg\}  \,.   
 \end{align}

        \itemB \textit{If $\sb=\xb$ is outgoing momentum of the boundary graph $\vuno$}. 
 For $\sb=(x_1,y_2,y_3)$ an outgoing momentum of $V_1$ we derive now the SDE for
 $G\hp 6 _{\mm | V_1}$. Then $\mathcal R=V_1$, $\mtc Q=\mm$,
 by definition. The $I\hp c_{\mtc C,\B}$-coefficients are given by
 \begin{align*}
 I\hp c_{\mm, \varnothing}(\xb,\yb,\zb) & = \mtf f_{\mm; s_c}\hp c (\zb) \cdot G_{V_1}(\xb,\yb)\, , \\
 I\hp c_{\varnothing, \mm}(\xb,\yb,\zb)& = \mtf f_{\varnothing; s_c}\hp c \cdot G_{V_1|\mm}(\xb,\yb,\zb)
 \,.
 \end{align*}
The $H\hp {c,\tau}_{\mtc C,\B}$-terms are computed from
the set
\[
\mtr{Br}(\vunoup,\beta=1,c) =\begin{cases}
 \{2\} & c=1, \qquad \\
\,\,\emptyset & \mbox{otherwise} \,,
\end{cases}
\]
since only the colour-$1$ swap at the vertex $\sb$ with the vertex
with outgoing momentum $\yb^2$ (also in $V_1$) separates $\vuno$.  The
only contributions are therefore
\begin{align*}
H_{\varnothing,\mm}\hp {c=1,\tau=2}(\Xb)&= G\hp 2 (\xb)\cdot G\hp 4 _{\mm|\mm}(\yb,\zb)\,, 
\\
H_{\mm,\varnothing}\hp {c=1,\tau=2}(\Xb)&= G\hp 4_{\mm|\mm} (\xb,\zb) \cdot G\hp 2 (\yb)\,.
\end{align*}
Thus $G\hp 6_{ V_1 |\mm}$ satisfies, for all $(\xb,\yb,\zb)\in \mathcal F_{3,3}$, 
 \begin{align} \label{eq:SDEdifficultdos}
 &\!\!\! G\hp 6_{ V_1 |\mm}(\xb,\yb,\zb)  \\
 &=\bigg(\frac{-2\lambda}{E_\sb}\bigg) \times \Bigg\{\suml_{c=1}^3 
 \mathfrak f_{\mm |V_1 ; s_c}\hp c(\zb,\xb,\yb) + \nonumber
  \mathfrak f_{\mm |V_1 ; s_c}\hp c(\zb,\yb,\xb) \\
 & + \frac{1}{E(y_1,x_1)} \big[ G\hp 6_{\mm |\mm|\mm }(\xb,\yb,\zb) -
 G\hp 6_{\mm |\mm|\mm }(y_1,x_2,x_3,\yb,\zb)\big] \nonumber\\
 & + \frac{1}{E(z_1,x_1)} \big[ G\hp 6_{Q_1}(\xb,\yb,\zb) - G\hp 6_{Q_1}(z_1,x_2,x_3;\yb,\zb)\big] \nonumber\\
& +\frac1{E(x_2,y_2)}\big[G\hp 6_{V_3|\mm}(\xb,\yb,\zb)- G\hp 6_{V_3|\mm}(x_1,y_2,x_3; \yb,\zb) \big] \nonumber\\
  & +\frac1{E(z_2,y_2)}\big[G\hp 6_{F_{3;12}}(\xb,\yb,\zb)- G\hp 6_{F_{3;12}}(\xb;y_1,x_2,y_3;\zb) \big] \nonumber \\
  & +\frac1{E(x_3,y_3)}\big[G\hp 6_{V_2|\mm}(\xb,\yb,\zb)- G\hp 6_{V_{2} | \mm}(\xb; y_1,y_2,x_3 ;\zb) \big]\nonumber \\
  &   +\frac1{E(z_3,y_3)}\big[G\hp 6_{F_{2;13}}(\xb,\yb,\zb)-
  G\hp 6_{F_{2;13}}(\xb;y_1,y_2,z_3;\zb) \big]   \nonumber\\
  & -\suml_{c=1}^3 \bigg[
  \suml_{b_c} \frac{1}{E(s_c,b_c)}\big[ G\hp 6_{V_1|\mm} (\xb,\yb,\zb) 
  - G\hp 6_{  V_1|\mm} ( [\xb,\yb,\zb]|_{s_c\to b_c}) \big] 
  \nonumber\\
  & + \big[( \mathfrak f\hp c _{V_1;s_c}(\zb,\yb)+ \mathfrak f\hp c _{V_1;s_c}(\yb,\zb)\big]\cdot G \hp 2(\xb)
  + \mathfrak f\hp c _{\varnothing;s_c} G_{\mm|V_1}\hp 6(\xb,\yb,\zb) \bigg] \nonumber \Bigg\}
  \,,
 \end{align}
      \end{itemize}
    
  \subsection{Schwinger-Dyson equation for \protect $G_{\protect \meloncik|\protect \meloncik }\hp4$ }
There is only one SDE for the `disconnected-$\partial$' 4-point function.
For every  $(\xb,\yb,\zb)\in \mathcal F_{3,3}$,
    \begin{align}  
    & \nonumber
    \Gcmm (\xb,\yb) \\& = \bigg(\frac{-2\lambda}{E_{\xb}}\bigg)\times \suml_{c=1}^3 
    \bigg\{ 
       \suml_{\mathbf q_{\hat c}} \Gmelon (x_c,\mathbf q_{\hat c}) \cdot \Gcmm(\xb,\yb)
       +
       \Gmelon(\xb) \mathfrak{f}_{\meloncik, x_c}\hp c (\yb)
       \nonumber \\
    & \quad + \suml_{b_c}\frac{1}{E(b_c,x_c)} \big[ 
    \Gcmm(\xb,\yb)
    -
    \Gcmm(b_c\xb_{\hat c},\yb)
    \big] \label{eq:SDEmm}
    \\ 
    & \quad +
    \frac{1}{E(y_c,x_c)} \big[ 
    \Gcc(\xb,\yb)
    -
    \Gcc(b_c\xb_{\hat c},\yb)
    \big]  \nonumber
    \\ & \quad +
    \mathfrak{f}\hp c_{x_c,\meloncik|\meloncik} (\xb,\yb)
    +
    \mathfrak{f}\hp c_{x_c,\meloncik|\meloncik} (\yb,\xb)
     \nonumber
    \bigg\}\,.  
    \end{align}
   Only this equation is not new, but was already (directly) derived 
   in \cite{cinco}, in notation of Table \ref{table:notation}.

      \subsection{Schwinger-Dyson equation for \protect 
      $G_{\protect \meloncik|\protect \meloncik |\protect\meloncik}\hp6$ }
      Similarly, since one can permute the arguments of $G\hp 6_{\mm|\mm|\mm}$, 
      it satisfies only one SDE:
      \begin{align}  
       \bigg(1+\frac{2\lambda}{E_{\xb}}\suml_{c=1}^3 \suml_{\mathbf q_{\hat c}}& \Gmelon (x_c,\mathbf q_{\hat c})\bigg)
      \times \Gsmmm (\xb,\yb,\zb)    \nonumber \\
    & = \bigg(\frac{-2\lambda}{E_\xb}\bigg)
    \suml_{c=1}^3\bigg\{ 
    \mathfrak{f}\hp c_{\meloncik; x_c}(\yb)\Gcmm(\xb,\zb)
    +
    \mathfrak{f}\hp c_{\meloncik; x_c}(\zb)\Gcmm(\xb,\yb)
    \nonumber 
    \\ \label{eq:SDEmmm}
    & \quad  + \Gmelon(\xb) \cdot \mathfrak f\hp c_{\meloncik|\meloncik}(\yb,\zb)
    \\
    \nonumber 
    &  \quad - \suml_{b_c}\frac{1}{E(x_c,b_c)} \big[
    \Gsmmm(\xb,\yb,\zb) -\Gsmmm(b_c\xb_{\hat c}, \yb,\zb)
    \big]
    \\
    \nonumber 
    & \quad  + \frac{1}{E(y_c,x_c)}
    \big[
    \Gsmc(\zb,\xb,\yb)
    -\Gsmc(\zb,y_c\xb_{\hat c},\yb)
    \big]
    \\
    &  \quad + \frac{1}{E(z_c,x_c)}
    \big[
    \Gsmc(\yb,\xb,\zb)
    -\Gsmc(\yb,y_c\xb_{\hat c},\zb)
    \big]
    \nonumber \\
    & \quad + \suml_{\sigma\in\Sym(3)}
    \sigma \cdot \mathfrak f\hp c_{\meloncik |\meloncik | \meloncik} (\xb,\yb,\zb)
     \bigg\}\,. \nonumber
      \end{align}
  We kept the graph notation in order to ease the reading of the graph movements. Equivalently,
           \begin{align}  
       \bigg(1+\frac{2\lambda}{E_{\xb}}\suml_{c=1}^3 \suml_{\mathbf q_{\hat c}}& G\hp2 (x_c,\mathbf q_{\hat c})\bigg)
      \times G\hp6 _{\mm | \mm | \mm }(\xb,\yb,\zb)    \nonumber \\
    & = \bigg(\frac{-2\lambda}{E_\xb}\bigg)
    \suml_{c=1}^3\bigg\{ 
    \mathfrak{f}\hp c_{\mm; x_c}(\yb)G\hp 4_{\mm|\mm}(\xb,\zb)
    +
    \mathfrak{f}\hp c_{\mm; x_c}(\zb)G\hp 4_{\mm|\mm}(\xb,\yb)
    \nonumber 
    \\ \label{eq:SDEmmmNotation}
    & \quad  + \Gmelon(\xb) \cdot \mathfrak f\hp c_{\mm|\mm}(\yb,\zb)
    \\
    \nonumber 
    &  \quad - \suml_{b_c}\frac{1}{E(x_c,b_c)} \big[
     G\hp6 _{\mm | \mm | \mm }(\xb,\yb,\zb) - G\hp6 _{\mm | \mm | \mm }(b_c\xb_{\hat c}, \yb,\zb)
    \big]
    \\
    \nonumber 
    & \quad  + \frac{1}{E(y_c,x_c)}
    \big[
    G\hp 6_{\mm | V_c}(\zb,\xb,\yb)
    -G\hp 6_{\mm | V_c}(\zb,y_c\xb_{\hat c},\yb)
    \big]
    \\
    &  \quad + \frac{1}{E(z_c,x_c)}
    \big[
    G\hp 6_{\mm | V_c}(\yb,\xb,\zb)
    -G\hp 6_{\mm | V_c}(\yb,y_c\xb_{\hat c},\zb)
    \big]
    \nonumber \\
    & \quad + \suml_{\sigma\in\Sym(3)}
    \sigma \cdot \mathfrak f\hp c_{\mm |\mm | \mm} (\xb,\yb,\zb)
     \bigg\} \,.\nonumber
      \end{align}

\section{Towards higher-dimensional Tutte equations}\label{sec:TutteSDE}

We compare the result with matrix models loop 
equations and, their equivalent, Tutte equations that count
discrete surfaces.

\subsection{Tutte equations and matrix models} \label{sec:TutteMatrix}
This material is based on the exposition 
by Eynard \cite[Ch. 1 \& 2]{CountingSurfaces}.
There, three facts are proven:

\begin{enumerate}
 \item The 
generating functions $\mathcal T_{l_1+1,l_2,\ldots,l_\kappa}$
 of connected \textit{maps}\footnote{A \textit{map} is a 
concept slightly more general than a gluing of a
collection of $n_\alpha$ $\alpha$-gons 
by their sides ($\alpha \geq 3$). Maps might have 
 certain number $\kappa$ of marked faces of perimeters $l_1,\ldots,l_\kappa  \geq 0$, . Their 
Euler characteristic is $\chi=$ $\#$vertices $-$ $\#$edges $+$ $\#$unmarked faces $=2-2g-\kappa$, being $g$ 
the \textit{genus} of the map.
The precise concept will not be needed here
and we refer to \cite[Sect. 1.1.2]{CountingSurfaces}
for the definition in terms of permutations.
} with $\kappa$ marked faces (boundaries)
of perimeters $l_1+1,\ldots,l_\kappa$ satisfy
Tutte equations:
\begin{subequations} \label{eq:Tutte}
\begin{align}  \label{eq:Tutte_0}
\mathcal T\hp{0}_{l+1} = {\sum_{j=3}^d \lambda_j \mathcal T\hp{0}_{l -j-1}} + 
\sum_{l_1+l_2 = l - 1  } \mathcal T\hp{0}_{l_1}  \mathcal T\hp{0}_{l_2}   
\end{align}
in the planar, single-boundary ($\kappa=1$) case, while for all $g\geq 1$, setting $K=\{l_2,\ldots,l_\kappa\}$,
\begin{align}
\mathcal T\hp{g}_{l_1+1,K}   \label{eq:Tutte_g}
&=\sum_{\alpha=3}^{d} \lambda_\alpha \mathcal T\hp g_{l_1+\alpha-1,K} +
 \sum_{m=2}^\kappa l_m \mathcal T_{ l_1+l_m-1,K\setminus\{l_m\} } \hp{g } \\ 
 & +
  \sum_{j=0}^{l_1-1}  \bigg[
  \mathcal T\hp{g-1}_{j,\,l_1-1-j,\,K} +
  \sum_{\substack{g_1+g_2=g \\ {J\subset K} }}
  \mathcal T\hp{g_1}_{j,J} \times \mathcal T\hp{g_2}_{l_1-1-j,K\setminus J}
   \bigg]\,. \nonumber 
\end{align}
\end{subequations}
In the last two equations the superindex $h$ in $\mathcal T\hp h_{\ldots}$
means the restriction to genus-$h$ maps. Also the formal variables $\lambda_\alpha$
count the number $n_\alpha$ of $\alpha$-agons in each map. That is, 
the $t^{v(\mathfrak m)}\lambda_3^{n_3(\mathfrak m)}\cdots \lambda_d^{n_d(\mathfrak m)}$-coefficient 
of $\mathcal T\hp h_{l_1,\ldots,l_\kappa}$ counts, modulo automorphisms,
how many genus-$h$ maps $\mathfrak m$ are there, having precisely
$n_\alpha(\mathfrak m)$ $\alpha$-agons  
with marked faces of lengths $l_m>0$  ($\alpha=3,\ldots,d$; $m=1,\ldots,\kappa$). To render this number finite, the variable $t$ counts
the number of vertices $v(\mathfrak m)$ of the map $\mathfrak{ m}$.
These $\mathcal T$-generating functions are not independent but related via
Tutte equations --- and in fact obey a rather universal 
relation known as Topological Recursion. 
\vspace{.2cm}
\item The matrix model 
\begin{equation}
 \label{eq:MM}
 Z=\int_{\mathrm{formal}} \dif M \ee^{-\frac{N}{t}\big[\frac{M^2}{2}-V(M)\big]}, \qquad V(x)=
\sum_{\alpha = 3} ^{d}\lambda_\alpha  \frac{x^\alpha}{\alpha} \,,
\end{equation}
satisfies Migdal's loop equations \cite{Migdal}
\begin{align} \nonumber 
&\sum_{j=0}^{l_1-1 } \langle \Tr (M^j) \Tr (M^{l_1-1-j})\rangle  
+\sum_{m=2}^{\kappa} l_m \langle \Tr(M^{l_1-1+l_j} \prod_{\substack{i=2 \\ i\neq j}} ^\kappa  \Tr(M^{l_i}) \rangle  
\\
=  & \frac{N}{t}\Big\langle \Tr (M^{l_1} \cdot [M-V'(M)] ) \times \Tr (M^{l_2})\cdots 
\Tr(M^{l_\kappa})\Big\rangle \,\,.
\label{eq:SDE_MM} 
\end{align}
These expressions are the SDE for the matrix-valued correlator 
\begin{equation} 
\label{eq:CorrelatorMM}
\langle M^{l_1} \Tr (M^{l_2}) \cdots \Tr (M^{l_\kappa}) \rangle \,, \end{equation}
defined for a function  $\Phi: M_N(\re)\to M_N(\re)$ by 
\[ \langle \Phi(M)\rangle = \frac{\int \dif M \Phi(M)\ee^{ -\frac{N}{t}\Tr(M^2/2 -V(M) )}}{\int \dif M \ee^{ -\frac{N}{t}\Tr(M^2/2 -V(M) )}}\,. \]

\item The crux of the matter is that Tutte Equations \eqref{eq:Tutte}
hold if and only if the SDE \eqref{eq:SDE_MM} for the 
matrix model \eqref{eq:MM} do. The bridge is the following.
For closed maps $\log Z$, the logarithm of the formal integral \eqref{eq:MM},
is well-known to yield the generating function 
of connected closed maps (cf. \cite{Brezin:1977sv}). 
The formal variables $\lambda_\alpha$ in both cases coincide:
if only maps consisting of, say, triangulations and quadrangulations are to be counted,
one sets a 
cubic and quartic interaction in the matrix model (in eq. \eqref{eq:MM} $\lambda_\alpha=\delta_3^\alpha \lambda_3+\delta^4_\alpha \lambda_4$).
For maps with $\kappa$ marked faces, if $x_1,\ldots,x_\kappa$ are formal variables and
one defines $W_\kappa$ by
\begin{equation} \label{eq:Neumann}
W_\kappa(x_1,\ldots,x_\kappa)= \bigg\langle \Tr \frac{1}{x_1-M}\cdots \Tr \frac{1}{x_\kappa-M} \bigg\rangle_{\mtr{connected}},
\end{equation}
in the sense of Neumann series,
then $\mtc T_{l_1\ldots l_\kappa}$ can be recovered by taking residues at $x_i\to\infty$ as follows:
\begin{equation} \quad \,\,\,\,
 \mathcal T_{l_1\ldots l_\kappa}= (-1)^\kappa\displaystyle 
     \mathop{\mathrm{Res}}\displaystyle_{x_1\to \infty } \ldots  \mathrm{Res}_{x_\kappa\to \infty } \big [x^{l_1}_1 \cdots x^{l_\kappa}_\kappa W_\kappa(x_1,\ldots,x_\kappa)\big]\,.
\label{eq:residues}
\end{equation}
Tutte equations \eqref{eq:Tutte} for all genera $g\in \Z_{\geq 0}$ emerge by taking
the small-$t$ expansion\footnote{To see the subtleties between small-$t$ expansion and an $1/N$-expansion 
we refer to \cite[Sect. 1.2.4]{CountingSurfaces}} 
$\mtc T_{l_1+1,\ldots l_\kappa}=\sum_{g} (N/t)^{2-2g-\kappa}\mtc T_{l_1+1,\ldots l_\kappa} \hp g$\,.
% Conversely, if $\mtc T_{l_1\ldots l_\kappa}$ are known,
% \[
% W_\kappa(x_1,\ldots,x_\kappa) = \sum_{l_1,\ldots,l_\kappa=0}^\infty \frac{\mathcal T_{l_1,\ldots,l_\kappa}}{x_1^{l_1},\ldots,x_\kappa^{l_\kappa}}
% \]
     \end{enumerate}

 \vspace{.4cm}

     \subsection{Parallel between Tutte equations and disconnected-$\partial$ SDE of TFT} \label{sec:Parallel}
      
     We contrast now elements appearing in the SDE of Theorem \ref{thm:SDEDisconnected}
     with Tutte equations, as well as the derivation of both sets.
     
     The first parallel, depicted in 
     Table \ref{tab:Boundaries}, concerns the role of the boundaries in each framework. 
       \begin{table}[h] 
    \includegraphics[width=1.0\textwidth]{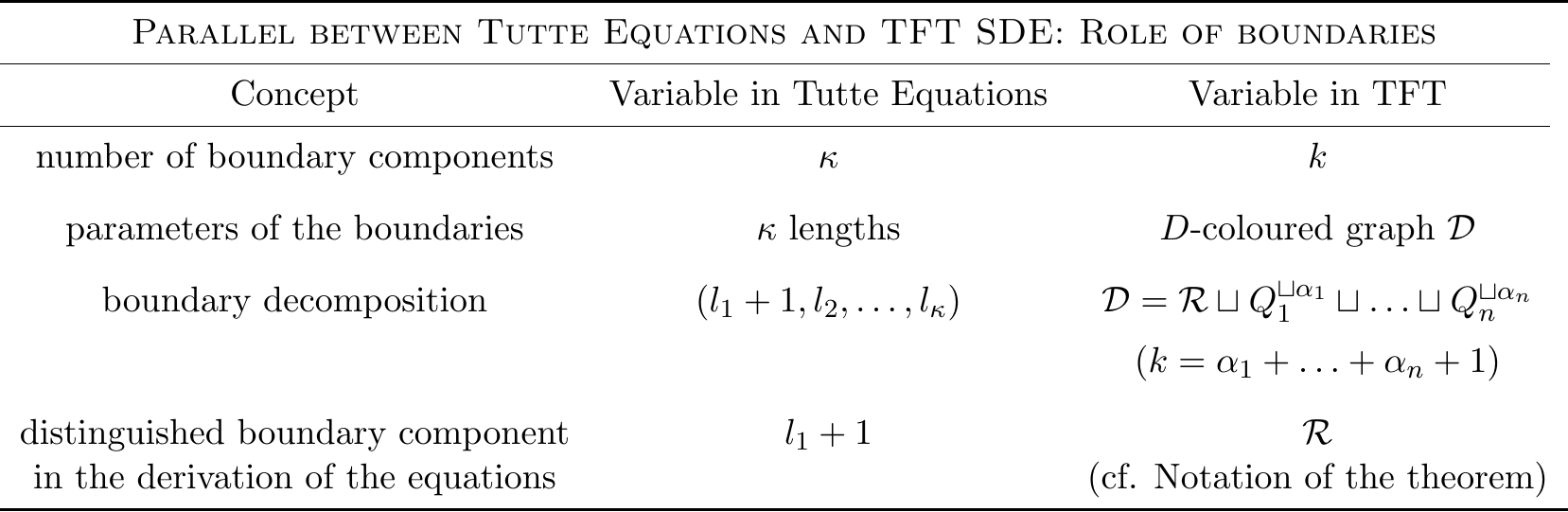} 
    \vspace{.05cm}
    \caption{Boundaries in both frameworks } \label{tab:Boundaries}
     \end{table}
     
     Moreover, we can regard each Tutte equation 
 as a set of operations in the input, namely the perimeters $l_1+1,\ldots,l_\kappa$
 of the marked faces, and distinguish the two cases $\kappa=1$ and $\kappa >1$. 
 The output for the \textit{connected}-$\partial$ Tutte equations \eqref{eq:Tutte_0}
 is illustrated in Table \ref{tab:SingleMM}; 
 similarly, for the disconnected-$\partial$ Tutte equations
 Table \ref{tab:MultiMM} depicts the types of operations 
 on the list of perimeters. \\
 
 A second similitude is the role played by a 
 distinguished boundary in each case. Tutte equations
 can be derived \cite[Ch. 1]{CountingSurfaces} by distinguishing a single boundary length,
 say 
 $l_1+1$, among the list of perimeters
 $l_1+1,\ldots,l_\kappa$, and seeing the possible effects\footnote{Incidentally, this amounts to four possible scenarios 
 that account for each one of the terms in eq.
     \eqref{eq:Tutte_g}:     
     \begin{itemize}
      \item[I.] if the marked edge separates an ordinary face from 
      a marked face: this accounts for the term $ \lambda_\alpha \mathcal T\hp g_{l_1+\alpha-1,K}$
      \item[II.] if the marked edge separates two marked marked faces  $ \to l_m \mathcal T_{ l_1+l_m-1,K\setminus\{l_m\} } \hp{g }$
      \item[III.] if the marked edge bounds twice the same face:
  \begin{itemize}
   \item[III.a] and the marked edge is not a bridge $  \to\mathcal T\hp{g_1}_{j,J} \times \mathcal T\hp{g_2}_{l_1-1-j,K\setminus J}$
   \item[III.b] and the marked edge is a bridge $  \to\mathcal T\hp{g-1}_{j,\,l_1-1-j,\,K}$. 
  \end{itemize}
     \end{itemize}} that the removal of a marked edge of this $(l_1+1)$-agon 
     has at the level of the $\mathcal T$-generating functions.
In our TFT SDE, the role of that marked face is 
taken by the boundary component denoted $\mathcal R$.
This undergoes the transformations explained in Table \ref{tab:MultiTM}.
  
This new operations are an improvement of description given in \cite{SDE}.
Put into the Tutte Equations perspective, the analogue of Table \ref{tab:SingleMM} are precisely those operations on a \textit{connected} $\partial$-graph of \cite{SDE}. 
The contribution of Theorem \ref{thm:SDEDisconnected} is
to give the full set of operations 
described and interpreted in Table \ref{tab:MultiTM}. \par 
One can compare term by term Tutte vs. Schwinger-Dyson equations.
The terminology refers to Tables \ref{tab:MultiMM} and \ref{tab:MultiTM}.
They do not match 1:1, due to the fact that in TFT the number
of operations on boundaries increases.\par
Finally, boundaries ---and therefore SDE--- of TFT are more 
complex than their matrix counterparts due to their
lack of symmetry. As pointed out in Section 
\ref{sec:main}, the choice of vertex
$\sb$ does matter in the sense that different choices
unrelated by non-trivial graph automorphisms 
of the distinguished boundary $\mathcal R$ 
yield different SDE.

%  \vspace{2cm}
   
     \begin{landscape}

     \thispagestyle{empty}
\begin{table}
   \includegraphics[width=1.2\textwidth]{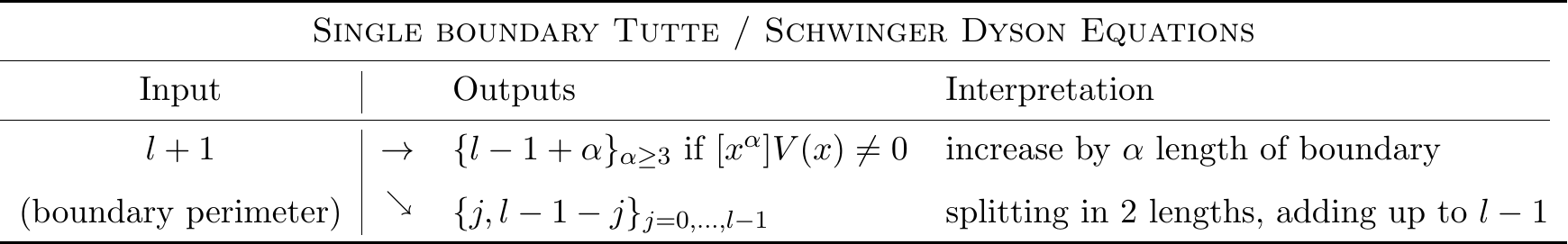}\vspace{.2cm}
    \captionof{table}{Operations on the  boundaries (marked face) for the connected boundary Tutte for random maps / Schwinger-Dyson equations for the matrix model 
    $Z=\int \dif M \ee^{-\frac{N}t(\frac{M^2}2-V(M))}$ with $V(x)=\sum_{\alpha = 3} ^{d}\lambda_\alpha  {x^\alpha}/{\alpha} $}
    \label{tab:SingleMM}
\end{table}

%  \hspace{.10cm}

\begin{table}[h]
    \includegraphics[width=1.5\textwidth]{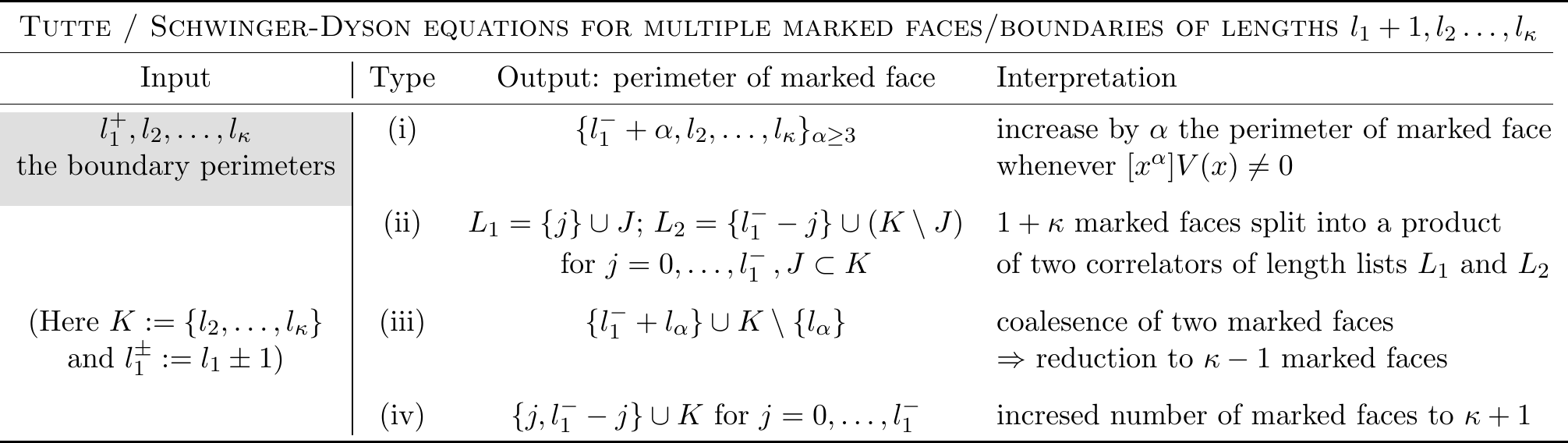}\vspace{.2cm}
    \captionof{table}{Operations on the lengths on the boundaries  of  disconnected boundary Tutte / Schwinger-Dyson equations }
    \label{tab:MultiMM}
    \end{table}
   \end{landscape}
    
\[\phantom \int\]
\vspace{1cm}
         \begin{table}[h]
    \includegraphics[width=.95\textwidth]{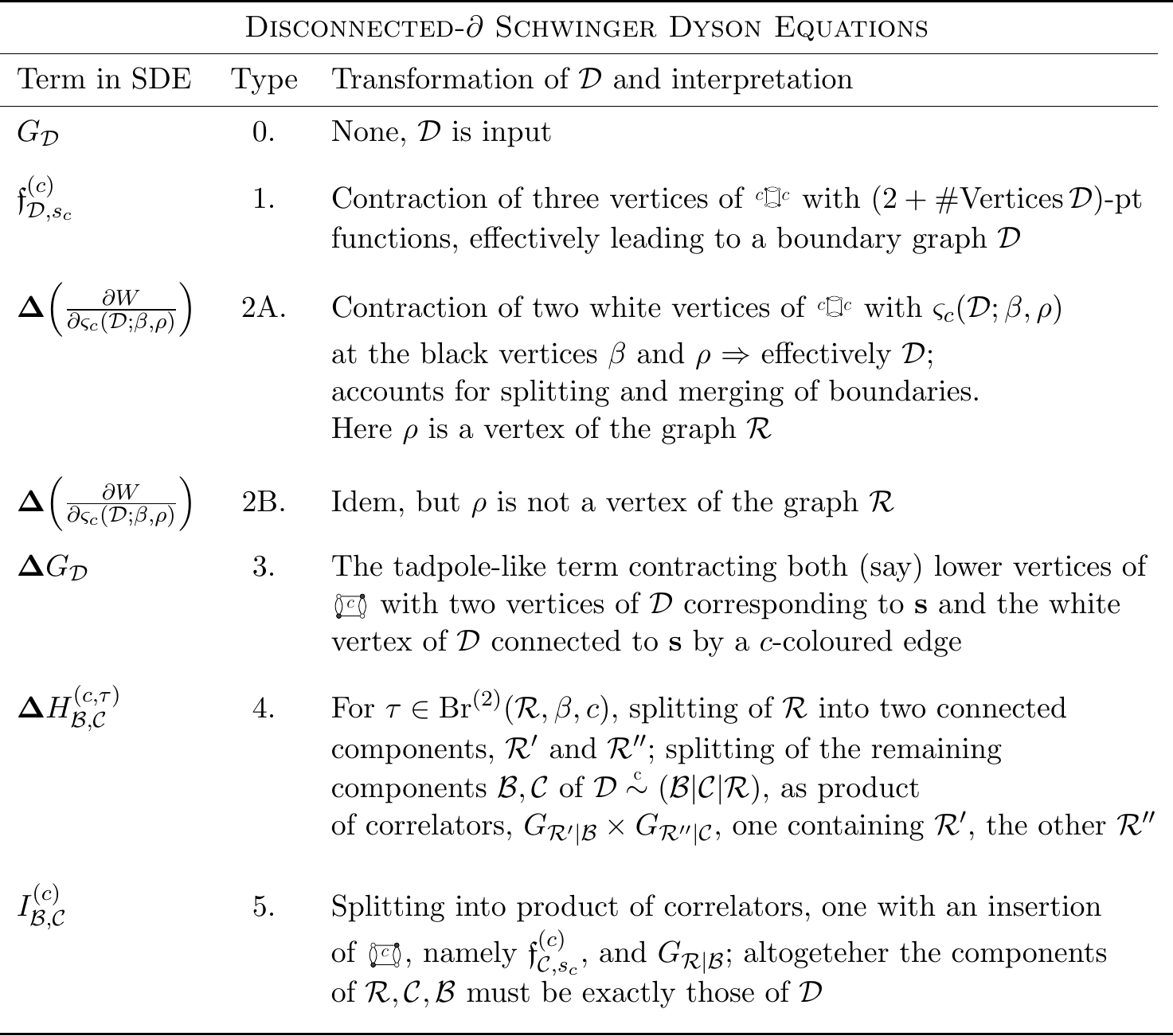} 
    \vspace{.2cm}
    \caption{Tensor model analogue of Table 
   \ref{tab:MultiMM}. The notation in the column `Term in SDE'
   refers each row in the Eq. \eqref{eq:SDEs}. In types 2 through 4,
   $\boldsymbol  \Delta L$ 
   refers to the line of Eq. \eqref{eq:SDEs} containing a term of the \textit{form}  $[L(\mathbf X)-L(\mathbf X|_{s_c\to \xi })]$ 
      that can be easily read off (together with its factors and sums over
      automorphisms and, as appropriate, over momenta)} 
   \label{tab:MultiTM}   
   \end{table}
%      \end{vplace}
    
\begin{landscape}
%   \[\phantom \int \]    
% \hspace{5cm}
\centering
\thispagestyle{empty}
\begin{table}
\includegraphics[width=1.54\textwidth]{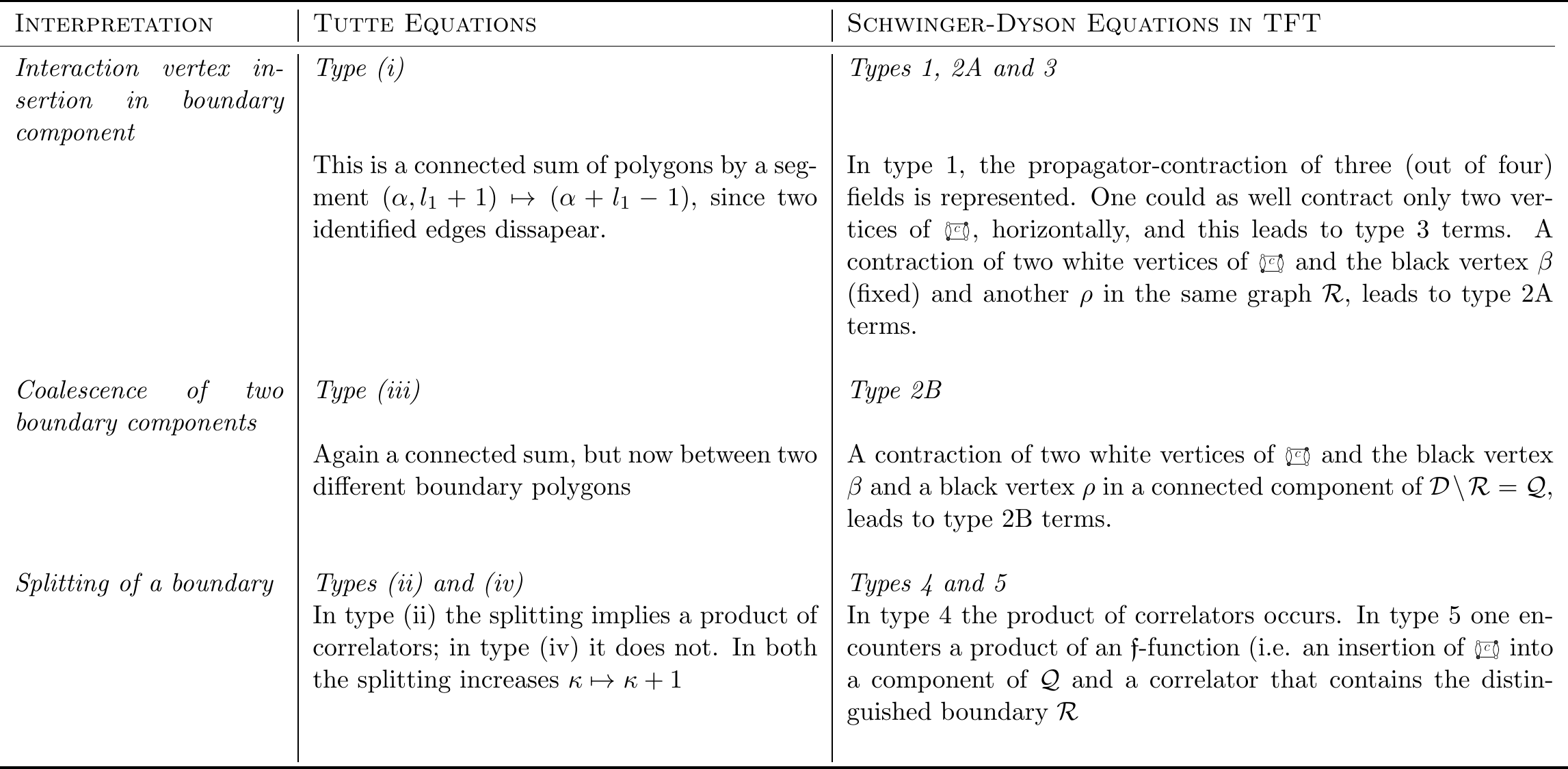}
\caption{Term-by-term comparison}
\end{table}

\end{landscape}
     
     \section{Conclusions and outlook} \label{sec:Outlook}  \thispagestyle{plain}
  
      We introduced the multi-variable graph calculus---a tool to prove 
      a general formula for the disconnected-$\partial$
     Schwinger-Dyson for the most general quartic-melonic TFT in 
     arbitrary rank (Thm \ref{thm:SDEDisconnected}). 
     In their description, a new set of 
     graph operations on an input graph has been exposed 
     in Table \ref{tab:MultiTM}. 
     This list is the tensorial equivalent of the matrix model operations on the boundary graphs of matrix models (cf. Table \ref{tab:MultiMM}).  \par 
     The well-known dictionary between the theory 
     of enumeration of random maps and matrix models 
     allows to pose two uses our SDE could 
     be useful for. The theories of 
     graph-encoded manifolds 
     \cite{pezzana,linsmulazzani,PaolaRita}
     bring tensor models into prominence as a theory of `random higher-dimensional maps'.
     For instance, gluings of octahedra \cite{Octahedra} are studied from the tensor model
     with (non-melonic) interaction vertex 
     \[\Cubo\]
     which is the dual graph-representation to an octahedron.
     The boundary-completeness of the quartic-melonic models \cite[Thm 1]{fullward} further studied 
     here shows a path to the theory of 
     higher-dimensional maps made of gluings of particular 
     triangulations of $D$-balls, in the rank-$D$ case. The natural candidate for higher-dimensional 
     Tutte equations is the set of relations 
    derived from the $1/N$-sectors  of the present SDE  
    that have the same boundary as well as value of Gur\u au degree. This 
    is based on the fact (recalled in Sect. \ref{sec:TutteSDE})
    that Tutte equations hide in the $N^{2-2g-\kappa}$-sectors
     of Migdal's  SDE for matrix models.
   
    The second guess is the existence of 
     a recursion allowing 
     to compute higher-point TFT-correlators 
     from some small number of lower-point ones, 
     analogous to the Topological Recursion (TR) \cite{RNARecTop,BorotNotes, EynardTopologicalRecursion,EynardOrantin,sulkowski} that
     satisfy matrix model correlators.
     What is not speculative is that, 
       assuming the dictionary of last paragraph, such tensorial TR 
       would require the disconnected-$\partial$ SDE, since 
       it is also a recursion in the number of boundaries 
       (presumably also in their Gur\u au's degree).
     The blobbed TR for quartic 
     \textit{tensor models} has already been obtained \cite{RecTopTM}---yet it would be 
     interesting to develop the purely tensorial\footnote{The Bonzom-Dartois TR bases
     on an initial Hubbard-Stratonovich transformation. Vertically cutting the pillow vertices,
     as the rank-$3$ $\vcv$, they ``map'' quartic melonic tensor 
     models to a suitable  intermediate field multi-matrix model, for which the authors develop a TR
     analogous to the introduced by Borot \cite{BorotBlobbed}.} cousin of the TR for \textit{tensor field theory}\footnote{
     The difference between tensor models and tensor field theory is here
     substantial. The former usually focuses on numerical observables $\langle \Tr_\B(\phi,\bar\phi) \rangle\in \C$ and the latter on functions (or distributions) $G_\B: \mathscr V(\B) \to \C$. In contrast, 
     the loop equations, Ward Identities \cite{rainbow} and SDE \cite{GurauSDE}
     are algebraic in the for tensor models, whereas for tensor field theory 
     `loop equations' \cite{fullward,SDE} are integro-differential, as shown also in \cite{cinco} explicitly.  
     }.

     \par 
      The large-$N$ limit of the disconnected-$\partial$ SDE should be analysed in order to access also 
      their physical significance. At leading order, their \textit{melonic approximation} \cite{us} is expected to yield closed equations.
      A significant progress in this direction has been undertaken in \cite{cinco}, whose techniques
      could \cite{Pascalie} be extended to the present disconnected-$\partial$ SDE in a next project.   
      \par
      Heading towards a quantum gravity perspective, 
      objects appearing in the Functional Renormalization Group  
    \cite{BenGeloun:2018ekd} ---or Ward-constrained flows \cite{Lahoche:2019vzy}--- are expected to be described in terms of graph-generated 
    functionals studied here. Together with the boundary $\!\!$-completeness 
     of the quartic-melonic models \cite[Thm 1]{fullward}, 
     this motivates us to study the geometric nature of the flow 
     from a simple quartic model.
     
      \par
      On the purely mathematical side, systems of graph-group actions can be
      extended to Lie groups actions and to calculi in  
      infinitely many graph-variables by using rigorous analytic tools.
      It would be also interesting to
      consider the coefficient functions $u_g$ directly in certain
      algebra of functions. One could dispense with the functions
      $\{\mathscr V(g) \to \C\}$ by using instead directly (non-commutative) algebras.  
   \par
   
   Finally, the (symmetric) monoidal structure on the 
   set of boundary graphs emerges in a natural way. This guides us towards the 
   language of    Topological Quantum Field Theories (TQFT) \cite{Atiyah}.
   Since these boundary graphs triangulate boundary states, an interesting
   program would 
   be to obtain discrete TQFT
   from matrix models and `TQFT with observables' \cite{oeckl} 
   from tensor models, 
   or enframe these in Oeckl's positive boundary formalism (\textit{op. cit.}),
   which also facilitates the gluing-boundary procedure that TQFT provides.
   In the tensor and matrix models case, the gluing of boundaries
   should be implemented as an operation $\wedge_\A$
   on two correlation functions sharing a boundary state $\A$,
   $G_{\A|\B|\cdots|\mtc C'}\wedge_\A G_{\A|\B'|\cdots|\mtc C'}$,
   which should be related to $G_{\B|\B'|\cdots |\mtc C|\mtc C'}$
   due to their geometrical interpretation. \\

     \subsection*{Acknowledgements}
%        \small
       The author thanks Raimar Wulkenhaar and Adrian Tanas\u a  for hospitality;
     and Romain Pascalie for helpful hints and carefully reading the draft 
     (any error is the author's responsibility). Thanks to the Faculty of Physics, Astronomy and Applied Computer Science,
      Jagiellonian University (Cracow, Poland), where part of this article was written, for hospitality. 
     The author acknowledges the Short-Term Scientific Mission program of the COST Action MP 1405 
     for this mobility opportunity.  \par      
      This research was funded by the Deutsche Forschungsgemeinschaft, 
     SFB 878 (Mathematical Institute of the University of M\"unster, Germany).
     Subsequently it was carried out at the Institute of Theoretical Physics, University of Warsaw and has been supported by the TEAM programme of the Foundation
for Polish Science co-financed by the European Union under the
European Regional Development Fund (POIR.04.04.00-00-5C55/17-00).

    \appendix
    \section{The first coefficients of the $Y$-term} \label{sec:App}
    For completeness, we give the first coefficients of the $Y$-term,
    keeping in mind the notation simplification (Table
    \ref{table:notation}).  The computation of these functions is
    presented in detail in \cite{fullward}. As before, 
    the set equality $\{a,b,c\}=\{1,2,3\}$ holds.

           \allowdisplaybreaks[1]
     
\begin{subequations}
\begin{align}
\mathfrak f_{\mathrm{m}; s_a}\hp{a}(\mathbf{x}) & =
G\hp{4}_{V_a} (\mathbf{x},s_a,x_b,x_c)  \\
&+ 
\sum\limits _{c \neq a }\sum\limits_{q_b\in\mathbb Z}G\hp{4}_{V_c}
(\mathbf{x};s_a,q_b,x_c) +
\sum\limits_{q_b,q_c} G_{\mathrm{m}|\mathrm{m}} \hp 4(\mathbf{x}; s_a,q_b,q_c),  \nonumber  \\ 
\mathfrak f_{ V_a; s_a}\hp{a}(\mathbf{x},\mathbf{y}) &  \nonumber
=\frac13 \Big( G_{Q_a}\hp 6(s_a,x_b,x_c,\mathbf{x},\mathbf{y})+
{\vphantom{\sum\limits_q}} \mbox{cyclic perm. in $(s_a,x_b,x_c), \xb $ and $ \yb$} \Big) 
\\&  + \frac13 \Big(
G_{K_{3,3}}\hp 6 (s_a,x_b,y_c;\mathbf{x},\mathbf{y} ) +
\mbox{cyclic perm.} \Big)\nonumber {\phantom{\sum\limits_q}}   \nonumber \\
&+ \sum\limits_{q_b}G_{F_{b;ac}}\hp 6 (\mathbf{x};\mathbf{y};s_a,q_b,y_c) + \displaystyle
\sum\limits_{q_c}G_{F_{c;ab}} \hp 6 (\mathbf{x};\mathbf{y};s_a,q_c,y_b)\\&
+\frac12 \sum\limits_{q_b,q_c}G_{\mathrm{m}| V_a}\hp 6 
(s_a,q_b,q_c;\mathbf{x};\mathbf{y}),  \nonumber   \\
\mathfrak f_{ V_b; s_a}\hp{a}(\mathbf{x},\mathbf{y}) & = 
\frac13\Big(\sum\limits_{q_b} G_{Q_b}\hp 6(s_a,q_b,y_c;\mathbf{x},\mathbf{y}) + \mbox{cyclic perm.} \Big)+ G_{F_{c;ab}} \hp 6 (s_a,y_b,x_c ;\mathbf{x};\mathbf{y}) \nonumber\\
&+ G_{F_{c;ab}} \hp 6 (\mathbf{x} ;s_a,x_b,x_c ;\mathbf{y})
+\sum\limits_{q_b}G_{F_{a;bc}} \hp 6 (\mathbf{x};\mathbf{y};s_a,q_b,y_c) 
\\ & 
-\displaystyle\frac12 \sum\limits_{q_b,q_c}G_{\mathrm{m}| V_b}\hp6 (s_a,q_b,q_c;  \mathbf{x};\mathbf{y}),  \nonumber \\
\mathfrak{f}_{\mathrm{m}|\mathrm{m}; s_a}^{(a)}
(\mathbf{x},\mathbf{y})  &
=  \Big( \sum\limits_{q_b,q_c}
G^{(6)}_{\mathrm{m}|\mathrm{m}|\mathrm{m}}\left(s_a,q_b,q_c,\mathbf{x},\mathbf{y}\right) + \text{cyclic perm.} \Big) \nonumber \\
&+ G^{(6)}_{\mathrm{m}|V_a}\left(\mathbf{x},s_a,y_b,y_c,\mathbf{y}\right) + \sum \limits_{q_c} G^{(6)}_{\mathrm{m}|V_b}\left(\mathbf{x},s_a,y_b,q_c,\mathbf{y}\right) 
\\ & 
+ \sum \limits_{q_b} G^{(6)}_{\mathrm{m}|V_c}\left(\mathbf{x},s_a,q_b,y_c,\mathbf{y}\right)
+ \sum \limits_{q_c} G^{(6)}_{\mathrm{m}|V_b}\left(\mathbf{x},\mathbf{y},s_a,y_b,q_c\right)
\nonumber \\
& + \sum \limits_{q_b} G^{(6)}_{\mathrm{m}|V_c}\left(\mathbf{x},\mathbf{y},s_a,q_b,y_c\right) + G^{(6)}_{\mathrm{m}|V_a}\left(\mathbf{x},\mathbf{y},s_a,y_b,y_c\right) \nonumber \\ 
& + G^{(6)}_{F_{a;bc}}\left(\mathbf{x},s_a,x_b,y_c,\mathbf{y}\right) .   \nonumber
\end{align}
     \end{subequations}
       
% \small 
    
%     \bibliographystyle{alpha}

% ------------------------------------------------------------------------
\end{document}